\documentclass[11pt,a4paper]{article}
\usepackage{amsmath}
\usepackage{enumerate}
\usepackage{mathtools}
\usepackage{amsthm}
\usepackage{amssymb}
\usepackage{mathtools}
\usepackage{dsfont}
\usepackage{bbm}
\usepackage[margin=1in]{geometry}
\usepackage{caption}
\usepackage{subcaption}
\usepackage{times}
\usepackage{bm}
\usepackage{natbib}
\usepackage{graphicx}
\usepackage[ruled]{algorithm2e}
\usepackage{stmaryrd}
\usepackage{arydshln}
\usepackage{multirow} 
\usepackage{subcaption}
\usepackage[title]{appendix}

\DeclareMathOperator*{\argmax}{arg\,max}
\DeclareMathOperator*{\argmin}{arg\,min}
\def\T{{ \mathrm{\scriptscriptstyle T} }}
\newcommand{\iid}{\overset{\mathrm{iid}}{\sim}}

\usepackage[]{algorithm2e}
\usepackage{tikz}
\usetikzlibrary{decorations.markings}
\usetikzlibrary{arrows}
\usepackage{bm}

\newtheorem{assumption}{Assumption}
\newtheorem{theorem}{Theorem}
\newtheorem{lemma}{Lemma}
\newtheorem{definition}{Definition}
\newtheorem{corollary}{Corollary}
\newtheorem{proposition}{Proposition}
\newtheorem{condition}{Condition}

\numberwithin{equation}{section}

\newtheoremstyle{example}{}{}{}{}{\bfseries}{\smallskip}{\newline}{}
\theoremstyle{example}
\newtheorem{example}{Example}
\providecommand{\keywords}[1]{{Keywords: \hspace{1mm}} #1}

\usepackage{xcolor}

\DeclareRobustCommand\dashdotted{\tikz[baseline=-0.5ex]\draw[thick, dash dot] (0.06,0)--(0.5,0);}

\DeclareRobustCommand\dotted{\tikz[baseline=-0.5ex]\draw[thick,dotted] (0.06,0)--(0.5,0);}
\DeclareRobustCommand\dottedmid{\tikz[baseline=-0.5ex]\draw[thick,dotted,opacity = 0.75] (0.06,0)--(0.5,0);}

\DeclareRobustCommand\dashed{\tikz[baseline=-0.5ex]\draw[thick,dashed] (0,0)--(0.52,0);}
\DeclareRobustCommand\dashedmid{\tikz[baseline=-0.5ex]\draw[thick,dashed,opacity = 0.5] (0,0)--(0.52,0);}
\DeclareRobustCommand\dashedlow{\tikz[baseline=-0.5ex]\draw[thick,dashed,opacity = 0.3] (0,0)--(0.52,0);}

\DeclareRobustCommand\full {\tikz[baseline=-0.5ex]\draw[thick] (0,0)--(0.5,0);}
\DeclareRobustCommand\fullmid  {\tikz[baseline=-0.5ex]\draw[thick,opacity = 0.5] (0,0)--(0.5,0);}
\DeclareRobustCommand\fulllow  {\tikz[baseline=-0.5ex]\draw[thick, opacity = 0.25] (0,0)--(0.5,0);}

\DeclareRobustCommand\dot {\tikz[baseline=-0.55ex]\fill[] circle[radius=1.8pt];}
\DeclareRobustCommand\dotmid {\tikz[baseline=-0.52ex]\fill[opacity = 0.4] circle[radius=1.8pt];}

\DeclareRobustCommand \sqrlow {\tikz[]{\filldraw[opacity = 0.25] (0,0)
rectangle (0.2cm,0.2cm);}}

\DeclareRobustCommand \sqrmid {\tikz[]{\filldraw[opacity = 0.4] (0,0)
rectangle (0.2cm,0.2cm);}}

\DeclareRobustCommand \sqr {\tikz[]{\filldraw[opacity = 0.7] (0,0)
rectangle (0.2cm,0.2cm);}}

\usetikzlibrary{shapes.misc}
\tikzset{cross/.style={cross out, draw=black, minimum size=2*(#1-\pgflinewidth), inner sep=0pt, outer sep=0pt},
cross/.default={2.2pt}}
\DeclareRobustCommand \crossmid {\tikz[baseline = -0.5ex] \draw[opacity =0.7] (0,0) node[cross] {};}
\usepackage[breaklinks=true,hyperfootnotes=false]{hyperref}

\mathtoolsset{showonlyrefs=true}

\begin{document}

\title{Martingale posterior distributions}

\author{Edwin Fong$\hspace{0.3mm}^{1,2}$,  Chris Holmes$\hspace{0.3mm}^{1,2}$ \& Stephen G. Walker$\hspace{0.3mm}^3$  }
\date{}
\maketitle
\begin{abstract}
The prior distribution on parameters of a sampling distribution is the usual starting point for Bayesian uncertainty quantification. In this paper, we present a different perspective which focuses on missing observations as the source of statistical uncertainty, with the parameter of interest being known precisely given the entire population. We argue that the foundation of Bayesian inference is to assign a distribution on missing observations conditional on what has been observed. In the conditionally i.i.d. setting with an observed sample of size $n$, the Bayesian would thus assign a predictive distribution on the missing $Y_{n+1:\infty}$ conditional on $Y_{1:n}$, which then induces a distribution on the parameter. 
Demonstrating an application of martingales, Doob shows that choosing the Bayesian predictive distribution returns the conventional posterior as the distribution of the parameter. Taking this as our cue, we relax the predictive machine, avoiding the need for the predictive to be derived solely from the usual prior to posterior to predictive density formula. We introduce the \textit{martingale posterior distribution}, which returns Bayesian uncertainty directly on any statistic of interest without the need for the likelihood and prior, and this distribution can be sampled through a computational scheme we name \textit{predictive resampling}. To that end, we introduce new predictive methodologies for multivariate density estimation, regression and classification that build upon recent work on bivariate copulas.\\

\noindent \keywords{Bayesian uncertainty; Copula; Martingale; Predictive inference \vspace{-5mm}}
\end{abstract}
{\let\thefootnote\relax\footnote{{$^1$The Alan Turing Institute}}}
{\let\thefootnote\relax\footnote{{$^2$Department of Statistics, University of Oxford }}}
{\let\thefootnote\relax\footnote{{$^3$Department of Statistics and Data Sciences, University of Texas at Austin}}}

\section{Introduction }\label{sec:intro}

{ Statistical uncertainty in a parameter of interest arises due to missing observations. If a complete population is observed, then the parameter of interest can be assumed to be known precisely. In this paper, we argue that the Bayesian accounts for this uncertainty by constructing a distribution on the missing observations conditional on what has been observed. This in turn induces a distribution
on the parameter given the observed data, which we will see is the posterior distribution. In this work, we will describe and generalize this framework in detail for the case where the observations are independent and identically distributed (i.i.d.), and we will also briefly consider  other data structures.

In the conditionally i.i.d. case, given $Y_{1:n}\iid F_0$ where $F_0$ is the unknown true sampling distribution, the missing observations are the remaining $Y_{n+1:\infty}$, and as such we focus our  modelling efforts directly on the predictive density
\begin{equation} \label{eq:joint_predictive}
p(y_{n+1:\infty} \mid  y_{1:n}). 
\end{equation}
Here, the construction of the predictive density is for parameter inference, and not for forecasting future observations as is more usual.} For inference, we assume that the object of interest is fully defined once all the observations have been viewed, which we write as $\theta_\infty = \theta(Y_{1:\infty})$.  It is clear then that \eqref{eq:joint_predictive} induces a distribution on $\theta_{\infty}$,  and we call this scheme of imputing $Y_{n+1:\infty}$ and computing $\theta_\infty$ as \textit{predictive resampling}. A key observation is that $Y_{1:\infty}$ will always contain the observed $Y_{1:n} = y_{1:n}$ as the predictive Bayesian considers the observed sample to be fixed, in contrast to the frequentist consideration of other possible values of $Y_{1:n}$.

{For conditionally i.i.d. observations}, the traditional Bayesian approach is to elicit a prior density $\pi(\theta)$ and likelihood function $f_{\theta}(y)$, derive the posterior $\pi(\theta \mid y_{1:n})$, then compute the predictive  density through
\begin{equation}\label{eq:posterior_predictive}
p(y \mid y_{1:n}) = \int f_\theta(y)\, \pi(\theta \mid y_{1:n})\,d\theta. 
\end{equation} 
{ A concise summary of our approach is the following: while \cite{DeFinetti1937} provided a representation of Bayesian inference which relies on exchangeability and the prior distribution, \cite{Doob1949} provided a framework which relies solely, in the i.i.d. case, on the predictive distribution. We will see that Doob's framework is more flexible and the mathematical requirement amounts to the construction of a martingale. It is this flexibility provided by Doob's framework which we exploit in this paper.} In fact, through Doob's theorem,  we will see that predictive resampling as described above is identical to posterior sampling when using \eqref{eq:posterior_predictive} as the predictive and $\theta$ indexes the sampling density, in which case $\theta_\infty \sim \pi(\theta \mid y_{1:n})$. Denoting $p(y)$ as the prior predictive, this connection is illustrated below for the traditional Bayesian case:
\begin{eqnarray}
f_{\theta}(y), \pi(\theta) \xrightarrow[\rm{Bayes' \, rule }]{} \, \, \, \pi(\theta \mid y_{1:n}) & \xrightarrow[\int  f_\theta(y) \, \pi(\theta | y_{1:n}) \,d\theta]{\rm{~ ~ posterior ~  predictive ~~ }} & p(y \mid y_{1:n}) \nonumber \\
~ & ~ & ~ \nonumber \\ 
{\pi}(\theta \mid y_{1:n}) & \xleftarrow[Y_{n+1:\infty} \, \sim \, {p}(\cdot  \mid y_{1:n})]{\rm{~ ~ Doob's  ~ theorem ~~ }} & {p}(y \mid y_{1:n}) \, \, \, \xleftarrow[\rm{predictive \, update}]{} {p}(y)  \nonumber 
\end{eqnarray}
 
 However, the traditional Bayesian focus on the prior on $\theta$ makes no appeal to the underlying cause of the uncertainty, that is the unobserved part of the study population $Y_{n+1:\infty}$. Furthermore, the traditional prior to posterior computation is becoming increasingly strained as model complexity and data sizes grow.  In our work, we advocate the predictive resampling strategy - given $y_{1:n}$, our starting point is directly the predictive model \eqref{eq:joint_predictive} and the target statistic of interest $\theta_{\infty}$, noting now that $\theta_\infty$ is no longer restricted to indexing the sampling density. We relax de Finetti's assumption of exchangeability, but we must now take care to construct \eqref{eq:joint_predictive} so that $\theta_N$ is indeed convergent to some $\theta_\infty$, where $\theta_N = \theta(Y_{1:N})$ can be viewed as an estimator. We highlight here that we use $n$ and $N$ for the size of the observed dataset and the imputed population respectively.  In the spirit of Doob, we rely heavily on martingales, which also aid in ensuring that expectations of limits coincide with fixed quantities seen at the sample of size $n$. This can be regarded as a predictive coherency condition, and we designate the distribution of $\theta_\infty$ as the \textit{martingale posterior}. Our choice of \eqref{eq:joint_predictive} will be density estimators based on recent ideas in the literature, specifically the \textit{conditionally identically distributed} (c.i.d.) sequence of \citet{Berti2004} and bivariate copula update of \citet{Hahn2018}.

We now discuss why one would want to go through the route of obtaining the martingale posterior via the induced distribution of $\theta_\infty$ from \eqref{eq:joint_predictive} rather than the traditional likelihood-prior construction. Firstly, predictive models are probabilistic statements on observables, which removes the need to elicit subjective probability distributions on parameters which may have no real-world interpretations and only  index the sampling density.
Secondly, the martingale posterior establishes a direct connection between prediction and statistical inference, opening up the possibility of using modern probabilistic predictive methods for inference \citep{breiman2001}, and transparently acknowledges the source of statistical uncertainty as the missing $Y_{n+1:\infty}$. Thirdly, working directly with predictive distributions is highly practical.  For an elicited 1-step ahead predictive, we can predictively resample by carrying out the recursive update
$$
 \{p(y\mid y_{1:N-1}),y_{N}\} \mapsto p(y\mid y_{1:N})
$$
to sample $Y_{n+1:N}$ for a large enough $N$ such that $\theta_N$ has effectively converged to a sample from the martingale posterior, or $N$ matches a known finite study population size. 
In complex scenarios such as multivariate density estimation and regression, we introduce new copula-based methodologies where our computations remain  exact, GPU-friendly and parallelizable, returning us Bayesian uncertainty without any reliance on Markov chain Monte Carlo (MCMC). Finally, a predictive approach more clearly delineates the core similarities and differences between Bayesian and frequentist uncertainty.

{We will focus on the conditionally i.i.d. data setting in this work, which corresponds to exchangeable traditional Bayesian models. In this setting, the martingale posterior can indeed be regarded as a generalization of the traditional Bayesian model, as the class of c.i.d. models is more general and contains the class of exchangeable models which we will see in Section \ref{sec:pred_coherence}. In more complex data structures beyond i.i.d. data, such as those encountered in hierarchical modelling or time series, our framework would still apply. In this case, the missing observations we require may no longer be $Y_{n+1:\infty}$, and model elicitation would no longer only involve a sequence of predictive distributions. For example, a simple hierarchical setting is the observation process $Y_i \sim p(y_i | \theta_i )$, where $\theta_i$ is itself drawn from an unknown $G_0$ and we may be interested in some functional  $\gamma(G_0)$.  Here, we only observe $Y_{1:n} = y_{1:n}$, so the missing observations of interest are now the unobserved random effects $\theta_{1:\infty}$. We can thus seek to impute $\theta_{1:n} \sim p(\theta_{1:n} \mid y_{1:n})$ from the data, followed by the missing remainder $\theta_{n+1:\infty} \sim p(\theta_{n+1:\infty} \mid \theta_{1:n})$. Computing $\gamma(\theta_{1:\infty})$ would then return us a posterior sample. For the remainder of the paper, we will focus only on the i.i.d. case and leave the details of non-i.i.d. settings for future work.}

In Section \ref{sec:PI}, we formally investigate the connection between predictive and posterior inference, and introduce a predictive framework for inference and the resulting martingale posterior. We then utilize the bootstrap as a canonical example to distinctly compare Bayesian and frequentist uncertainty. We postpone discussion of related work until Section \ref{sec:related} in order to provide context beforehand.  In Section \ref{sec:PR}, we discuss predictive coherence conditions for martingale posteriors, utilizing c.i.d. sequences.  In Section \ref{sec:copula}, we revisit the bivariate copula methodology of \citet{Hahn2018} for univariate density estimation, and extend it to obtain the martingale posterior. We then generalize this copula-based method to multivariate density estimation, regression and classification. Section \ref{sec:illustrations} then provides a thorough demonstration of the above methods through examples. In Section \ref{sec:theory}, we discuss some theoretical properties of the martingale posterior with the copula-based methodology. Finally, we discuss our results in Section \ref{sec:discussion}.

\section{A predictive framework for inference}
\label{sec:PI}
\subsection{Doob's theorem and Bayesian uncertainty}\label{sec:foundations}

{ Uncertainty quantification lies at the core of statistical inference, and Bayesian inference is one framework for handling uncertainty in a formal manner. The Bayesian begins with the random variables $(\Theta, Y_1,Y_2,\ldots)$, where  $(Y_1,Y_2,\ldots)$ are the observables of interest, and $\Theta$ is the parameter which indexes the sampling density $f_\Theta(y)$. We assume throughout that the appropriate densities exist. For conditionally i.i.d. data, the Bayesian elicits a joint probability model for the observables and parameter with joint density
\begin{equation}\label{eq:joint_param}
p(\theta,y_{1:N}) =\pi(\theta)\, \prod_{i=1}^N f_\theta(y_i)
\end{equation}
for each $N$. Here, the density $\pi(\theta)$ represents prior knowledge about the parameter which generates the observations, and under a Subjectivist point of view,
$\Pi(A)=\int_A \pi(\theta)\,d\theta$
represents the subjective probability that the generating parameter value $\Theta$ lies in the set $A$. Marginalizing out $\Theta$ gives the joint density of the observables, 
\begin{equation}\label{eq:joint}
p(y_{1:N}) = \int\prod_{i=1}^N f_\theta(y_i)\,  d\Pi(\theta).
\end{equation}}

De Finetti however argued that the direct likelihood--prior interpretation of the Bayesian model was insufficient, as $\Theta$ is of a ``metaphysical" nature and probability statements should only be on observables \citep{Bernardo2009}. This then motivated the notion of exchangeability, where the joint probability $P$ of the observables $Y_{1:N}=(Y_1,\ldots,Y_N)$ is invariant to the ordering of $Y_i$ for all $N$. { Through de Finetti's representation theorem \citep{DeFinetti1937} and extensions thereof (e.g. \citet{Hewitt1955}), the assumption of exchangeability induces the likelihood-prior form of the joint density in \eqref{eq:joint} (where $\Pi$ may not have a density), which motivates such a specification of the Bayesian model.} The representation theorem however is only part of the story. As alluded to in the Section \ref{sec:intro}, the source of statistical uncertainty is the lack of the infinite dataset $Y_{n+1:\infty}$ with which we could pin down any quantity of interest precisely. Bayesian uncertainty through the lens of the prior is still opaque in this regard, even with the representation theorem. The key to understanding the source of uncertainty lies in the {predictive imputation} of observables, for which we require a result from Doob.

{ Let us assume that data has yet to be observed, so the missing observations are $Y_{1:\infty}$. Following the discussion in Section \ref{sec:intro}, one can regard \eqref{eq:joint} as the predictive density on the missing population, and can estimate the parameter indexing the sampling density as a function of the imputed $Y_{1:N}$. An appropriate and intuitive point estimate for the Bayesian is the posterior mean, which we write as
$$\bar{\theta}_N = E\left[\Theta \mid Y_{1:N} \right].$$ A secondary result of \cite{Doob1949} confirms that the prior uncertainty in $\Theta$ is equivalent to the predictive uncertainty in $Y_{1:\infty}$.
\begin{theorem}[\cite{Doob1949}]\label{Th:Doob_consistency}
Assume $\Theta$ is in a linear space with $E\left[|\Theta | \right] < \infty$, and $(\Theta,Y_1,Y_2,\ldots)$ is distributed according to \eqref{eq:joint_param}, so $\Theta \sim \Pi$. Under measurability and identifiability conditions on $f_\theta$, we have
\begin{equation}
\bar{\theta}_N \to \Theta \quad\mbox{a.s.}
\end{equation}
\end{theorem}
For the above result, the key is to rely on $\bar{\theta}_N$ being a martingale, that is
$$
E\left[\bar{\theta}_N \mid Y_{1:N-1}\right] = \bar{\theta}_{N-1} 
$$
almost surely. Doob's martingale convergence theorem then ensures that $\bar{\theta}_N$ converges to a limit almost surely. For $\Theta$ in more general metric spaces, consistency results with general notions of posterior expectations are provided in \citet[Theorem 6.8]{Ghosal2017}. As an aside, we highlight that \cite{Doob1949} provides a more general result: the Bayesian posterior distribution converges weakly to the Dirac measure $\delta_{\Theta}$ almost surely for $\Pi$-almost every $\Theta$ as $N \to \infty$. The technical details of a more general version of this result can be found in \citet[Theorem 6.9]{Ghosal2017}.  In the Bayesian nonparametric case where $\Theta$ is a probability density function, we have a nonparametric extension of the above results \citep{Lijoi2004}. }

Returning to the task at hand, we can summarize the above by considering two distinct methods of sampling $\Theta$ from the prior $\Pi$ {before} seeing any data. The first is to draw $\Theta \sim \Pi$ directly, which is the opaque view of the inherently random parameter that we are trying to shed light on. The second, which inspires the remainder of our paper, begins with sequentially imputing the unseen observables $Y_1,Y_2,Y_3\ldots$ from the sequence of predictive densities
\begin{equation}
\begin{aligned}
Y_1 &\sim p(\cdot),\quad Y_2 &\sim p(\cdot \mid y_1), \quad Y_3 \sim p(\cdot \mid y_2,y_1), \quad \ldots
\end{aligned}
\end{equation}
until we have the complete information $Y_{1:\infty}$ in the limit. Given this random infinite dataset, the limiting point estimate $\bar{\theta}_\infty$, that is the posterior mean computed on the entire dataset, is in fact distributed according to $\Pi$. This equivalence highlights the fact that {\em{a priori}} uncertainty in $\Theta$ and uncertainty in $Y_{1:\infty}$ are one and the same, and the function $\bar{\theta}$  provides a means to precisely recover our quantity of interest when all information is made available to us.

Of course, such an interpretation is equally valid  {\em{a posteriori}}, that is after we have observed $Y_{1:n}=y_{1:n}$. Here, sampling $\Theta \sim \Pi(\cdot\mid y_{1:n})$ is equivalent to sampling $Y_{n+1:\infty}$ conditional on $y_{1:n}$ and computing $\bar{\theta}_\infty$ as if we have observed the infinite dataset, noting that $Y_{1:n} = y_{1:n}$ is now fixed. This can be seen by simply substituting the prior $\pi$ in \eqref{eq:joint_param}, \eqref{eq:joint} and Theorem \ref{Th:Doob_consistency} with the posterior $\pi(\cdot \mid y_{1:n})$. In conclusion, Doob's result highlights that the Bayesian seeks to simulate what is needed to pin down the parameter but is missing from reality, that is $Y_{n+1:\infty}$ in the i.i.d. case,  and we find this to be a compelling justification for the Bayesian approach.

We now conclude this section with a concrete demonstration of the equivalence between posterior sampling and the forward sampling of $Y_{n+1:\infty}$ through a simple normal model with unknown mean based on an example from \cite{Hahn2015}.
 \begin{example}\label{ex:bayes_normal}
 Let $f_\theta(y) = \mathcal{N}(y \mid \theta,1)$, with $\pi(\theta) = \mathcal{N}(\theta \mid 0,1)$. Given an observed dataset	 $y_{1:n}$, the tractable posterior density takes on the form $\pi(\theta \mid y_{1:n}) = \mathcal{N}(\theta \mid \bar{\theta}_n,\bar{\sigma}^2_n)$ where
 \begin{equation*}
 \bar{\theta}_n=  \frac{\sum_{i=1}^n y_i}{n+1},\quad  \bar{\sigma}^2_n =  \frac{1}{n+1}\cdot
 \end{equation*}
 The posterior  predictive density then takes on the form $p(y \mid y_{1:n}) = \mathcal{N}(y \mid \bar{\theta}_n,1 + \bar{\sigma}_n^2)$. For observed data, we generated $y_{1:n} \iid f_\theta(y)$ for $n = 10$ with $\theta = 2$, giving $\bar{\theta}_n = 1.84$. 
 
 We can plot the independent sample paths for the posterior mean, $\bar{\theta}_{n+1:N}$, as we recursively forward sample $Y_{n+1:N}$, where $N = n+1000$ in this example. In Figure \ref{fig:normal_means}, we see that the sample paths of $\bar{\theta}_{n+i}$ each converge to a random ${\Theta}$ as $i$ increases, with the density of $\bar{\theta}_N$ very close to the analytic posterior. From Doob's consistency theorem, we know this is exact for $N \to \infty$.
 \begin{figure}[!h]
\centering
 \includegraphics[width=0.95\textwidth]{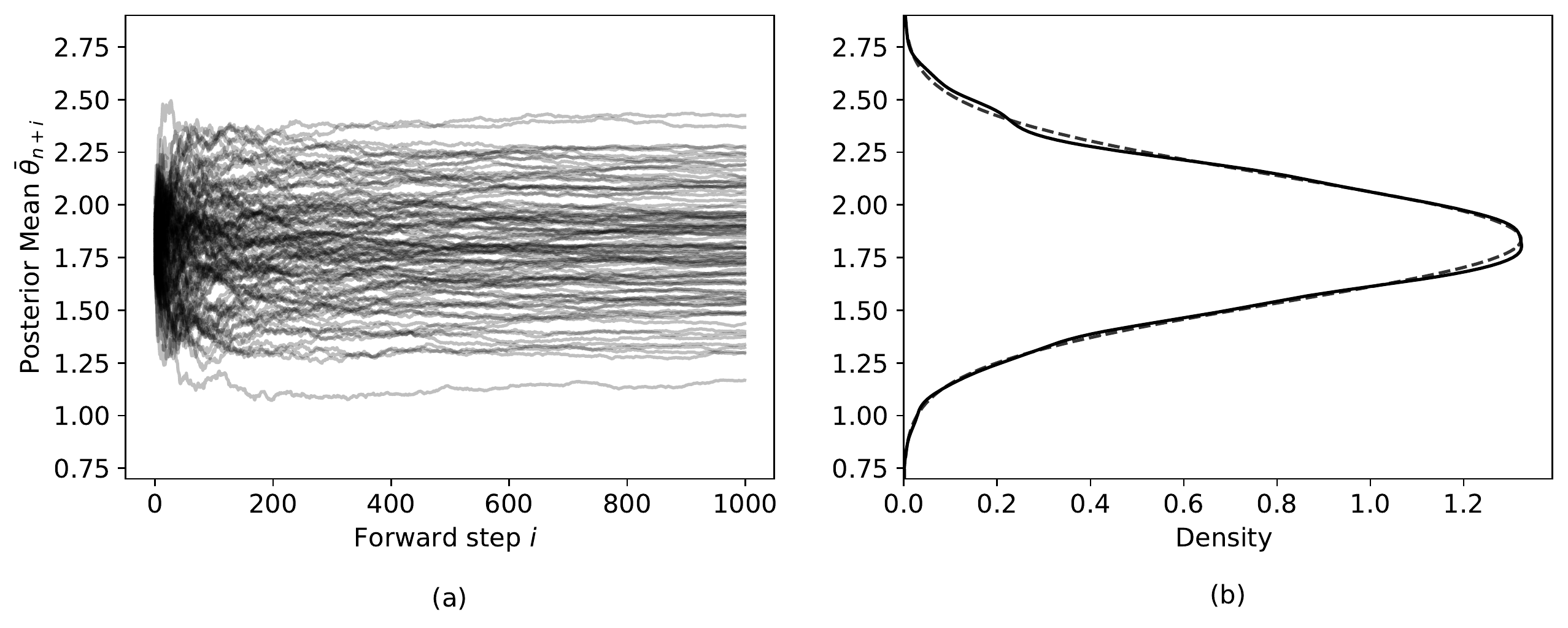}
\caption{(a) Sample paths of $\bar{\theta}_{n+i}$ through forward sampling; (b) Kernel density estimate of $\bar{\theta}_{N}$ samples (\full) and analytical posterior density $\pi(\theta \mid y_{1:n})$ (\dashed)} 
\label{fig:normal_means}
\end{figure}
 \end{example}

\subsection{The methodological approach}\label{sec:pred_inf}

Through Doob's result in Theorem \ref{Th:Doob_consistency}, we have demonstrated the predictive view of Bayesian inference as a means to understand how the posterior uncertainty in $\Theta$ arises from the missing information $Y_{n+1:\infty}$.  The predictive view of Bayesian inference partitions posterior sampling into two distinct tasks. The first is the simulation of $Y_{n+1:\infty}$ through the sequence of 1-step ahead predictive distributions to assess the uncertainty that arises from the missing observables. The second is the recovery of the parameter of interest $\Theta$ from the simulated complete information, which is facilitated by the limiting posterior mean point estimate $\bar{\theta}_\infty$. The uncertainty in $\Theta$ then flows from the uncertainty in $Y_{n+1:\infty}$. Inspired by this, we will now demonstrate the practical importance of this interpretation by introducing a predictive framework for inference built exactly on these two tasks. This framework eliminates the need for the usual likelihood--prior construction of the Bayesian model, and as such generalizes the traditional Bayesian posterior to the martingale posterior.

 \subsubsection{Sampling the missing data}
For the predictive Bayesian, the role of the posterior $\pi( \theta \mid y_{1:n})$ is to aid in the updating of the predictive density, $p(\cdot \mid y_{1:N-1}) \mapsto p(\cdot\mid y_{1:N})$ after observing $Y_{N}$, and the likelihood and prior can be viewed as merely intermediate tools to construct the sequence of predictives \citep{Roberts1965}.  
To obviate the need of a likelihood-prior specification,  our proposal is to specify the sequence of 1-step ahead predictive densities $\{p(\cdot\mid y_{1:N})\}_{N\geq n}$ directly, which implies a joint density through the  factorization
\begin{equation}\label{eq:chainrule}
p(y_{n+1:N}\mid y_{1:n}) = \prod_{i=n+1}^N p(y_i \mid y_{1:i-1}).
\end{equation} 
However, we must take care in our elicitation of $\{p(\cdot\mid y_{1:N})\}_{N\geq n}$ to ensure the existence of the limit $\theta_\infty$. As this is technical, we defer a formal discussion of this choice and the conditions required to Section \ref{sec:PR}. For now, we point out that a sufficient condition is for the 1-step ahead predictive densities to satisfy a martingale condition similar to that of Doob, with details given in Section \ref{sec:pred_coherence}. It may seem that constructing  this sequence will incur too much complexity, but we will show this is in fact feasible and desirable. One key idea is to utilize a general sequential updating procedure whereby given an observed $Y_{N}=y_{N}$, we have a direct and tractable  iterative update
$\{p(\cdot\mid y_{1:N-1}),y_{N}\}\mapsto p(\cdot\mid y_{1:N}).$

 \subsubsection{Recovering the quantity of interest}\label{sec:loss_functions}
We now discuss the second task: given a sample $Y_{n+1:\infty}$, we require a procedure to recover the quantity of interest. In a traditional parametric Bayesian model, the quantity of interest is usually the unknown parameter $\theta$ that indexes the sampling density, and as shown by Doob, the limiting posterior mean $\bar{\theta}_\infty$ serves this purpose. A more general framework is the decision task discussed in \citet{Bissiri2016}, where the aim is to minimize a functional of an unknown distribution function $F_0$ from which samples $Y_{1:n}$ are i.i.d.. For some loss function $\ell(\theta, y)$, the quantity of interest $\theta$ is now defined as
\begin{equation}\label{eq:stat_of_interest}
\theta_0 = \argmin_{\theta} \int \ell(\theta,y)\,dF_0(y).
\end{equation}
More details can be found for example in \cite{Huber2004} and \cite{Bissiri2016}.  Typical examples are $\ell(\theta, y) = |\theta-y|$ for the median, $\ell(\theta,y) = (\theta-y)^2 $ for the mean, and $\ell(\theta,y) = -\log f_\theta(y)$ for the Kullback-Leibler minimizer between some parametric density $f_\theta$ and the sampling density $f_0$. The choice of the negative log--likelihood is also interesting as it allows us to target the parameters of a parametric model without the assumption that the model is well--specified \citep{Walker2013,Bissiri2016}. {While misspecification under our framework is still an open question, the Bayesian bootstrap has particularly desirable theoretical and practical properties under misspecification \citep{Lyddon2018,Lyddon2019, Fong2019a}.}
We will also consider more general forms of $\theta_0$, e.g. the density of $F_0$.

Working now in the space of probability distributions, the traditional Bayesian approach would be to elicit a prior on $F$, perhaps nonparametric, and derive the posterior $\Pi(d F\mid y_{1:n})$. {Here, $F$ represents the Bayesian's subjective belief on the unknown true $F_0$.} A posterior sample of $\theta$ is then obtained as follows: draw $F\sim\Pi(d F\mid y_{1:n})$ and compute the $\theta$ minimizing
$\int \ell(\theta,y)\,d F(y)$. For our generalization beyond the likelihood-prior construction, we do not have a posterior mean nor a posterior $F$, and thus require an alternative to recover the quantity of interest given a sample of $Y_{n+1:\infty}$ conditioned on $y_{1:n}$. Our proposal is to construct the random limiting empirical distribution function 
\begin{equation}\label{eq:limiting_empirical}
F_\infty(y)=\lim_{N\to\infty}\frac{1}{N}\left\{\sum_{i=1}^n \mathbbm{1}(y_i\leq y)+  \sum_{i=n+1}^N \mathbbm{1}(Y_i\leq y)\right\}
\end{equation}
and take $\theta$ to minimize $\int \ell(\theta,y)\,d F_{\infty}(y)$. Here, our $F_\infty$ takes the place of the posterior draw of $F$, and its existence will rely on the martingale condition as mentioned above. We can write $\theta(F_\infty)$ or $\theta(Y_{1:\infty})$ interchangeably for the parameter of interest computed from the  completed information. If we specify $p(\cdot \mid y_{1:n})$ through the usual likelihood--prior construction, then sampling $F$ from the posterior in fact yields the same random distribution function as $F_\infty$ almost surely;  this theoretical justification for the limiting empirical distribution function $F_\infty$ is in Appendix \ref{Appendix:Bayes_pred_emp}.

\subsection{The martingale posterior}\label{sec:martingale_posterior}
Our framework for predictive inference is summarized as follows. {Suppose we observe $Y_{1:n}$ i.i.d. from some unknown $F_0$ and are interested in the $\theta_0$ defined by \eqref{eq:stat_of_interest}.} We specify a sequence of predictive densities $\{p(\cdot\mid y_{1:n})\}_{n\geq 0}$ which satisfies the martingale condition to be discussed in Section \ref{sec:pred_coherence} and implies a joint distribution through \eqref{eq:chainrule}. We then impute an infinite future dataset through
\begin{equation*}\label{eq:Bayes_PR_seq}
Y_{n+1}  \sim p(\cdot\mid y_{1:n}),\quad
Y_{n+2} \sim p(\cdot \mid y_{1:n+1}),
\quad \ldots,\quad
Y_{N}\sim p(\cdot \mid y_{1:N-1})
\end{equation*}
for $N\to \infty$. Given the infinite random dataset $Y_{n+1:\infty}$ and the corresponding empirical distribution function $F_\infty$, we compute $\theta_\infty = \theta\left(F_\infty\right)$. We designate the distribution of $\theta_\infty$ as the {martingale posterior}, where we use the notation $\Pi_\infty$ for comparability to traditional Bayes.

\begin{definition}[Martingale posterior] 
The martingale posterior distribution is defined as 
\begin{equation}\label{eq:mart_post}
\Pi_{\infty}(\theta_\infty \in A \mid y_{1:n}) = \int \mathbbm{1}\{{\theta(F_{\infty})} \in A\}\,  d\Pi(F_{\infty} \mid y_{1:n}) \,,
\end{equation}
for measurable set $A$.
\end{definition}

Drawing samples of $\theta_\infty$ from the martingale posterior involves repeating the above simulation procedure given above. We refer to this Monte Carlo scheme as {predictive resampling}, which has strong connections with the Bayesian bootstrap of \citet{Rubin1981}, as we will see in Section \ref{sec:BB}.  In practice however, we may be unable to simulate $N \to \infty$, or the study population may be of finite size $N$. In this case, we can instead impute $Y_{n+1:N}$ for finite $N$, giving us the analogous empirical distribution function $F_N$ and parameter $\theta_N = \theta(F_N)$ or $\theta(Y_{1:N})$.

\begin{definition}[Finite martingale posterior]

The finite martingale posterior is similarly defined as 
\begin{equation}\label{eq:mart_post_N}
\Pi_N(\theta_N \in A \mid y_{1:n}) = \int \mathbbm{1}\{{\theta(y_{1:N})} \in A\}\, p(y_{n+1:N} \mid y_{1:n})\,d y_{n+1:N}\,.
\end{equation}
\end{definition}
In the finite form, the role of the two constituent elements, $p(y_{n+1:N} \mid y_{1:n})$ and $\theta(y_{1:N})$, is even clearer. { For infinite populations, we also highlight that the value of $\theta_N$ varies around $\theta_\infty$, but this may be negligible for sufficiently large $N$. If the population is actually finite and of size $N$, then $\theta_N$ would be the actual target and thus not an approximation.} Finally, we reiterate that the martingale posterior \eqref{eq:mart_post} is equivalent to the traditional Bayesian posterior when using \eqref{eq:posterior_predictive} as the predictive. A summary of the notation and an illustration of the imputation scheme is provided respectively in Appendices \ref{Appendix:notation}, \ref{Appendix:tables}.

\subsection{The Bayesian bootstrap} \label{sec:BB} 
The resemblance of the martingale posterior to a bootstrap estimator should not have gone unnoticed, as both involve repeated sampling of observables followed by computing estimates from the sampled dataset. 
The Bayesian bootstrap of \citet{Rubin1981} is often described as the Bayesian version of the frequentist bootstrap. After observing $y_{1:n}$, one draws a random distribution function from the posterior through
\begin{equation}
w_{1:n} \sim \text{Dirichlet}(1,\ldots,1), \quad {F}(y) = \sum_{i=1}^n w_i \, \mathbbm{1}(y_i \leq y).
\end{equation}
 A posterior sample  of the statistic of interest can then be computed as $\theta({F})$. One interpretation of the Dirichlet weights is to generate uncertainty through the randomization of the objective function \citep{Newton1994,Jin2001,Newton2020,Ng2020}. Closer to our perspective are the connections to Bayesian nonparametric inference, which have been explored in much detail within the literature as it is the non-informative limit of a posterior Dirichlet process \citep{Lo1987,Muliere1996,Ghosal2017}.  Recent work has exploited the computational advantages of the Bayesian bootstrap for scalable nonparametric inference; see \cite{Saarela2015,Lyddon2018,Fong2019a,Newton2020, Knoblauch2020,Nie2020}.

\subsubsection{The empirical predictive}
Within the framework of martingale posteriors, the Bayesian bootstrap has a particularly elegant interpretation that follows from the equivalence to the P\'olya urn scheme \citep{Blackwell1973,Lo1988}. The Bayesian bootstrap is equivalent to the martingale posterior if we define our sequence of predictive probability distribution functions to be the sequence of empirical distribution functions, that is
\begin{equation}\label{eq:empirical}
{P}(Y_{n+1} \leq y \mid y_{1:n}) = F_n(y)=  \frac{1}{n}\sum_{i=1}^n \mathbbm{1}(y_i\leq y).
\end{equation}
 This is easy to see as sampling $Y_{n+1} \sim F_n(y)$ amounts to drawing with replacement 1 of $n$ colours with probability $1/n$ from the urn, and updating to $F_{n+1}(y)$  is equivalent to reinforcing the urn, that is
 $$
 F_{n+1}(y) = \frac{n}{n+1}F_n(y) + \frac{1}{n+1}\mathbbm{1}(y_{n+1}\leq y).
 $$
 Continuing on to $\infty$, the proportions of colours converge in distribution to  the Dirichlet distribution. Interestingly, this choice of predictive implies an exchangeable future sequence from the connection to the Dirichlet process. {The atomic support of the predictive is however slightly problematic if $F_0$ is continuous, as any new observations from $F_0$ will be assigned a predictive probability of zero; we will introduce methodology that remedies this in Section \ref{sec:copula}. Generalizations to other atomic predictives can for example be found in \cite{Zabell1982, Muliere2000}.}

{One can consider the empirical distribution function as the simplest nonparametric predictive for i.i.d. data, and can thus regard the Bayesian bootstrap as the simplest Bayesian nonparametric model.} The uncertainty from the Bayesian bootstrap arises not from the random weights, but from the sequence of empirical predictive distributions.  We resample with replacement, treating each resampled point as a new observed datum; this fundamental observation is our motivation for the term {predictive resampling}.

\subsubsection{Comparison to the frequentist bootstrap}

{Throughout this section, we have assumed the existence of an underlying $F_0$ from which $Y_{1:n}$ are i.i.d., which in turn implies the existence of an unknown true $\theta_0$ much like the frequentist. This has some connections to frequentist consistency under our framework, which we discuss in Section \ref{sec:consistency}. The posterior random variable $\theta_\infty$ then represents our subjective uncertainty in $\theta_0$ after observing $Y_{1:n} = y_{1:n}$.} The Bayesian bootstrap and Efron's bootstrap \citep{Efron1979} are then ideal vessels for the contrasting of Bayesian and frequentist uncertainty.  Both methods are nonparametric and begin by constructing the empirical predictive $F_n$ as in \eqref{eq:empirical} from the atoms of $y_{1:n}$ as an estimate of $F_0$, and both involve resampling. The key difference lies in {how} the resampling is carried out.

The frequentist draws a dataset of size $n$ i.i.d. from $F_n$, which we write as $Y^*_{1:n}$ with corresponding empirical distribution function $F_n^*$, and computes $\theta(F_n^*)$ as a random sample of the estimator. The Bayesian on the other hand draws an infinite future dataset $Y_{n+1:\infty}$ through predictive resampling, and computes $\theta(F_\infty)$  as a random sample of the estimand, where $F_\infty$ is the limiting empirical distribution function of $\{y_{1:n}, Y_{n+1:\infty}\}$, noting again that the Bayesian holds $y_{1:n}$ fixed. This is summarized in Algorithms \ref{alg:bayes_bootstrap} and \ref{alg:efrons_bootstrap}. { Notably, the specification in both bootstraps are equivalent: it is merely the elicitation of $F_n(y)$, which entirely characterizes both types of uncertainty.}

\begin{figure}
\small
  \centering
  \begin{minipage}{.45\linewidth}
  \centering
\begin{algorithm}[H]\label{alg:bayes_bootstrap}
\DontPrintSemicolon
  \SetAlgoLined
{Set ${F}_n$ from the observed data $y_{1:n}$}\;
\For{$j \gets 1$ \textnormal{\textbf{to}} $B$} {
  \For{$i \gets n+1$ \textnormal{\textbf{to}} $\infty$} {
  Sample $Y_{i}  \sim {F}_{i-1}$\;
  Update  ${F}_{i} \mapsfrom \left\{{F}_{i-1}, Y_{i}\right\}$ \;
  }
 Compute ${F}_\infty$ from $\{y_{1:n}, Y_{n+1:\infty }\}$\;
 Evaluate   ${\theta}_\infty^{(j)} =\theta({F}_\infty) $ 
 }
Return $\{\theta_\infty^{(1)},\ldots,\theta_\infty^{(B)} \}$\;
\caption{Bayesian bootstrap}
\end{algorithm}
  \end{minipage}
  \hspace{2mm}
  \begin{minipage}{.45\linewidth}
  \centering
\begin{algorithm}[H]\label{alg:efrons_bootstrap}
\DontPrintSemicolon
  \SetAlgoLined
{Set ${F}_n$ from the observed data $y_{1:n}$\; }
\For{$j \gets 1$ \textnormal{\textbf{to}} $B$} {
  \For{$i \gets 1$ \textnormal{\textbf{to}} $n$} {
  Sample $Y^*_{i}  \sim {F}_{n}$\;
  No update to ${F}_n$  \; 
  \vspace{.1mm}
  }
 Compute ${F}^*_n$ from $\{Y^*_{1:n }\}$\;
 Evaluate   ${\theta}_n^{(j)} =\theta({F}_n^*) $ 
 }
  Return $\{\theta_{n}^{(1)}, \ldots, \theta_{n}^{(B)}\}$ \;
\caption{Efron's bootstrap}
\end{algorithm}
  \end{minipage}
\end{figure}

\subsection{Related work}\label{sec:related}
There have been many others that shared de Finetti's view on the emphasis on observables for inference.  The work of \cite{Dawid1984,Dawid1992a,Dawid1992b} on prequential statistics, a portmanteau of probability/predictive and sequential, is one such example. {In his work,  Dawid focuses on the importance of forecasting, and introduces statistical methodology that assign predictive probabilities and assesses these methods on their agreement with the observed data.  In particular, \cite{Dawid1984} recommends eliciting a sequence of 1-step ahead predictive distributions as we do,  but motivates this by arguing that forecasting is the main statistical task. As pointed out in Section \ref{sec:intro}, this is in contrast to our case, where parameter inference is the main task of interest. We will see in Section \ref{sec:pred_coherence} that stricter conditions are required on this sequence of predictives for inference.} Another strong proponent of the predictive approach is the work of Geisser: he believed that the prediction of observables was of much greater importance than the estimation of parameters, which he described as ``artificial constructs" \citep{Geisser1975}. His emphasis on the predictive motivated cross-validation {\citep{Stone1974, Geisser1974}, which is now popular for Bayesian model evaluation \citep{Vehtari2002, Gelman2014}. 
Works such as \cite{Dawid1985,Lauritzen1988} also consider parameters as functions of the infinite sequence of observations using the notion of repetitive structures.}  Finally, the work of Rubin on both the potential outcomes model \citep{Rubin1974} and multiple imputation \citep{Rubin2004} highlights the idea of inference via imputation.

An early application of what is essentially finite predictive resampling and martingale posteriors is Bayesian inference for finite populations, first discussed in \cite{Roberts1965,Ericson1969} and later by \citet{Geisser1982,Geisser1983}.  A finite population Bayesian bootstrap is described in \citet{Lo1988}, in which a finite P\'olya urn is used to simulate from the posterior. The `P\'olya posterior' of \citet{Ghosh1997} uses the same approach following an admissibility argument.  These methods have applications in  survey sampling or the interim monitoring of clinical trials \citep{Saville2014}. 

There have been recent exciting directions of work that investigate the predictive view of Bayesian nonparametrics (BNP). \citet{Fortini2000} investigate under what conditions parametric models arise from the sequence of predictive distributions using the concept of predictive sufficiency, and derive conditions such that the joint distribution is exchangeable. \cite{Fortini2012,Fortini2014}  discuss the construction of a range of popular exchangeable BNP priors through a sequence of predictive distributions, motivated through a predictive de Finetti's representation theorem \citep[Theorem~2]{Fortini2012}.   \citet{Berti2020} then generalize the nonparametric approach to  c.i.d. sequences; we will later see that c.i.d. sequences, as introduced in \citet{Berti2004}, play a crucial role in our work. However, the previously described methods are mostly constrained to the discrete case. \cite{Hahn2015} and \cite{Hahn2018} construct c.i.d. models through a predictive sequence for univariate density estimation, respectively utilizing the kernel density estimator and the {bivariate copula}.  \citet{Hahn2015} also discusses the connection of Bayesian uncertainty and prediction with a weaker argument, and gives a similar example to our Example \ref{ex:bayes_normal}. Predictive resampling is then used to sample nonparametric densities from a finite martingale posterior; however \citet{Hahn2015} instead specifies the predictive distribution $P_N$ for large $N$ and works backwards to find the sequence of predictives. \citet{Fortini2020} analyze the predictive recursion algorithm of \cite{Newton1998} and the implied underlying quasi-Bayesian model. In their work, they carry out predictive resampling to simulate from the prior law of the mixing distribution in an example, and obtain its asymptotic distribution under the c.i.d. model, that is an asymptotic approximation to the martingale posterior. { An interesting aside is the recent work of \cite{Waudby2020} which utilizes adaptive betting with martingale conditions for the purpose of constructing frequentist confidence intervals.} We aim to unify these related strands of research under a single framework.

\section{Predictive resampling for martingale posteriors} \label{sec:PR}
For the martingale posterior, we now embark on the task of eliciting the general 1-step ahead predictive distributions, with the traditional Bayesian posterior predictive as a special case.  For notational convenience, we write the sequence of predictive probability distribution functions estimated after observing $Y_{1:i} = y_{1:i}$ as 
\begin{equation}\label{eq:pred_seq}
P_i(y)\vcentcolon=  P(Y_{i+1} \leq y \mid y_{1:i}), \quad i \in \{1,2,\ldots\}
\end{equation}
which may have corresponding density functions  $p_i(y)$.  The subscript indicates the length of the conditioning sequence, and there may be a $P_0(y)$ as some initial choice. For a general sequence of predictives, where exchangeability no longer necessarily holds, we instead define our joint distribution on $y_{1:N}$ through this sequence of 1-step ahead predictives and the chain rule as in \eqref{eq:chainrule}. {The Ionescu-Tulcea theorem \citep[Theorem 5.17]{Kallenberg1997} guarantees the existence of such a joint distribution as we take $N \to \infty$, which has been pointed out by works such as \cite{Dawid1984,Fortini2012, Berti2020}}.

Beyond the traditional Bayesian posterior predictive, there is good justification for specifying the model with 1-step ahead predictives, instead of say $m$-step ahead. 
It is simple to interpret and estimate a 1-step ahead predictive as the decision maker's best estimate of the unknown sampling distribution function $F_0$, and methods such as maximum likelihood estimation already do this. There are also connections with forecasting and prequential statistics \citep{Dawid1984}. Finally, we will see that a 1-step update of the predictive allows for the enforcing of the c.i.d. condition for predictive coherence.

While the prescription of \eqref{eq:pred_seq} remains a subjective task, we find it to be no more subjective than the selection of a likelihood function. {There is no longer a need to elicit subjective distributions on parameters which merely index the sampling distribution with no physical meaning, which has been described as `intrinsic'  \citep{Dawid1985}.} In nonparametric inference, we also do not need to elicit priors directly on the space of probability distributions, which can be cumbersome. The uncertainty arises simply from the elicitation of \eqref{eq:pred_seq}. It is clear that we can still use external information and subjective judgement not provided by the data $y_{1:n}$ in this construction.  

\subsection{A practical algorithm for uncertainty}\label{sec:MP_N}
 Given the model specification \eqref{eq:pred_seq}, suppose we wish to undertake inference on a statistic of interest $\theta(F_0)$, defined through a loss function $\ell(\theta,y)$ as in \eqref{eq:stat_of_interest}. We can obtain finite martingale posterior samples through predictive resampling  given in Algorithm \ref{alg:predictive_resampling}, noting the similarity to the Bayesian bootstrap algorithm.
\begin{figure}[ht]
\small
  \centering
  \begin{minipage}{.5\linewidth}
\begin{algorithm}[H]\label{alg:predictive_resampling}
\DontPrintSemicolon
  \SetAlgoLined
{Compute $P_n$ from the observed data $y_{1:n}$\;
$N>n$ is a large integer}\;
  \For{$j \gets 1$ \textnormal{\textbf{to}} $B$}{
  \For{$i \gets n+1$ \textnormal{\textbf{to}} $N$} {
  Sample $Y_{i}  \sim {P}_{i-1}$\;
  Update  $P_{i} \mapsfrom \left\{P_{i-1}, Y_{i}\right\} $\;
  }
 Compute $F_N$ from $\{y_{1:n}, Y_{n+1:N }\}$\;
 Evaluate   ${\theta}^{(j)}_N =\theta(F_N) $ or ${\theta}^{(j)}_N =\theta({P}_N) $ \;}
 Return $\{\theta_N^{(1)},\ldots,\theta_N^{(B)} \} \iid \Pi_N(\cdot \mid y_{1:n})$\;
\caption{Predictive resampling}
\end{algorithm}
  \end{minipage}
\end{figure}

In summary, we run a forward simulation starting at $P_n(y)$ by consecutively sampling from the 1-step ahead predictives and updating as we go. For large $N$, we now have a random dataset $\{y_{1:n},Y_{n+1:N}\}$ from which we can compute the empirical distribution function $F_N(y)$ and statistic of interest $\theta(F_N)$.  In particular, only when the sequence of predictives takes on the form \eqref{eq:posterior_predictive}, combined with the self-information loss, $-\log f_{\theta}(y)$,  is this procedure equivalent to traditional Bayesian inference.

The empirical distribution is atomic, which may be problematic if the object of interest $\theta_0$ requires the limiting  $F_\infty$ to be continuous, for example if $\theta_0$ is the probability density  of $F_0$ or a tail probability. In this case, we can instead compute $\theta(P_N)$, where $P_N$ is the random predictive distribution function conditioned on $\{y_{1:n}, Y_{n+1:N}\}$, which would typically be continuous. We can regard $P_N$ as the finite approximation to the limiting predictive distribution function $P_\infty := \lim_{N\to\infty} P_N$, which serves the same purpose as the limiting empirical $F_\infty$ in Section \ref{sec:loss_functions}. In fact, $P_\infty$ and $F_\infty$ coincide for traditional Bayesian models, and even for the more general c.i.d. sequence of predictives that we will consider shortly. We discuss this in Appendix \ref{Appendix:limit_pred_emp}, borrowing results from \citet{Doob1949}, \cite{Berti2004} and \cite{Lijoi2004}.

Some experimental and theoretical guidance for selecting a sufficiently large $N$ to estimate $P_\infty$ is given in Sections \ref{sec:illustrations} and \ref{sec:theory}. However, it is also interesting to consider a {finite} population, where the $F_0$ of interest is indeed the empirical distribution function of a population of size $N$, as discussed in Sections \ref{sec:martingale_posterior} and \ref{sec:related}. In this case, truncating predictive resampling at $N$ indeed returns the correct uncertainty in any parameter of interest $\theta(Y_{1:N})$ of the finite population.

\subsection{Predictive coherence and conditionally identically distributed sequences}\label{sec:pred_coherence}

{The notion of coherence on one's belief on the parameter $\theta$ is key to the subjective Bayesian, where coherence may be defined in a decision-theoretic sense \citep[Chapter~2.3]{Bernardo2009} or through Dutch book arguments (e.g. \citet{Heath1978}). Extensions of coherence to forecasting are given in \citet{Lane1984, Berti1998},  and more examples of coherence in general can be found in \citet{Robins2000, Eaton2004}. More recently, the notion of coherence of belief updating was introduced in \cite{Bissiri2016}, where a belief update on a statistic of interest $\theta$ is coherent if the update is equivalent whether computed sequentially with $y_1$ followed by $y_2$ or with  $\{y_1,y_2\}$ in tandem through an additive loss condition. } In bypassing the traditional likelihood-prior construction, we must forsake the usual coherence of belief updating and exchangeability. Instead, we specify conditions for a valid martingale posterior entirely in terms of the predictive distribution function, which we term \textit{predictive coherence}. 
 
Suppose we observe $Y_{1:n}$ i.i.d. from some $F_0$ and construct $P_n(y)$ as in \eqref{eq:pred_seq}. We can then view the predictive machine $P_n(y)$ as the best estimate of the unknown distribution function $F_0$ from which the data arose, incorporating all observed data and any possible subjective knowledge. The first minimal condition is that the sequence of predictive distribution functions $P_{n+1}(y),P_{n+2}(y)\ldots$  converges to a random distribution function. Secondly, we would ensure that predictive resampling does not introduce any new information or bias, as $P_n$ is already our best summary of the observed $y_{1:n}$, and the procedure should merely return uncertainty.  Formally, we write these conditions respectively as follows:
 \begin{condition}[Existence]\label{cond:existence}
 The sequence $P_{n+1}(y),P_{n+2}(y),\ldots$ converges  to a random $P_\infty(y)$ almost surely for each $y\in \mathbb{R}$, where $P_\infty$ is a random probability distribution function.
 \end{condition}
  \begin{condition}[Unbiasedness]\label{cond:unbiased}
 The posterior expectation of the random distribution function satisfies
$$E\left[P_\infty(y) \mid y_{1:n} \right]=P_n(y)$$ almost surely for each $y\in \mathbb{R}$.   
 \end{condition}
 Under Condition \ref{cond:existence}, $P_\infty$ is {defined} through the sequence of predictives, and we can thus treat $P_\infty$ directly as the random distribution function without the need for an underlying Bayes' rule representation. This in turn gives us the posterior uncertainty in any statistic $\theta(P_\infty)$.  Condition \ref{cond:unbiased} is stricter, and implies that $P_n$ is our best estimate of $F_0$ and is equal to the posterior mean.

 Fortunately, Conditions \ref{cond:existence} and \ref{cond:unbiased} are satisfied if the sequence $Y_{n+1},Y_{n+2},\ldots$ is \textit{conditionally identically distributed} (c.i.d.), as introduced and studied in \citet{Berti2004}. Many useful properties of c.i.d. sequences have been shown in their work, which we now summarize. {The sequence  $Y_{n+1},Y_{n+2},\ldots$ is c.i.d if we have 
\begin{equation*}\label{eq:cid}
P\left(Y_{i+k}\leq y \mid y_{1:i}\right) = P_i(y), \quad \forall k > 0
\end{equation*}
almost surely for each $y \in \mathbb{R}$. This states that conditional on $y_{1:i}$, any future data points will be identically distributed  according to the predictive $P_i$.} This predictive invariance is particularly natural as a minimal predictive coherence condition, and serves as an analogue to de Finetti's exchangeability assumption in the predictive framework. In fact, as shown in \cite{Kallenberg1998}, the c.i.d. condition is a weakening of exchangeability, and \citet{Berti2004} also show that  c.i.d. sequences are asymptotically exchangeable, which  we state formally in Theorem \ref{Th:cid_exchangeable} in Section \ref{sec:theory_mv}.

{An equivalent formulation of c.i.d. sequences which connects closely to the predictive coherency conditions is that $P_i(y)$ is a martingale for $i\in \{n+1, n+2,\ldots\}$, that is
\begin{equation}\label{eq:martingale}
E\left[P_{i}(y) \mid y_{1:i-1} \right] \equiv \int P_{i}(y)\, dP_{i-1}(y_i) = P_{i-1}(y)
\end{equation}  
almost surely for each $y\in \mathbb{R}$, noting that $P_i$ depends on $y_i$ as in \eqref{eq:pred_seq}.} Relying again on Doob's martingale convergence theorem \citep{Doob1953}, the sequence $P_n(y),P_{n+1}(y),\ldots$ converges to $P_\infty(y)$ almost surely for each $y\in \mathbb{R}$, and $P_\infty$ can be shown to be a random probability distribution function \citep{Berti2004}; we state this formally in Theorem \ref{Th:weak} in Section \ref{sec:theory_mv}. In this case, we also designate the distribution of $P_\infty$ as the martingale posterior when we do not specify $\theta_\infty$. Condition \ref{cond:unbiased} is then satisfied as the sequence $P_{n+1}(y), P_{n+2}(y), \ldots$ is uniformly integrable. Furthermore, we are guaranteed the existence of the limiting empirical distribution function $F_\infty$ as required in Section \ref{sec:loss_functions}, and in fact $F_\infty(y) = P_\infty(y)$ almost surely so the interchangeability of $\theta(F_\infty)$ and $\theta(P_\infty)$ is justified. This equivalence, as well as the convergence of $\theta(Y_{1:N})$ with $N$ for a certain class of parameters, is discussed in Appendix \ref{Appendix:cid_pred_emp}. {Although not explored here, connections of the c.i.d. property to other notions of coherence, such as those given at the start of this subsection, would be interesting to investigate especially given the absence of the prior distribution.}

Although the above predictive coherence conditions are for a valid martingale posterior, we still need to specify a sequence of predictive distributions. Clearly the traditional Bayesian posterior predictive satisfies the above conditions, but in the interest of computational expediency or the desire to bypass the likelihood-prior construction, we may wish to consider more general predictive machines. The remainder of this paper will consider recursive predictive densities using bivariate copulas. 

\section{Recursive predictives with bivariate copulas } \label{sec:copula}
In this section, we focus primarily on the elicitation of the sequence of predictives \eqref{eq:pred_seq} in the continuous case, where $p_i(y)$ is the density of $P_i(y)$ in \eqref{eq:pred_seq}. Analogous predictives are derivable for the discrete case, and these are obtained in \cite{Berti2020}. In particular, we investigate the prescription of this sequence of predictives through a recursive manner, that is for $i \in \{0,1,\ldots\}$
\begin{equation*}\label{eq:recursive_update}
p_{i+1}(y) = \psi_{i+1}^\rho\left\{p_{i}(y), y_{i+1} \right\}
\end{equation*}
where $\psi_i^\rho$  is a sequence of update functions, possibly parameterized by a hyperparameter $\rho$. In this case, we require an initial guess $p_0(y)$ for our recursion, which plays the role of a prior guess on $f_0$. A recursive update of this form is not necessary for a martingale posterior, but it allows for simple satisfaction of conditions for predictive coherence as discussed in Section \ref{sec:pred_coherence}, and computations for predictive resampling will also be significantly easier. Furthermore, when one is only interested in estimating $p_n(y)$, recursive updates may have computational advantages as one does not need to explicitly estimate the posterior. 

Recursive updates have previously been motivated as a fast alternative to MCMC in Dirichlet process mixture models (DPMM).  The predictive recursion algorithm was first introduced by \citet{Newton1998}, which estimates the mixing distribution through a recursive update, and its properties have been studied in detail in the literature; see \citet{Martin2018} for a thorough review. One interesting property shown in \citet{Fortini2020} is that the sequence of observables in Newton's algorithm is c.i.d.; however, the computation of the predictive densities is intractable and requires numerical integration, so we will not discuss this method further here. Direct recursive updates for the predictive density were then introduced in \cite{Hahn2015,Hahn2018,Berti2020}, all of which satisfy the c.i.d. condition.
The bivariate copula method of \cite{Hahn2018} is particularly tractable and well motivated, and we will now build on this method in this section.

\subsection{Bivariate copula update}
{To satisfy the c.i.d. condition required for predictive coherence, we can extend the martingale condition to hold for the sequence of densities $p_n,p_{n+1},\ldots$ such that for $i\in \{n+1,n+2,\ldots\}$
\begin{equation}\label{eq:martingale_density}
E\left[p_{i}(y) \mid y_{1:i-1} \right] \equiv \int p_{i}(y) \, p_{i-1}(y_{i})\, dy_{i} = p_{i-1}	(y)
\end{equation}  
almost surely for each $y \in \mathbb{R}$, assuming the expectations exist. We highlight again that $p_i$ depends on $y_i$ as it is the density of \eqref{eq:pred_seq}.} The above is a sufficient condition for \eqref{eq:martingale} to hold, so our sequence is c.i.d. and the existence and unbiasedness conditions are satisfied giving us a valid martingale posterior. In fact, the martingale convergence theorem shows that $p_i(y) \to p_\infty(y)$ almost surely for each $y \in \mathbb{R}$, but more assumptions are needed to show that $p_\infty$ is the density of $P_\infty(y)$; we explore this in Theorem \ref{Th:absolute_continuity} in Section \ref{sec:theory_mv}.

One particular tractable form of update rule $\psi_i^\rho$ that satisfies \eqref{eq:martingale} is the {bivariate copula} \citep{Nelsen2007} update interpretation of Bayesian inference first introduced in \citet{Hahn2018} for univariate data.   A bivariate copula is a bivariate cumulative distribution function $C:[0,1]^2 \to [0,1]$ with uniform marginal distributions, and in the cases we consider it will have a probability density function $c:[0,1]^2 \to \mathbb{R}$. The bivariate copula can be regarded as characterizing the dependence between two random variables independent of their marginals, which can be seen through Sklar's theorem in the bivariate case.
\begin{theorem}[\citet{Sklar1959}]\label{Th:Sklar}
For a bivariate cumulative distribution function $F(y_1,y_2)$ with continuous marginals $F_1(y_1),F_2(y_2)$, there exists a unique bivariate copula $C$ such that
\begin{equation*}
F(y_1,y_2) = C\{F_1(y_1),F_2(y_2)\}.
\end{equation*}
Furthermore, if $F$ has a density $f$ with marginal densities $f_1,f_2$, we can write
\begin{equation*}
f(y_1,y_2) = c\{F_1(y_1),F_2(y_2)\}\, f_1(y_1) \, f_2(y_2)
\end{equation*}
where $c$ is the density of $C$.
\end{theorem}

This holds for higher dimensions, but we state it for $d= 2$ as this is what we will be working with. From this, we can see that the bivariate copula can model the dependence structure between consecutive predictive densities, and thus we have the following corollary, with the proof given in Appendix \ref{Appendix:corr_copula_CID}.
\begin{corollary}\label{corr:copula_CID}
The sequence  of conditional densities $p_0,p_{1},\ldots$ satisfies the martingale condition \eqref{eq:martingale_density} if and only if there exists a unique sequence of bivariate copula densities $c_1,c_2,\ldots$ such that 
\begin{equation}\label{eq:copula_update}
p_{i+1}(y) = c_{i+1}\{P_i(y),P_i(y_{i+1})\}\, p_i(y)
\end{equation} 
for $i \in \left\{0,1,\ldots \right\}$ and $P_i$ is the distribution function of $p_i$. 
\end{corollary}

In the univariate case, we can thus elicit a c.i.d. model through a sequence of copulas, that is we have \eqref{eq:copula_update} as our update function $\psi_{i+1}^\rho$. { We highlight that $c_{i+1}$ is the bivariate copula density that models the dependence between $\{Y_{i+1},Y_{i+2}\}$ conditioned on $Y_{1:i}$. Although the sequence $c_{i+1}$ can technically depend arbitrarily on $y_{1:i}$ (and the sample size $i+1$) without violating the martingale condition, we will later constrain this dependence.} As all exchangeable Bayesian models are c.i.d., there exists a unique sequence of copulas which may or may not be tractable that characterize the model \citep{Hahn2018}. This sequence takes on exactly the form
\begin{equation}\label{eq:Bayes_cop}
p_{i+1}(y) =\underbrace{\frac{\int f_\theta(y)\,f_\theta(y_{i+1})\,\pi(\theta \mid y_{1:i}) \, d\theta}{p_i(y) \, p_i(y_{i+1})}}_{c_{i+1}\{P_i(y),P_i(y_{i+1}) \}}\, p_i(y ).
\end{equation}
The copula density arises following Theorem \ref{Th:Sklar} as the numerator in \eqref{eq:Bayes_cop} is the joint density $p_i(y,y_{i+1})$ with marginal densities $p_i(y)$ and $p_i(y_{i+1})$.
Instead of specifying the likelihood and prior, we will now consider the specification of the sequence of copulas $c_i$ directly. The form for $c_i$ inspired by the DPMM is particularly attractive, and serves well as the canonical extension of the Bayesian bootstrap predictive to continuous random variables.  In the remainder of this section, we will first review the method of \citet{Hahn2018} for univariate density estimation, and extend the methodology to include predictive resampling and hyperparameter selection. We then introduce analogous copula updates for more advanced data settings, including multivariate density estimation, regression and classification.

\subsection{Univariate case}\label{sec:univariate}

Tractable forms of this sequence of copulas in Bayesian models are investigated in \cite{Hahn2018},  which correspond to conjugate priors. The update of particular interest is that of the DPMM \citep{Escobar1995} of the particular form
\begin{equation}\label{eq:DP_mixture}
f_G(y) = \int \mathcal{N}(y \mid \theta,1) \, dG(\theta),\quad
G \sim \text{DP}\left(a, G_0 \right),\quad G_0 = \mathcal{N}(\theta \mid 0,\tau^{-1}),
\end{equation}
where $a>0$ is the scalar precision parameter that we set to $a = 1$. The model is nonparametric, making it a strong candidate for a predictive update, but only the copula update for $i = 0$ is tractable. Inspired by this first update step, \citet{Hahn2018} suggest that the general update to compute the density $p_i(y)$ after observing $y_{1:i}$  for $i \in \{0,\ldots,n-1\}$ takes on the form
\begin{equation}\label{eq:DP_copdens}
\begin{aligned}
p_{i+1}(y) &= \left(1-\alpha_{i+1}\right)p_i(y) + \alpha_{i+1} c_\rho\left\{P_i(y),P_i(y_{i+1})\right\}p_i(y) \\
P_{i+1}(y) &= (1-\alpha_{i+1})P_i(y) + \alpha_{i+1} H_\rho\left\{P_i(y),P_i(y_{i+1})\right\}
\end{aligned}
\end{equation}
where $P_i(y)$ is the distribution function of $p_i(y)$. Here $c_\rho(u,v)$ is the bivariate Gaussian copula density and $H_\rho(u,v)$ is the conditional Gaussian copula of the forms:
\begin{equation}\label{eq:Gauss_copdens}
c_\rho(u,v) = \frac{\mathcal{N}_2\left\{ \Phi^{-1}(u),\Phi^{-1}(v)\mid 0,1,\rho\right\}	}{\mathcal{N}\{\Phi^{-1}(u) \mid 0,1\}\, \mathcal{N}\{\Phi^{-1}(v) \mid 0,1\}}, \quad H_\rho(u,v) = \Phi\left\{ \frac{\Phi^{-1}(u) - \rho \Phi^{-1}(v)}{\sqrt{1-\rho^2}}\right\}
\end{equation}
where $\Phi^{-1}$ is the standard inverse normal distribution function and $\mathcal{N}_2$ is the standard bivariate density with correlation $\rho \in (0,1)$. The role of $\rho$ as a bandwidth will be explored shortly. The update \eqref{eq:DP_copdens} is then a mixture of the independent copula density and the Gaussian copula density, and the sequence $\alpha_i=\mathcal{O}\left(i^{-1}\right)$ ensures the update approaches the independent copula as $i \to \infty$. { Although asymptotic independence is not necessary for the martingale condition, this property holds for Bayesian sequences of copulas \citep{Hahn2018}, and is indeed important for frequentist consistency when estimating $p_n$ as we will see in Section \ref{sec:consistency}. We will see the specific suggested form of $\alpha_i$ at the end of this subsection.}

Note the similarity of the update in \eqref{eq:DP_copdens} to the generalized P\'olya urn  for the Dirichlet process, which for $c = 1$ has the update
${P}_{i+1}(y) =(1-\alpha_{i+1}){P}_i(y) + \alpha_{i+1}\,\mathbbm{1}(y_{i+1} \leq y)$.
We can thus interpret \eqref{eq:DP_copdens} as a smooth generalization of the Bayesian bootstrap update for continuous distributions. One can also interpret  \eqref{eq:DP_copdens} as a Bayesian kernel density estimate (KDE) that satisfies the c.i.d. condition, as the regular KDE cannot satisfy this condition \citep{West1991}. The update can be visualized in Figure \ref{fig:cop_illustration}, where for convenience we write $u_i = P_i(y), \, v_i = P_i(y_{i+1})$. The Gaussian copula kernel $c_\rho\left(u_i,v_i \right)p_i(y)$ is a data dependent kernel roughly centered at $y_{i+1}$, as shown in the left.  The kernel becomes sharper as $\rho$ increases, and we recover the Bayesian bootstrap in the limit of $\rho \to 1$ (with $\alpha_i = 1/i$). The update is then a mixture of $p_{i}(y)$ and the copula kernel, which gives us $p_{i+1}(y)$ in the right panel.

\begin{figure}[!h]
    \centering
        \includegraphics[width=0.95\textwidth]{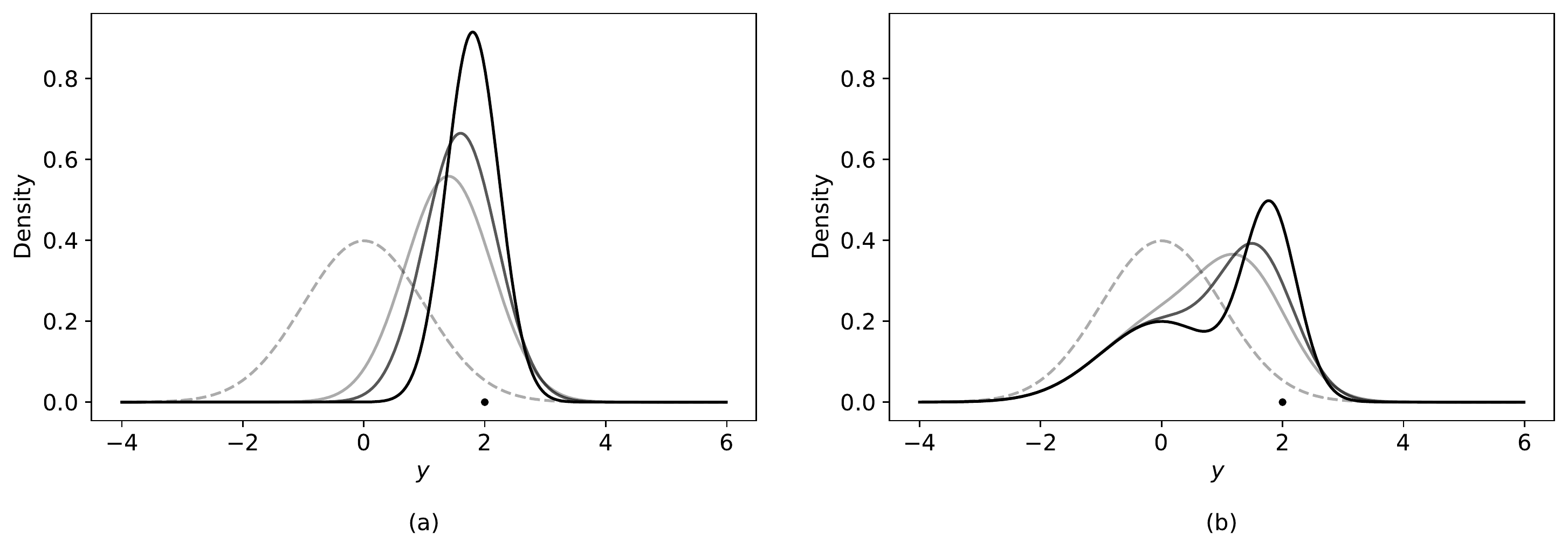}
    \caption{ Current predictive density $p_i(y)$ (\dashedlow) and new datum $y_{i+1}$ (\dot); (a) Copula kernel $c_\rho\left(u_i,v_i\right) p_{i}(y)$ for correlation $\rho =0.7,0.8,0.9$ (\fulllow,\fullmid,\full); (b) Corresponding updated predictive density $p_{i+1}(y)$ (\fulllow,\fullmid,\full) for $\alpha_{i+1} = 0.5$; note that we write $u_i = P_i(y), v_i = P_i(y_{i+1})$} \label{fig:cop_illustration}\vspace{-2mm}
\end{figure}

The recursive update was first introduced to compute $p_n(y)$, but properties of the update make it a highly suitable candidate for predictive resampling. Firstly, by Corollary \ref{corr:copula_CID}, this update is guaranteed to provide a c.i.d. sequence and hence satisfy the existence and unbiasedness conditions. Secondly, the update of the predictive distribution is online, and does not require an expensive recomputation of the predictive distribution at each step.  Finally, the predictive resampling update is particularly computationally elegant as $y_{i+1} \sim P_i(y)$ implies that $P_i(y_{i+1}) \sim \mathcal{U}[0,1]$, so all that is required is the simulation of uniform random variables. The forward sampling step then involves simulating $V_{i} \sim \mathcal{U}[0,1]$ and computing
\begin{equation}\label{eq:1d_copula_pr}
\begin{aligned}
{p}_{i+1}(y) &= \left[1-\alpha_{i+1}  + \alpha_{i+1}  c_\rho\left\{{P}_i(y),V_i\right\}\right] p_i(y) \\
{P}_{i+1}(y) &= (1-\alpha_{i+1}){P}_i(y) + \alpha_{i+1} H_\rho\left\{{P}_i(y),V_i\right\}
\end{aligned}
\end{equation}
iterated over $i \in \{n,\ldots,N\}$, which gives us a random ${p}_N(y)$ at the end. There is no need to actually sample $Y_{i+1} \sim P_i(y)$, which is possible but is more computationally expensive. In Section \ref{sec:theory}, we will see that this update form allows easy analysis of the theoretical properties of predictive resampling.

The bandwidth $\rho$ controls the smoothness of the density estimate, which we can set in a data-dependent manner as we show in Section \ref{sec:hyperparameters}. On the other hand, the sequence $\alpha_i$ is responsible for the uncertainty as we will see in Section \ref{sec:theory}, so extra care must be taken when eliciting this. \cite{Hahn2018} suggest the form $\alpha_i = (i+1)^{-1}$ inspired from the stick-breaking process of the posterior DP as in the Bayesian bootstrap, which works well for estimating $p_n(y)$ but we find this performs poorly when predictive resampling, giving too little uncertainty. This was also observed in \citet{Fortini2020} in the case of Newton's recursive method. However, it should be observed that the posterior over the mixing distribution $G$ is actually a mixture of DPs, that is
\begin{equation*}
\begin{aligned}
[G \mid \theta_{1:n},y_{1:n}] \sim \text{DP}\left(a + n, \frac{a G_0 + \sum_{i=1}^n \delta_{\theta_i}}{a + n}\right), \quad 
[\theta_{1:n} \mid y_{1:n}] \sim \pi(\theta_{1:n} \mid y_{1:n})
\end{aligned}
\end{equation*}
where $\pi(\theta_{1:n} \mid y_{1:n})$ is intractable. As shown in Appendix \ref{Appendix:alpha_deriv}, we only require the simplifying assumption of $\pi(\theta_{1:n} \mid y_{1:n}) = \prod_{i=1}^n G_0(\theta_i)$, which corresponds to each datum belonging to its own cluster in a similar spirit to the KDE. This then returns us the same copula update as \eqref{eq:DP_copdens} with 
\begin{equation}\label{eq:alpha_i}
\alpha_i = \left(2- \frac{1}{i} \right)\frac{1}{i+1}\cdot
\end{equation}
Intuitively, the additional mixing over $\theta_{1:n}$ results in the inflated value  compared to $\alpha_i=(i+1)^{-1}$. Note this is still $\mathcal{O}(i^{-1})$, matches with initial update step for $i = 1$, and works much better in practice as it approaches $0$ more slowly. We use this sequence for the remainder of the copula methods. 

\subsection{Multivariate case} 

In this section, we extend the univariate method to multivariate data $\mathbf{y} \in \mathbb{R}^d$, allowing us to both learn $p_n(\mathbf{y})$ recursively, and retain the c.i.d. sequence so we can  predictively resample to obtain uncertainty. Even without predictive resampling, a general multivariate density estimator $p_n(\mathbf{y})$ is of interest, as the KDE is known to perform poorly in high dimensions; see \citet{Wang2019} for a review. Computation for the multivariate DPMM \citep{Maceachern1994,Escobar1995,Neal2000} may scale poorly as the number of dimensions grows. Variational inference (VI) is a quicker approximation as demonstrated in \citet{Blei2006}, but there is strong dependence on the optimization procedure, which may impair performance in high dimensions. A copula method for bivariate data is suggested in the appendix of \cite{Hahn2018}, but it does not scale well with dimensionality and is not c.i.d.. A recursive method for multivariate density estimation is introduced in \citet{Cappello2018}, but numerical integration on a grid is still required, which scales exponentially with $d$, or a Monte Carlo scheme is required. \cite{Fortini2020} propose a multivariate extension of Newton's recursive method, but it also requires an approximate Monte Carlo scheme to evaluate the predictive density.

Extending the above argument in Corollary \ref{corr:copula_CID} to multivariate data is not as straightforward,  as we would like to factorize the joint density into $p_i(\mathbf{y}, \mathbf{y}_{i+1}) = k(\mathbf{y},\mathbf{y}_{i+1})p_i(\mathbf{y})p_i(\mathbf{y}_{i+1})$, which does not have the copula interpretation like in the 2-dimensional case. Furthermore, building high-dimensional copulas is a difficult task, and bivariate copulas are good building blocks for higher dimensional dependency \citep{Joe1996,Bedford2001,Aas2009}. 

\subsubsection{Factorized kernel}\label{sec:multivariate_copula}
With the above in mind, we now consider the first step update of a multivariate DPMM below
\begin{equation}\label{eq:mv_DP_mixture}
\begin{aligned}
f_G(\mathbf{y}) = \int \prod_{j=1}^d\mathcal{N}(y^j \mid \theta^j,1) \, dG(\bm{\theta}),\quad 
 G \sim \text{DP}\left(a, G_0 \right), \quad G_0(\bm{\theta}) = \prod_{j=1}^d\mathcal{N}(\theta^j \mid 0,\tau^{-1})
\end{aligned}
\end{equation}
where $y^j$ is the $j$-th dimension of $\mathbf{y}$, and likewise for $\theta^j$. Note the factorized normal kernel and independent priors for each $\theta^j$. From this, we see that we can factorize $p_0(\mathbf{y})=\prod_{j=1}^d p_0(y^j)$. It is shown in Appendix \ref{Appendix:multivariate} that the first update step takes on the form
\begin{equation*}\label{eq:mv_DP_firststep}
p_{1}(\mathbf{y}) = \left[1-\alpha_1 + \alpha_1 \prod_{j=1}^d c_\rho\left\{P_0(y^j),P_0(y_1^j)\right\}\right] p_0(\mathbf{y})
\end{equation*}
 where $y_i^j$ is the $j$-th dimension of the $i$-th data point. However, naively using this update for $i>1$ will result in the sequence $p_i(\mathbf{y})$ no longer satisfying the martingale condition in \eqref{eq:martingale_density}, and we also find that it performs poorly empirically. A simple but key extension allows us to retain the c.i.d. sequence:
\begin{equation}\label{eq:mv_DP_copdens}
p_{i+1}(\mathbf{y}) = \left\{1-\alpha_{i+1} + \alpha_{i+1} \prod_{j=1}^d c_\rho\left(u_i^{j},v_i^{j}\right)\right\} p_i(\mathbf{y})
\end{equation}
where
\begin{equation*}\label{eq:mv_DP_conditional_def}
\begin{aligned}
u_i^{j} &= P_i\left(y^j \mid  y^{1:j-1}\right), \quad v_i^{j} &= P_i\left(y_{i+1}^j \mid y_{i+1}^{1:j-1}\right).
\end{aligned}
\end{equation*}
The input to the bivariate normal copula is now the \textit{conditional} cumulative distribution function at $\mathbf{y}$ and $\mathbf{y}_{i+1}$ for a particular dimension ordering, and this change ensures many desirable properties. First, we can verify that the martingale condition \eqref{eq:martingale_density} now holds through a multivariate change of variables from $\mathbf{y}_{i+1}$ to $v_i^{1:d}$, so the c.i.d. condition is satisfied. By marginalizing $y^d, y^{d-1},\ldots, y^{k+1}$ in descending order, we also have that the marginals for a single ordering of dimensions has the same update
\begin{equation}\label{eq:mv_DP_marginal}
p_{i+1}\left({y}^{1:k}\right) = \left\{1-\alpha_{i+1} + \alpha_{i+1} \prod_{j=1}^k c_\rho\left(u_i^{j},v_i^{j}\right)\right\} p_i\left(y^{1:k}\right). 
\end{equation}
From this, we can update the conditional distribution functions via
\begin{equation}\label{eq:mv_DP_conditional}
\begin{aligned}
u_{i+1}^{k} =\left\{(1-\alpha_{i+1})u_i^k + \alpha_{i+1} H_\rho\left(u_i^k,v_i^k \right) \prod_{j=1}^{k-1} c_\rho\left(u_i^{j},v_i^{j}\right)\right\} \frac{p_i\left(y^{1:k-1}\right)}{p_{i+1}\left(y^{1:k-1}\right)}
\end{aligned}
\end{equation}
and likewise for $v_{i+1}^{k}$. As a result, all terms in the update \eqref{eq:mv_DP_copdens} can be computed tractably, with no need for numerical integration or approximations, allowing us to extend this method to any number of dimensions as computation complexity is linear in $d$. Notably, we must specify an ordering of the dimensions of $\mathbf{y}$, which at first may seem undesirable. However, it is not an assumption on dependence, and the only implication is that the subset of ordered marginal distributions continue to satisfy \eqref{eq:mv_DP_marginal}, that is a sort of marginal coherence. Interestingly, the form of \eqref{eq:mv_DP_marginal} suggests that $p_i\left(y^{1:k}\right)$ depends only on the first $k$ dimensions of $\mathbf{y}_{1:i}$. Practically, we find the dimension ordering makes little difference, and we recommend selecting the ordering such that any conditional or marginal distributions of interest remain tractable. In Appendix \ref{Appendix:copula_discrete} we provide an extension to the above for mixed-type data.

Predictive resampling again takes on a simple form due to the nature of the update \eqref{eq:mv_DP_copdens}. We can imagine drawing each dimension of ${\mathbf{Y}} \sim P_i(\cdot)$ in a sequential nature, that is 
\begin{equation}\label{eq:mv_simulation}
[Y^{1}] \sim P_i\left(y^1\right),\,\quad
[Y^{2}\mid y^1] \sim P_i\left(y^2 \mid y^{1}\right),\quad
\ldots,\quad
  [Y^d \mid {y}^{1:d-1}] \sim P_i\left(y^d\mid {y}^{1:d-1}\right).
\end{equation}
Denoting $V_i^j := P_i\left(Y^j \mid Y^{1:j-1}\right)$, we then have that $V_i^j \iid \mathcal{U}[0,1]$ for $j = \{1,\ldots, d\}$, which we can substitute into \eqref{eq:mv_DP_copdens} and \eqref{eq:mv_DP_conditional}, similar to the univariate case. Predictive resampling again only requires sampling $d$ independent uniform random variables for each forward step and computing the update.

\subsection{Regression}\label{sec:copula_regression}

We now consider extending the copula method and predictive resampling to the regression setting, where we have univariate $y_{i} \in \mathbb{R}$ (which can be easily extended to multivariate) with corresponding covariates $\mathbf{x}_{i} \in \mathcal{X}$, where for example $\mathcal{X} = \mathbb{R}^d$. We will later also consider binary regression, where $y_{i} \in \{0,1\}$. One assumption is that the covariates are random, where we write $\{y_i, \mathbf{x}_i\} \iid f_0(y,\mathbf{x})$, and we are interested in $f_0(y_i \mid \mathbf{x}_i)$. We term this the `joint method', as we infer the full joint $f_0(y_i,\mathbf{x}_i)$ from which the conditional then follows. Examples of this are \citet{Muller1996, Shahbaba2009,Hannah2011}, where the prior on $f_0(y_i,\mathbf{x}_i)$ is a DPMM.   The second type of assumption, which we call the `conditional method', is the more common framework. Here we assume that $\mathbf{x}_{1:n}$ are fixed design points and the randomness arises from the response $y_{1:n}$, so we infer a {family} of conditional densities $\{f_\mathbf{x}(y): {\mathbf{x} \in \mathcal{X}}\}$. The most common framework is the additional assumption of $y_i = g(\mathbf{x_i}) + \epsilon_i$, where $\epsilon_i$ are independent zero-mean noise, and a prior on the mean function $g$ is assumed, e.g. a Gaussian process \citep{Rasmussen2003}. Alternatively, one can elicit a prior on $\{f_\mathbf{x}(y): {\mathbf{x} \in \mathcal{X}}\}$ directly, for example with mixture models based on the dependent Dirichlet process \citep{Maceachern1999}.  We recommend \citet{Wade2013,Wade2014,Quintana2020} for thorough reviews.

\subsubsection{Joint method}\label{sec:jointreg}
The joint method follows easily from the multivariate: we first estimate the joint predictive density $p_{i+1}(y, \mathbf{x})$, then compute the conditional $p_{i+1}(y \mid \mathbf{x}) = p_{i+1}(y, \mathbf{x})/p_{i+1}(\mathbf{x})$. Utilizing \eqref{eq:mv_DP_marginal}, we have the tractable update for the conditional density
\begin{equation}\label{eq:jointreg_conditional}
\begin{aligned}
p_{i+1}(y \mid \mathbf{x}) =  p_{i}(y \mid \mathbf{x})  \frac{\left\{1- \alpha_{i+1}+ \alpha_{i+1} c_{\rho_y}(q_{i},r_{i})\prod_{j=1}^d  c_{\rho_x} \left(u_{i}^j,v_{i}^j\right)\right\}}{\left\{1- \alpha_{i+1}+ \alpha_{i+1} \prod_{j=1}^d  c_{\rho} \left(u_{i}^j,v_{i}^j\right)\right\}}
\end{aligned}
\end{equation}
where
\begin{equation}\label{eq:jointreg_cdf}
\begin{aligned} 
q_{i}&=P_{i}(y\mid \mathbf{x}), \quad &&r_{i} =P_{i}(y_{i+1} \mid \mathbf{x}_{i+1})\\  u_{i}^j &=P_{i}\left(x^j \mid x^{1:j-1}\right),  &&v_{i}^j = P_{i}\left(x_{i+1}^j \mid x_{i+1}^{1:j-1}\right).
\end{aligned}
\end{equation}
Here, we can have separate bandwidths for $y$ and $\mathbf{x}$, and even one for each dimension of $\mathbf{x}$. The updates for $q_{i+1},r_{i+1}, u_{i+1}^j,v_{i+1}^j$ are the same as in  \eqref{eq:mv_DP_conditional}, and again all terms are tractable. Predictive resampling in this case requires simulating both $\{Y,\mathbf{X}\} \sim P_i(y,\mathbf{x})$ 
just like in \eqref{eq:mv_simulation}.

\subsubsection{Conditional method}\label{sec:conditreg}
When $\mathbf{x}$ is high-dimensional, it may be cumbersome to model $p_n(\mathbf{x})$ when we are only interested in the conditional density. The conditional method models $p(y \mid \mathbf{x})$ directly, and we turn to the dependent Dirichlet process (DDP) and its extensions for inspiration. In particular, consider the general covariate-dependent stick-breaking mixture model 
\begin{equation}\label{eq:DDP_mixture}
\begin{aligned}
f_{G_\mathbf{x}}(\mathbf{y}) = \int \mathcal{N}(y \mid \theta,1) \, dG_{\mathbf{x}}({\theta}), \quad
G_\mathbf{x} =\sum_{k=1}^\infty w_k(\mathbf{x})\,\delta_{\theta_k^*}
\end{aligned}
\end{equation}
where $w_k(\mathbf{x})$ follows an $\mathbf{x}$-dependent stick-breaking process, and $\theta_k^* \iid \mathcal{N}(\theta \mid 0, \tau^{-1})$. A full derivation is provided in Appendix \ref{Appendix:copula_regression}. We can show that the  update step of the predictive takes the form
\begin{equation}\label{eq:conditreg_conditional}
\begin{aligned}
p_{i+1}(y \mid \mathbf{x}) = \left\{1-\alpha_{i+1}(\mathbf{x},\mathbf{x}_{i+1})+ \alpha_{i+1}(\mathbf{x},\mathbf{x}_{i+1})\, c_{\rho_y}\left(q_i,r_i\right)\right\} p_i(y\mid \mathbf{x})
\end{aligned}
\end{equation}
where $\alpha_1(\mathbf{x},\mathbf{x}') = \sum_{k=1}^\infty E\left[ w_k(\mathbf{x})w_k(\mathbf{x}') \right]$, $\rho_y = 1/(1+\tau)$ and $q_i,r_i$ are as in \eqref{eq:jointreg_cdf}. The term $\alpha_1(\mathbf{x},\mathbf{x}')$ is tractable for some choices of the construction of $w_k(\mathbf{x})$, e.g. the kernel stick-breaking process \citep{Dunson2008}. Unfortunately this does not provide guidance on how to generalize to $\alpha_{i}(\mathbf{x},\mathbf{x}')$. Instead, we turn to the joint copula method in the previous section for inspiration, which can be written as \eqref{eq:conditreg_conditional} with \vspace{-2mm}
\begin{equation*}
\alpha_{i}(\mathbf{x},\mathbf{x}') = \frac{\alpha_{i}\prod_{j=1}^d  c_{\rho_x} \left(u_{i-1}^j,v_{i-1}^j\right)}{1- \alpha_{i} + \alpha_{i}\prod_{j=1}^d  c_{\rho_x} \left(u_{i-1}^j,v_{i-1}^j\right)}\cdot
\end{equation*}
This form of $\alpha_{i}(\mathbf{x},\mathbf{x}')$ can be viewed as a distance measure between $\mathbf{x}$ and $\mathbf{x}'$ that is dependent on $P_n(\mathbf{x})$ which is updated in parallel. To avoid modelling $P_n(\mathbf{x})$, we can simplify the above and consider the following as a distance function directly:
 \begin{equation}\label{eq:condit_alpha}
\alpha_{i}(\mathbf{x},\mathbf{x}') = \frac{\alpha_{i}\prod_{j=1}^d  c_{\rho_{x^j}} \left\{\Phi\left(x^j\right),\Phi\left(x'^j\right)\right\}}{1- \alpha_{i} + \alpha_{i}\prod_{j=1}^d   c_{\rho_{x^j}} \left\{\Phi\left(x^j\right),\Phi\left(x'^j\right)\right\}}
\end{equation}
which is equivalent to the joint method but leaving $P_n(\mathbf{x}) = P_0(\mathbf{x})$ without updating, providing us an increase in computational speed. This form requires $\mathbf{x}_{1:n}$ to be standardized for good performance, and we find that specifying independent bandwidths for each dimension in $\mathbf{x}$ works well. This method is similar to the normalized covariate-dependent weights of \cite{Antoniano2014}.
 
If $\mathbf{x}_{1:n}$ is indeed a subsequence of a deterministic sequence of design points $\mathbf{x}_1,\mathbf{x}_2,\ldots$, then predictive resampling simply involves selecting $\mathbf{x}_{i}$ for $i > n$ from this  sequence, and drawing $[Y_{i+1} \mid \mathbf{x}_{i+1}] \sim P_i(y \mid \mathbf{x}_{i+1})$. If $\mathbf{X}_{1:n}$ is actually random and we have chosen the conditional approach simply for convenience, then we can draw the future $\mathbf{X}_{n+1:N}$ from the sequence of empirical predictives as in the Bayesian bootstrap. We have however noticed some numerical sensitivity to this choice of $P_n(\mathbf{x})$ in the uncertainty in $p_n(y \mid \mathbf{x})$ for $\mathbf{x}$ far from the observed dataset; this is illustrated in Appendices \ref{Appendix:moon} and \ref{Appendix:lidar}.  Once again, conditional on $\mathbf{X}_{i+1} = \mathbf{x}_{i+1}$,  we have that $ P_i(Y_{i+1} \mid \mathbf{x}_{i+1}) \sim \mathcal{U}[0,1]$, so predictive resampling only consists of simulating independent uniform random variables and updating. An example of using the Bayesian bootstrap for the covariates is provided in Appendix \ref{Appendix:lidar}.

\subsubsection{Classification}
For classification, both the joint and conditional approach generalize easily to when $y_i \in \{0,1\}$. To this end, we can derive the copula update for a beta-Bernoulli mixture. As shown in  Appendix {\ref{Appendix:copula_classification}}, this gives 

\begin{equation}
\begin{aligned}
d_{\rho_y}\{q_i,r_i\} &= \begin{dcases} 1-\rho_y + \rho_y\,\frac{q_i\wedge r_i}{q_i\, r_i} &\quad \text{if } y = y_{i+1} \vspace{2mm}\\  1-\rho_y + \rho_y\,\frac{q_i - \{q_i\wedge (1-r_i)\}}{q_i\, r_i}&\quad \text{if } y \neq y_{i+1} \end{dcases}
\end{aligned}
\end{equation}
where $q_i = p_i(y \mid \mathbf{x}),r_i = p_i(y_{i+1}\mid \mathbf{x}_{i+1})$ and $\rho_y \in (0,1)$. We can simply replace the bivariate Gaussian copula density $c_{\rho_y}(q_i,r_i)$ in \eqref{eq:jointreg_conditional} and \eqref{eq:conditreg_conditional} with $d_{\rho_y}(u_i,v_i)$. One can check that $q_i$ is indeed a martingale when predictive resampling, and forward sampling can be done directly as drawing binary $Y_{n+1}$ from the Bernoulli predictive is straightforward. Unfortunately, we do not have the useful property of $P_i(y_{i+1}) \sim \mathcal{U}[0,1]$ in the discrete case, so predictive resampling beyond the Bayesian bootstrap for $\mathbf{X}_{n+1:N}$ is computationally expensive at $\mathcal{O}(N^2)$, or approximation via a grid is required. The Bayesian bootstrap for $\mathbf{X}_{n+1:N}$ is still feasible as we only need to compute $p_N(y \mid \mathbf{x})$ at the observed $\mathbf{x}_{1:n}$. An example of this method is provided in Appendix \ref{Appendix:moon}.

\subsection{Practical considerations}\label{sec:practicalcons}
In this subsection, we discuss some practical considerations.  Further details, such as those regarding sampling and optimization, are given in Appendix \ref{Appendix:practicalcons}.

\subsubsection{Initial density}
 For the copula methods, we require an initial guess $p_0(\mathbf{y})$ to begin our recursive updates, which can contain prior information. As it is a statement on observables, it is easier to elicit than a traditional Bayesian prior. In practice, we recommend standardizing each variable in the data  $y^j_{1:n}$ to have mean $0$ and variance $1$ and using the default initialization $\mathcal{N}(y^j \mid 0,1)$ for each dimension in an empirical Bayes fashion. For discrete variables, a suitable default choice is the uniform distribution over the classes. Finally, in the regression case, we can include prior information on the regression function, e.g.
$p_0(y \mid \mathbf{x}) = \mathcal{N}( y \mid \beta^\T \mathbf{x},1)$. However,  $p_0(y \mid \mathbf{x}) = \mathcal{N}( y \mid 0,1)$ tends to work well as a default choice. 

\subsubsection{Hyperparameters}\label{sec:hyperparameters}
As we recommend the fixed form of $\alpha_i$ in \eqref{eq:alpha_i}, the only hyperparameter in the copula update is the constant $\rho$ which parameterizes the bivariate normal copula in \eqref{eq:Gauss_copdens}.  While \citet{Hahn2018} suggest a default choice for $\rho$, we prefer a data-driven approach. Fortunately, there is an obvious method to select $\rho$ using the prequential log score of \citet{Dawid1984}, that is to maximize $\sum_{i=1}^{n} \log p_{i-1}(\mathbf{y}_{i})$ for density estimation or $\sum_{i=1}^{n} \log p_{i-1}({y}_{i} \mid \mathbf{x}_{i})$ for regression, which is related to a cross-validation metric \citep{Gneiting2007,Fong2020}. This fits nicely into our simulative framework, as $\rho$ is selected on how well the sequence of predictives forecasts consecutive data points, which then informs us on the future predictives for predictive resampling. We can also specify a separate $\rho_j$ for each dimension, which corresponds to differing length scales for the update from each conditional distribution. For optimization, gradients with respect to $\rho$  can be computed quickly using automatic differentiation.

\subsubsection{Permutations}
Due to our relaxation of exchangeability in Section \ref{sec:pred_coherence}, one downside to the copula update and c.i.d. sequences in general is the dependence of $p_n$ on the permutation of $y_{1:n}$ when there is no natural ordering of the data. For permutation invariance, we can average $p_n$ and the corresponding prequential log-likelihood over $M$ random permutations of $y_{1:n}$. We find in practice that $M = 10$ is sufficient, which is computationally feasible for moderate $n$ due to the speed of the copula update, and the method is also parallelizable over permutations. For predictive resampling, we then begin with the permutation averaged $p_n$ and forward sample with the copula update. From asymptotic exchangeability in Theorem \ref{Th:cid_exchangeable} in Section \ref{sec:theory_mv}, averaging over permutations is not required for forward sampling provided $N$ is chosen to be sufficiently large. Theoretical properties of permutation averaging are explored in  \citet{Tokdar2009,Dixit2019}, which we do not consider here.

\subsubsection{Computational complexity}
For computing $p_n(\mathbf{y})$ in the multivariate copula method, there is an overhead of first computing $v_i^j$ for $j \in \{1,\ldots,d\}$, $i \in \{0,\ldots,n-1\}$ using \eqref{eq:mv_DP_conditional}, which requires $\mathcal{O}\left(n^2 d \right)$ operations, followed by $\mathcal{O}\left(n d \right)$ operations to compute $p_n(\mathbf{y})$ at a single $\mathbf{y}$ (which is then parallelizable).  After computing $p_n(\mathbf{y})$, predictive resampling $N$ future observables requires $\mathcal{O}\left(Nd\right)$ for each sample of $p_N(\mathbf{y})$; this is fully parallelizable across test points and posterior samples. Interestingly, we first compute $p_n(\mathbf{y})$ and only predictively resample after if uncertainty is desired, allowing for large computational savings if we are only interested in prediction.  The regression methods have a similar computational cost.

\section{Illustrations}\label{sec:illustrations}

In this section, we demonstrate the martingale posteriors induced by the copula methods of the previous section. Code for all experiments is available online at {\href{https://github.com/edfong/MP}{\texttt{https://github.com/edfong/MP}}}. We will demonstrate the copula method on examples where $\theta_0$ is the density itself or the loss function induces a simple parameter, e.g. quantiles. However, any $\theta_0$ of interest (as in Section \ref{sec:loss_functions}) can technically be computed directly from the density or from $y_{1:n}$ and samples of $Y_{n+1:\infty}$, although this may require a high-dimensional grid or relatively expensive sampling. As a result, for cases with complex loss functions that do not rely on the smoothness of $F_\infty$ (e.g. a parametric log-likelihood), we recommend the Bayesian bootstrap instead as a computationally efficient predictive resampling approach. For examples regarding the Bayesian bootstrap, we refer the reader to  the references in Section \ref{sec:BB}, and we qualitatively compare the Bayesian bootstrap and the copula methods in Section \ref{sec:discussion}.

For all examples, we follow the recommendations of Section \ref{sec:practicalcons} for $P_0$ and averaging over permutations. We will  demonstrate the monitoring of convergence to $P_\infty$, but we set $N = n + 5000$ as a standard default for the number of forward samples, where $n$ is the size of the dataset. All copula examples are implemented in JAX \citep{Jax2018}, which is a Python package popular in the machine learning community. JAX is ideal for our copula  updates: its just-in-time compilation facilitates a dramatic speed-up for our iterative updates especially on a GPU,  and its efficient automatic differentiation allows for quick hyperparameter selection. Note that the first execution of code induces an overhead compilation time of between 10-20 seconds for all examples. We carry out all copula experiments on an Azure NC6 Virtual Machine, which has a one-half Tesla K80 GPU card. The copula methods consist of many parallel simple computations on a matrix of density values, which is very suitable for a GPU, unlike traditional MCMC. The DPMM with MCMC examples are implemented in the \texttt{dirichletprocess} package \citep{Ross2018}, which utilizes Gibbs sampling.  Other benchmarks are implemented in \texttt{sklearn} \citep{Pedregosa2011}. Unless otherwise stated, default hyperparameter values are set for baselines. As the baseline packages are designed for CPU usage, we run them on a 2.6 GHz 6-Core Intel Core i7-8850H CPU. Further details can be found in Appendix \ref{Appendix:baselines}.

\subsection{Density estimation}

\subsubsection{Univariate Gaussian mixture model}
We begin by demonstrating the validity of the martingale posterior uncertainty returned from predictive resampling by comparing to a traditional DPMM in a simulated example, where the true density is known. We also discuss the monitoring of convergence of predictive resampling. For the data, we simulate $n = 50$ and $n = 200$ samples from a  Gaussian mixture model:
$$
f_0(y) = 0.8 \, \mathcal{N}(y \mid -2,1) + 0.2 \, \mathcal{N}(y \mid 2,1).
$$
For all plots, we compute the copula predictive $p_n(y)$ on an even grid of size $160$. Figures \ref{fig:toy_gmm50} and \ref{fig:toy_gmm200} show the martingale posterior density using the copula method for $n = 50$ and $n = 200$ respectively, compared to the traditional DPMM of \cite{Escobar1995} with MCMC. We draw $B = 1000$ samples for both methods. We see that the resulting uncertainty and posterior means are comparable between the copula and DPMM, and the uncertainty decreases as $n$ increases. The true density is largely contained within the 95\% credible intervals.

For predictive resampling with the copula method, we judge convergence by considering the $L_1$ distance between the forward sampled $p_N$ and initial $p_n$. This is demonstrated in Figure \ref{fig:gmm_conv} for a single forward sample for $n = 50$. On the left, we have a numerical estimate of $\lVert p_N - p_{n} \rVert_1$ which converges to a constant, and likewise for $\lVert P_N - P_{n} \rVert_1$ on the right, where $\lVert \cdot \rVert_1$ is the $L_1$ norm and is computed on the grid.  We see in this example that $N = n+5000$ is sufficiently large for $p_N$ to approximate  $p_\infty$. When we are not plotting on a grid and instead predicting over some test set, we may instead monitor 

$$\frac{1}{n_\mathrm{test}}\sum_{i=1}^{n_{\mathrm{test}}}|p_N({y_i}) - p_n({y_i})|.$$

Optimization of the prequential log-likelihood gives us the optimal hyperparameter $\rho = 0.77$ and $0.78$ for $n= 50$ and $200$ respectively. The prequential log-likelihood is returned easily from the copula method, allowing for easy hyperparameter selection. However, computing the marginal likelihood for the DPMM is non-trivial, and thus setting the hyperparameters of the priors in a data-driven way, that is empirical Bayes, remains a difficult task. Here, we select the DPMM hyperparameters to match the smoothness of the posterior mean of the copula method for comparability of the uncertainty.

\begin{figure}[!h]
    \centering
        \includegraphics[width=0.97\textwidth]{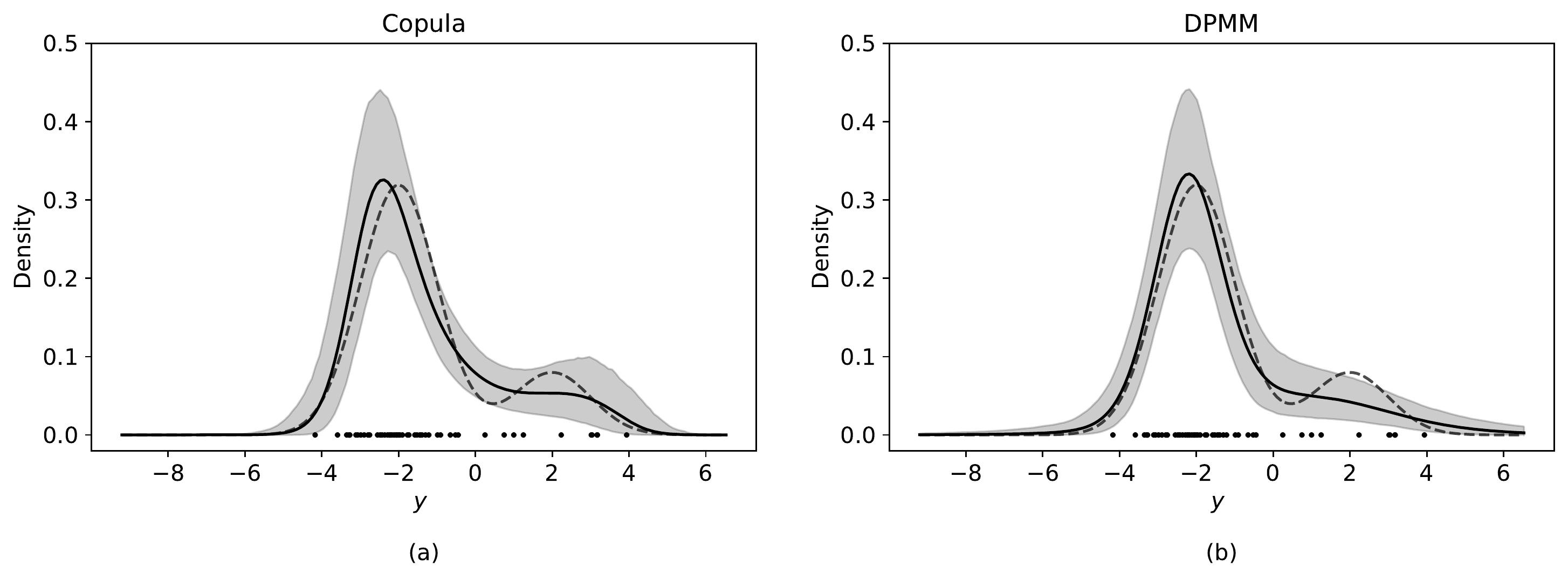}\vspace{-3mm}
    \caption{Posterior mean (\full) and 95\% credible interval (\sqrlow) of   (a) $p_N(y)$ for the copula method  and \\ (b) $p_\infty(y)$ for the DPMM, for $n = 50$ with true density (\dashedmid) and data (\dot)} \label{fig:toy_gmm50}
\end{figure}\vspace{-4mm}
\begin{figure}[!h]
    \centering
        \includegraphics[width=0.97\textwidth]{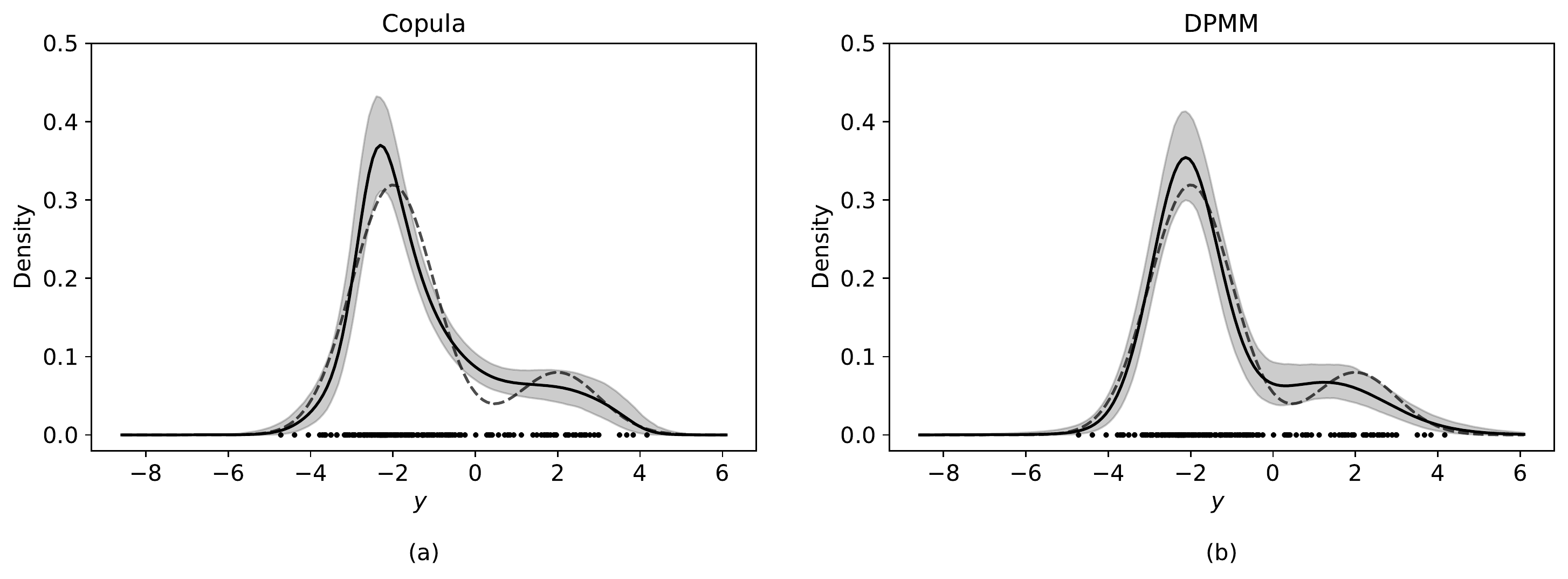}\vspace{-3mm}
    \caption{Posterior mean (\full) and 95\% credible interval (\sqrlow) of   (a) $p_N(y)$ for the copula method  and \\ (b) $p_\infty(y)$ for the DPMM, for $n = 200$ with true density (\dashedmid) and data (\dot)} \label{fig:toy_gmm200}
\end{figure}
\begin{figure}[!h]
    \centering
        \includegraphics[width=0.97\textwidth]{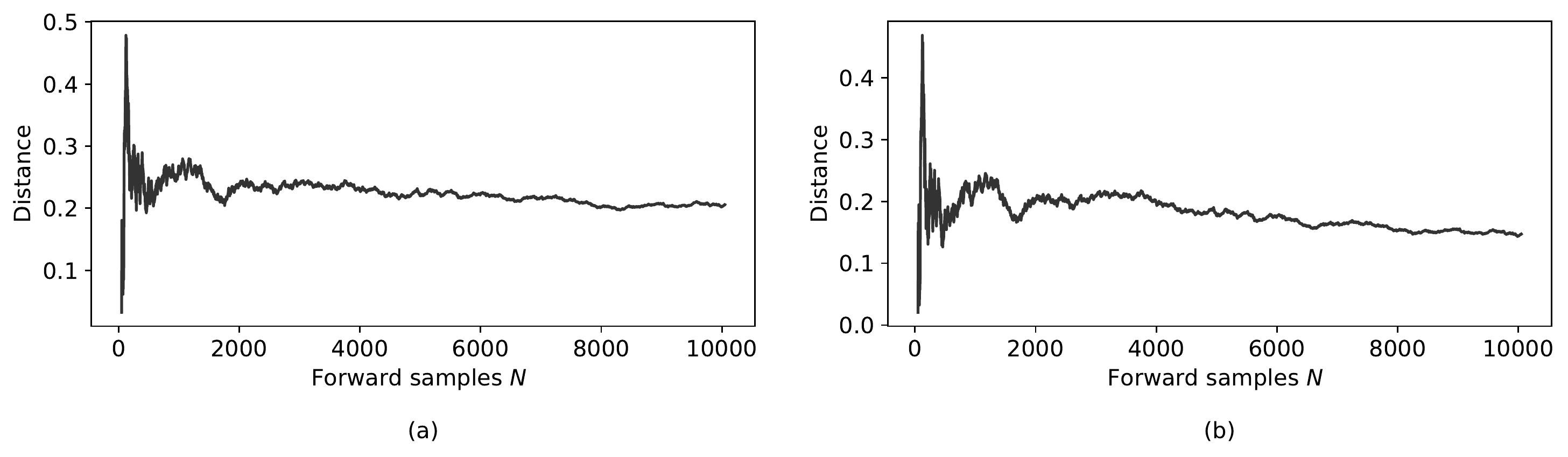}
    \caption{Estimated $L_1$ distance (a) $\lVert p_N - p_{n} \rVert_1$  and (b) $\lVert P_N - P_n \rVert_1$   for a single forward sample for $n = 50$} \label{fig:gmm_conv}
\end{figure}

\subsubsection{Univariate galaxy dataset}
We now demonstrate the martingale posterior sampling of a parameter of interest that requires a smooth density, through predictive resampling and the computation of $\theta(P_N)$. We analyze the classic `galaxy' dataset \citep{Roeder1990}, thereby extending the example of \citet{Hahn2018} to the predictive resampling framework. The dataset consists of $n = 82$ velocity measurements of galaxies in the Corona Borealis region.  For all plots, we compute $p(y)$ on an even grid of size $200$, and unnormalize after the copula method so that the scale of $y$ is in km/sec.

Figure \ref{fig:galaxy_dens} compares predictive resampling with the  copula method for $B = 1000$	posterior samples of $p_N$, where the selected bandwidth is $\rho = 0.93$. The bandwidth for KDE was computed through 10-fold cross-validation, and DPMM hyperparameters are set to the suggested values in \cite{West1991}. The 95\% credible intervals and posterior mean of the copula approach are comparable with that of the DPMM. Excluding compilation times, the optimization for $\rho$ and computation of $p_n(y)$ on the grid of size $200$ took 0.5 seconds, and predictive resampling took 2 seconds. In comparison,  DPMM with MCMC took 25 seconds for the same number of samples ($B = 1000$), where the samples are not independent; the plots for MCMC are thus produced with $B = 2000$.  Given this random density, we can also compute the statistics of interest $\theta$ directly from the grid of density values. Martingale posterior samples of the number of modes and 10\% quantiles of the random density are shown in Figure \ref{fig:galaxy_modes}, with comparison to the DPMM. Here the copula method tends to prefer 4 modes, whereas the DPMM prefers 5.

\begin{figure}[!h]
    \centering
        \includegraphics[width=0.97\textwidth]{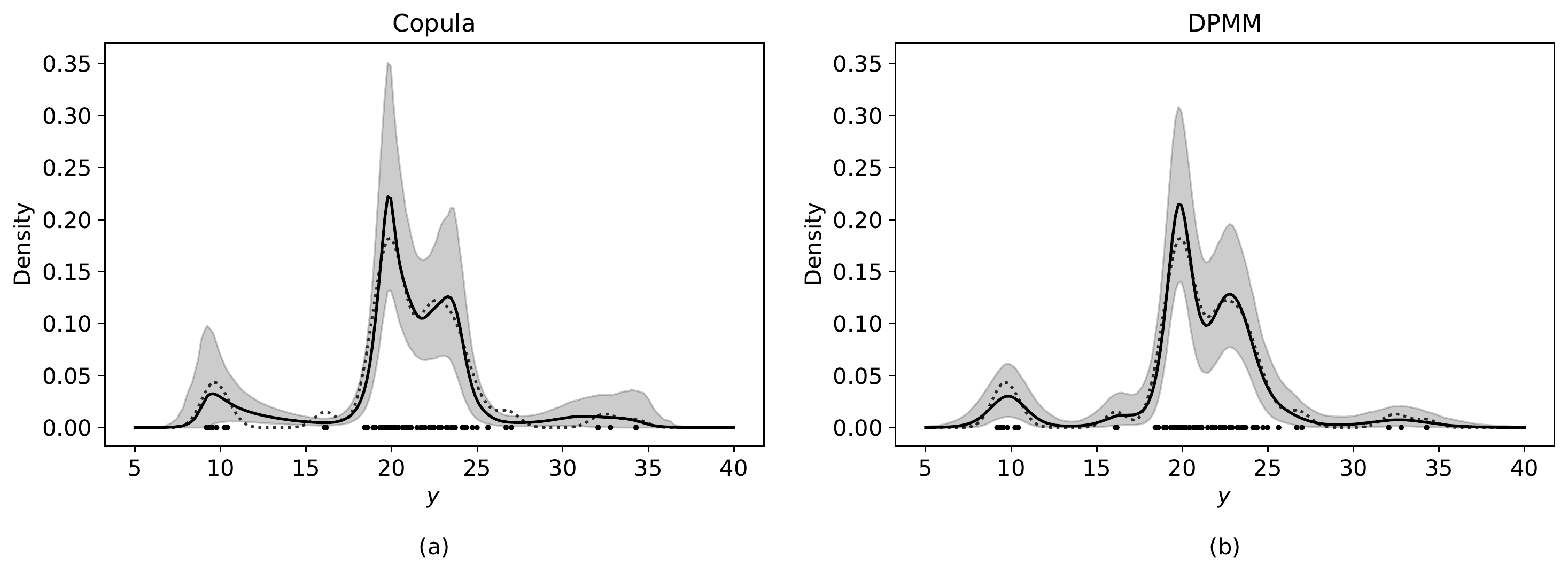}
    \caption{Posterior mean (\full) and 95\% credible interval (\sqrlow) of (a) $p_N(y)$ for the copula method  and \\ (b) $p_\infty(y)$ for the DPMM, with KDE (\dottedmid) and data (\dot)} \label{fig:galaxy_dens}
\end{figure}
\begin{figure}[!h]
    \centering
        \includegraphics[width=0.97\textwidth]{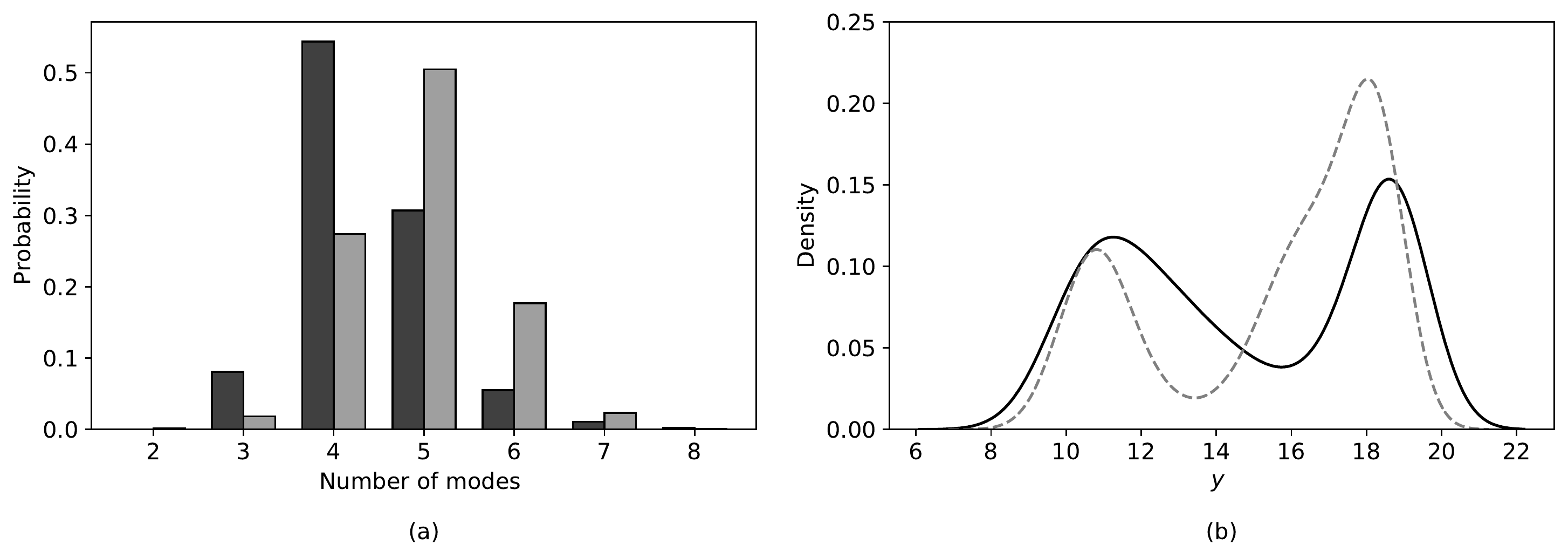}
    \caption{(a) Posterior samples of number of modes for the copula method (\sqr) and DPMM (\sqrmid); \\ (b) Posterior density of 10\% quantiles for the copula method (\full) and the DPMM (\dashedmid)} \label{fig:galaxy_modes}
\end{figure}

\subsubsection{Bivariate air quality dataset}
We demonstrate the martingale posterior for bivariate data using the method of Section \ref{sec:multivariate_copula}, which has large computational gains over posterior sampling with DPMM when the density is of interest, where the latter is expensive due to dimensionality. For this, we look at the `airquality' dataset \citep{Chambers2018}  from \texttt{DPpackage}. The dataset consists of daily ozone and solar radiation measurements in New York,  with $n = 111$ completed data points. For all plots, we compute $p_n(\mathbf{y})$ on a grid of size $25 \times 25$.

We fit the multivariate copula method of Section \ref{sec:multivariate_copula} with one bandwidth per dimension, and optimizing the prequential log-likelihood returns $\rho =[0.47,0.82]$. Predictive resampling $B = 1000$ martingale posterior samples returns us the martingale posterior mean and standard deviation of the bivariate density as shown in Figure \ref{fig:ozone_dens}. Again excluding compilation times, the optimization for $\rho$ and computation of $p_n(y)$ on the grid of size $625$ took 1 second, and predictive resampling took 10 seconds in total. For comparison,  the DPMM with MCMC required 4 minutes for the same number of samples. Further details and comparisons to the DPMM are given in Appendix \ref{Appendix:ozone}. 

Figure \ref{fig:ozone_conv} plots a martingale posterior sample of the density, with the corresponding $L_1$ distance convergence plot. We see that $N = 5000$ is again sufficient, which suggests a dimension independent convergence rate of $P_N \to P_\infty$. This is justified in the theory in Section \ref{sec:theory}. 

\begin{figure}[!h]
    \centering
        \includegraphics[width=0.97\textwidth]{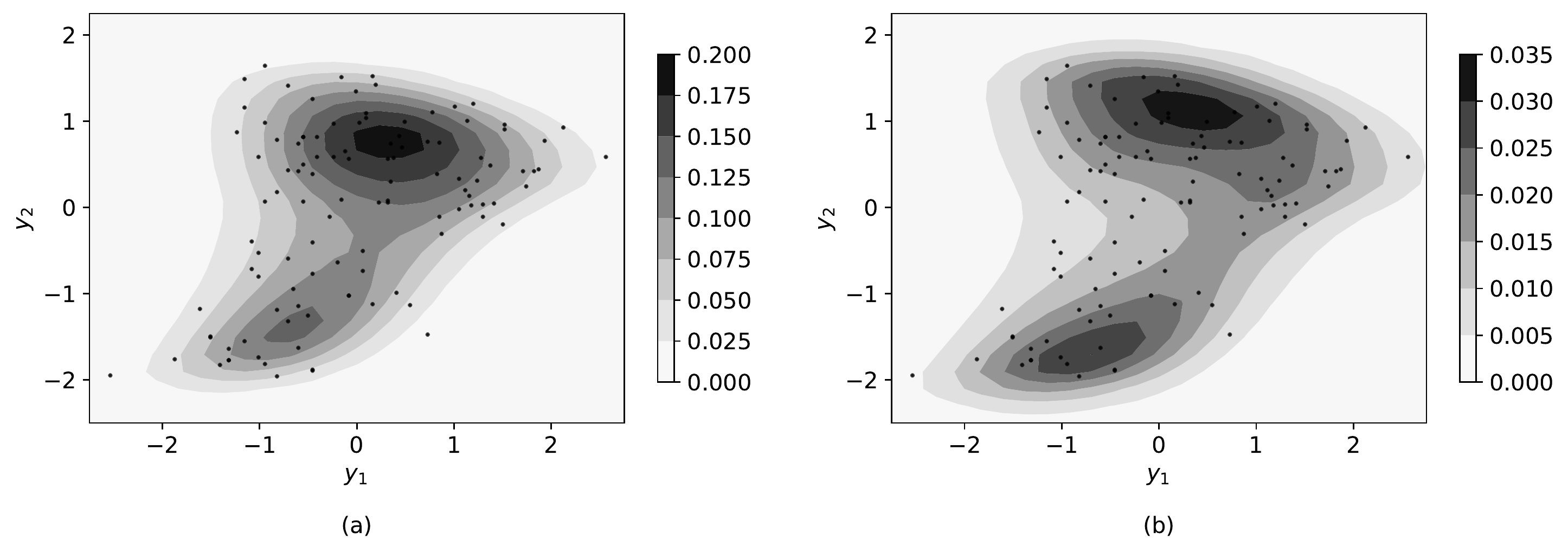}
    \caption{Posterior (a) mean and (b) standard deviation of $p_N(\mathbf{y})$ for the copula method with scatter plot of data (\dot)} \label{fig:ozone_dens}
\end{figure}

\begin{figure}[!h]
    \centering
        \includegraphics[width=0.97\textwidth]{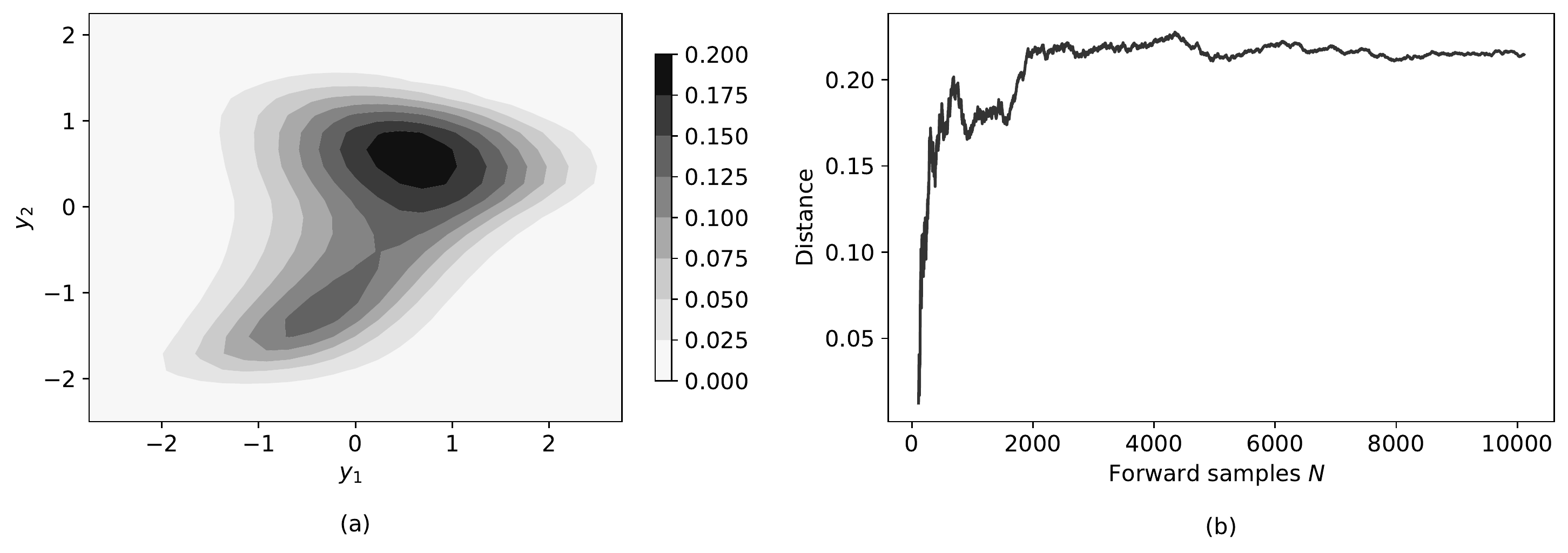}
    \caption{(a) Random sample of $p_N(\mathbf{y})$; (b) Corresponding estimated $||p_N - p_n ||_{1}$ } \label{fig:ozone_conv}
\end{figure}

\subsubsection{Multivariate UCI datasets}
In this section, we demonstrate the multivariate copula method of Section \ref{sec:multivariate_copula} as a highly effective  density estimator compared to the usual DPMM, as we do not need to deal with the posterior sampling or integration over high-dimensional parameters. We demonstrate on multivariate datasets from the UCI Machine Learning Repository \citep{Dua2019}. To prevent misleadingly high density values, we remove non-numerical variables and one variable from any pairs with Pearson correlation coefficient greater than $0.98$ (e.g. see \citet{Tang2012}). We compare to the KDE, DPMM and multivariate Gaussian, and evaluate the methods with a 50-50 test-train split and average the test log-likelihoods over 10 random splits. 

For the copula method, we use a single value of $\rho$ for all dimensions for a fair comparison to the KDE. We find that having distinct $\rho_{1:d}$ slightly improves predictive performance at the cost of  higher optimization times.  For the KDE, we use a single scalar bandwidth set through 10-fold cross-validation. For the DPMM, we set the Gaussian kernel to have diagonal covariance matrices and use VI \citep{Blei2006}. Using a full covariance matrix kernel is unreliable likely due to local optima for VI, and MCMC is too computationally expensive for large $d$. For the multivariate Gaussian, we use the empirical mean and covariance.

\begin{table}[!h]
\small
\begin{center}
\begin{tabular}{ l cccccc }
 Dataset & $n$ & $d$ &Gaussian & KDE &   DPMM (VI)  & Copula \\
 \hline 
 Breast cancer & $569$ &$26$ &$-17.8$ $(0.61)$&$-25.6$ $(0.29)$&$-33.4$ $(0.80)$&$\textbf{-13.0}$ $(0.26)$ \\  
 Ionosphere & $351$ &$32$ &$-49.4$ $(1.97)$&$-32.3$ $(0.79)$&$-36.5$ $(0.59)$   & $\textbf{-21.5}$ $(1.63)$   \\
  Parkinsons & $195$ & $16$&$-14.3$ $(0.54)$&$-15.6$ $(0.41)$&$-25.7$ $(0.92)$&$\textbf{\phantom{1}-9.9}$ $(0.28)$ \\
   Wine &$178$  &$13$ &$-16.1$ $(0.26)$& $-15.7$ $(0.20)$ & $-22.8$ $(0.61)$&$\textbf{-14.6}$ $(0.17)$    \\
\end{tabular}
\caption{Average test log-likelihood, standard errors (in brackets) and best performance in bold} \label{tab:mv_test_loglik}
\end{center}
\end{table}

\vspace{-5mm}
As shown in Table \ref{tab:mv_test_loglik}, the performance is significantly better on test data for these datasets. The better performance than the KDE is likely due to the regularizing effect of $p_0(\mathbf{y})$, which is important here as $n$ is only of moderate size. The DPMM (VI) likely performs poorly as the diagonal covariance cannot capture dependent structure, and the number of variational parameters is still high so optimization is difficult. We provide a more detailed analysis of the degradation in performance with dimensionality of the DPMM with VI in Appendix \ref{Appendix:highd_gmm}, where the copula method remains robust to dimensionality.

Overall, the run-times for the copula method, KDE and DPMM (VI) are similar, all of which are orders of magnitude faster than the DPMM with MCMC. For a single train-test split, the slowest example of the above (Breast cancer) for the copula method required less than 4 seconds in total to optimize $\rho$, while computing the overhead $v_i^j$ and predicting on the test data required less than 100ms. For the same example, the KDE and DPMM (VI) required around 1.5 and 6 seconds respectively.

\subsection{Regression and classification}

\subsubsection{Regression in LIDAR dataset}
We now demonstrate the joint copula regression method of Section \ref{sec:jointreg} on a non-linear heteroscedastic regression example, where the copula method performs well off-the-shelf.  We use the LIDAR dataset from \citet{Wasserman2006}, which consists of $n = 221$ observations of the distance travelled by the light and the log ratio of intensity of the measured light from the two lasers; the latter is the dependent variable. For the plots below, we evaluate the conditional density on a $y,x$ grid of $200 \times 40$ points.

For the copula method, we optimize the prequential conditional log-likelihood over the $M = 10$ permutations, and get $\rho_y = 0.90, \rho_x = 0.83$. The \textit{predictive} mean and 95\% central interval of $p_n(y \mid x)$ are shown in Figure \ref{fig:lidar_jdpmm}, compared to the DPMM, and we observe that the copula methods handle the nonlinearity better. The optimization, fitting and prediction on the grid took under 4 seconds for the copula method, compared to 5 minutes for the DPMM with MCMC for the same number of samples.

In Figure \ref{fig:lidar_jdpmm_x0}, we see martingale posterior samples of $p_N(y\mid x= 0)$ for the copula method compared to the DPMM. For reference, predictive resampling the $B=1000$ martingale posterior samples on  the $y$ grid for a single $x$ took under 3 seconds. One can see in Figure \ref{fig:lidar_jdpmm_x0} that there is more posterior uncertainty in the density $p_N(y \mid x = 0)$ for the copula methods, as the DPMM has a simpler mean function (weighted sum of linear). Convergence of the conditional density under predictive resampling is now dependent on the value of $x$. Figure \ref{fig:lidar_conv}(b) shows the $L_1$ distances as before for $x = 0$; however, we find that more forward samples are needed for $x$ far from the data. Figure \ref{fig:lidar_jdpmm_x-3} then shows martingale posterior samples of $p_N(y\mid x= -3)$ where $x$ is far from the data, and we see that both the copula and DPMM method have larger uncertainty as expected. However, predictive resampling for the {conditional} copula method of Section \ref{sec:conditreg} does not always demonstrate this desirable behaviour for outlying $x$; the joint and conditional methods are compared in Appendix \ref{Appendix:lidar} and this undesirable behaviour is noted in the next experiment.

One may also be interested in the uncertainty in  a point estimate for the function which we write as $\theta_x$, in this case the conditional median. In Figure \ref{fig:lidar_conv}(a), we plot the martingale posterior mean and 95\% credible interval of the conditional median of $P_N(y \mid x)$, where we see the uncertainty increasing with $x$ . Here we predictively resample on a $y,x$ grid of size $40 \times 40$ and compute the median numerically; this took 12 seconds for $B = 1000$ samples.  

\begin{figure}[!h]
    \centering
        \includegraphics[width=0.97\textwidth]{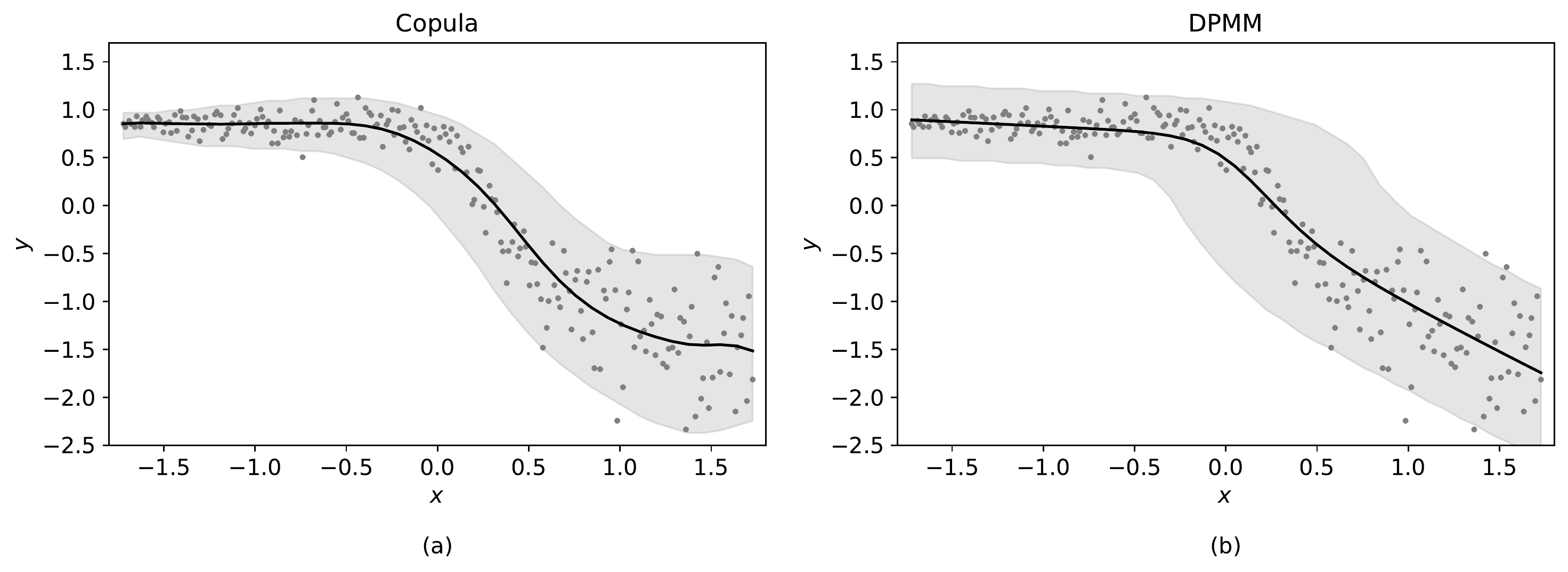}\vspace{-3mm}
    \caption{ $p_n(y \mid x)$ (\full) with 95\% predictive interval (\sqrlow) for the (a) joint copula method  and (b) joint DPMM, with data (\dotmid)} \label{fig:lidar_jdpmm}
\end{figure}
\begin{figure}[!h]
    \centering
        \includegraphics[width=0.97\textwidth]{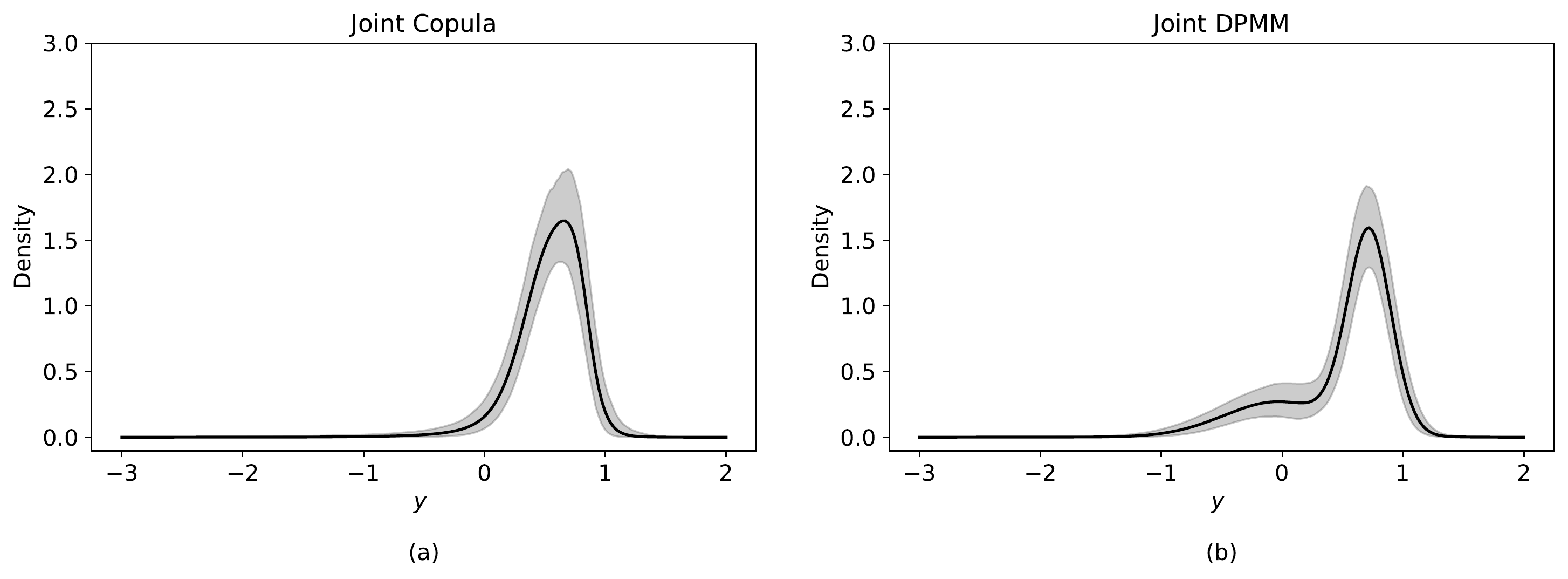}\vspace{-3mm}
    \caption{Posterior mean (\full) and 95\% credible interval (\sqrlow) of (a) $p_N(y \mid x=0)$ for the joint copula method and (b) $p_\infty(y \mid x = 0)$ for the joint DPMM } \label{fig:lidar_jdpmm_x0}\vspace{2mm}
\end{figure}
\begin{figure}[!h]
    \centering
        \includegraphics[width=0.97\textwidth]{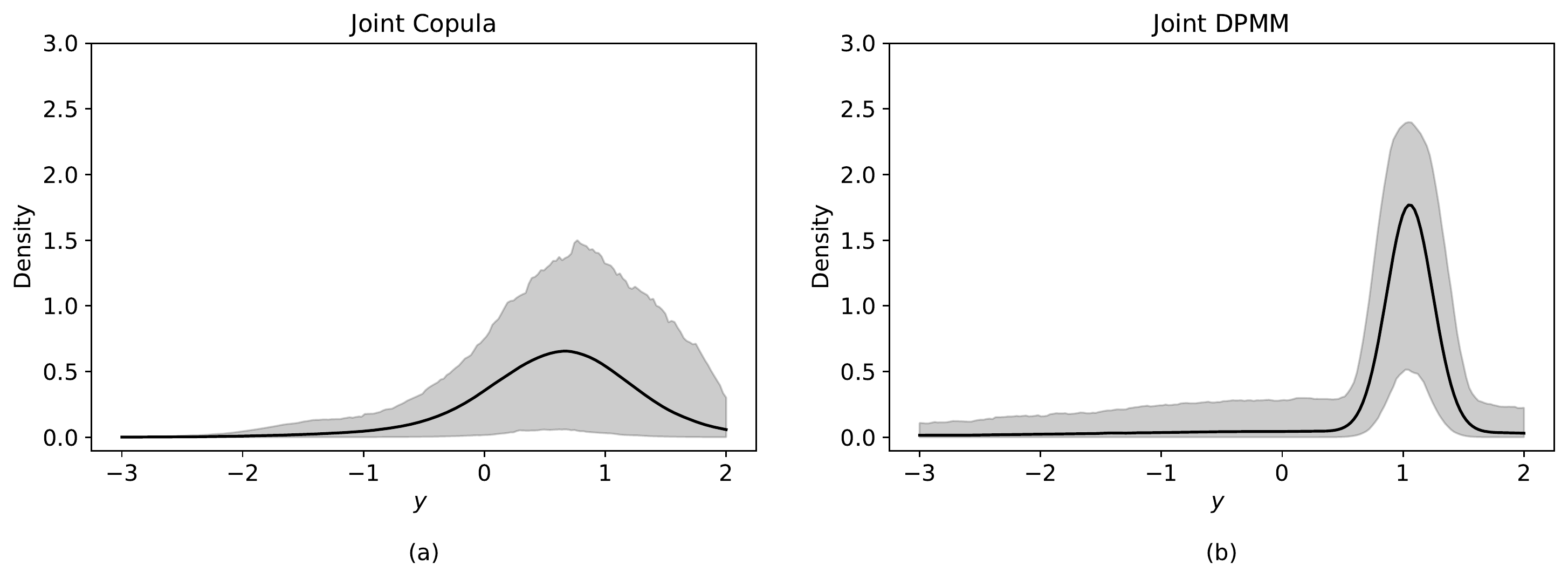}\vspace{-3mm}
    \caption{Posterior mean (\full) and 95\% credible interval (\sqrlow) of (a) $p_N(y \mid x=-3)$ for the joint copula method and (b) $p_\infty(y \mid x=-3)$ for the joint DPMM} \label{fig:lidar_jdpmm_x-3}\vspace{2mm}
\end{figure}
\begin{figure}[!h]
    \centering
        \includegraphics[width=0.97\textwidth]{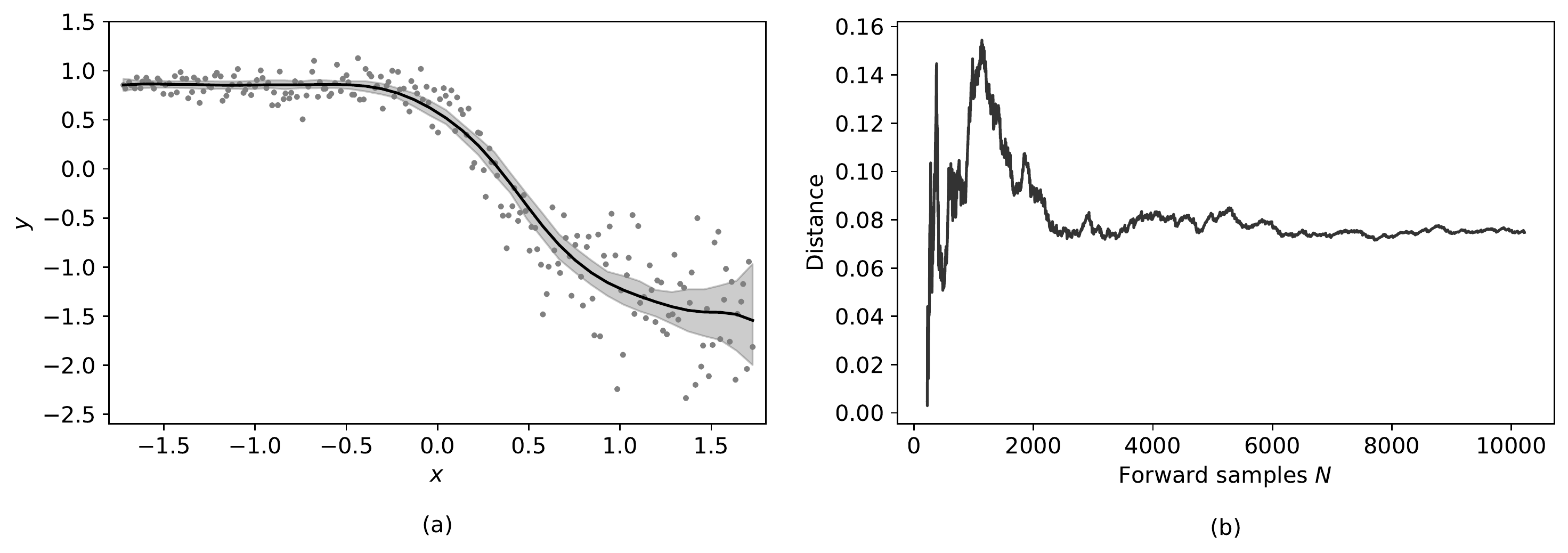}\vspace{-3mm}
    \caption{(a) Posterior mean (\full) and 95\% credible interval (\sqrlow) of the conditional median of $P_N(y \mid x)$, with data (\dotmid); (b) Estimated $L_1$ distance  $\lVert p_N(\cdot \mid {x}) - p_{n}(\cdot \mid {x}) \rVert_1$  for a single forward sample with $x=0$ } \label{fig:lidar_conv}\vspace{2mm}
\end{figure}

\subsubsection{Multivariate covariates in UCI datasets}
We now demonstrate the conditional copula method for prediction in the regression and classification setting with multivariate covariates, which is of particular interest to the machine learning community. For high-dimensional covariates, the conditional copula method performs better than the joint method, both in terms of computational speed and test log-likelihood. This is likely due to the dominance of estimating $P_n(\mathbf{x})$ in high dimensions, which disrupts the estimate of $P_n(y \mid \mathbf{x})$.

Similar to the multivariate density estimation, we demonstrate the regression and classification conditional copula methods on UCI datasets with scalar $y$ and multivariate $\mathbf{x}$. Again, we evaluate the methods with 10 random 50-50 test-train splits and evaluate the average test conditional log-likelihoods. We convert categorical variables into dummy variables, and report the preprocessed covariate dimensionality in Table \ref{tab:reg_test_loglik}. We compare to Bayesian linear regression and Gaussian processes (GP) with a single length scale RBF kernel as baselines for regression, and similarly to logistic regression and GPs with the logistic link and Laplace approximation for classification. We use the Laplace approximation as it is available off-the-shelf in \texttt{sklearn}, and we found that independent kernel length scales (ARD) performed worse due to overfitting given $n$ is moderate. For the conditional copula method, we have distinct bandwidths $\rho_{1:d}$ for each covariate, which we optimize through the prequential log-likelihood over $M = 10$ permutations.

\begin{table}[!h]
\footnotesize
\begin{center}
\begin{tabular}{ clccccc }
 & Dataset & $n$ & $d$ & Linear & GP &  Copula\\
 \hline 
Regression&  Boston & $\phantom{1}506$ &$13$&$-0.842$ $(0.043)$&$-0.404$ $(0.040)$&$-\textbf{0.351}$ $(0.025)$\\  
& Concrete  & $1030$ &$\phantom{1}8$ &$-0.965$ $(0.008)$&$-\textbf{0.364}$ $(0.014)$&$-0.445$  $(0.013)$    \\
  &Diabetes & $\phantom{1}442$ & $10$&$-1.096$ $(0.017)$&$-1.089$ $(0.015)$&$-\textbf{1.003}$ $(0.018)$  \\
   &Wine Quality   &$1599$  &$11$ &$-1.196$ $(0.017)$& $-\textbf{0.497}$ $(0.034)$ & $-1.143$ $(0.020)$   \\
   \hline
   Classification &Breast cancer &$\phantom{1}569$&$30$ &$-0.107$ $(0.005)$&$-0.105$ $(0.005)$&$-\textbf{0.096}$ $(0.008)$ \\
   &Ionosphere&$\phantom{1}351$&$33$&$-0.348$ $(0.005)$&$-\textbf{0.304}$ $(0.006)$&$-0.388$ $(0.016)$\\
   &Parkinsons&$\phantom{1}195$ &$22$ &$-0.352$ $(0.007)$&$-0.364$ $(0.013)$ &$-\textbf{0.257}$ $(0.010)$\\ 
   &Statlog &$1000$&$20$ &$-\textbf{0.530}$ $(0.009)$& $-0.542$ $(0.011)$&$-0.541$ $(0.006)$\\ 
\end{tabular}
\caption{Average test log-likelihood, standard errors (in brackets) and best performance in bold} \label{tab:reg_test_loglik}
\end{center}
\end{table}
\vspace{-5mm}

In Table \ref{tab:reg_test_loglik}, we see the test log-likelihoods, where the copula method is competitive with the GP, though in general we find that the GP provides a better estimate for the mean function for regression. Again, optimization took the most time due to the $d$ bandwidths, taking on average 30 seconds per fold for the slowest example (`Statlog'). The time for actual fitting and prediction on the test set was under 120ms per fold for all examples. The GP on the slowest examples required around 20 seconds per fold for the marginal likelihood optimizations, but computation time scales as $\mathcal{O}(n^3)$. 
\section{Theory}\label{sec:theory}
In this section, we provide a theoretical analysis of the martingale posteriors and predictive resampling using the copula update introduced in Section \ref{sec:copula}. We utilize the theory of c.i.d. sequences from the works of \cite{Berti2004,Berti2013}. 
We then show frequentist consistency (with little $n$) under relatively weak conditions for the multivariate copula update by extending the proof of \cite{Hahn2018}, and we discuss its implications.  All proofs are deferred to Appendix \ref{Appendix:proofs}. 

\subsection{Martingale posteriors for copula density estimation}\label{sec:theory_mv}
We first analyze the properties under predictive resampling of  the multivariate copula recursive update for the martingale posterior. We write $P_i(\mathbf{y})$ as the joint cumulative distribution function of the density $p_i(\mathbf{y})$ with update \eqref{eq:mv_DP_copdens}, and consider predictive resampling starting at $p_n(\mathbf{y})$ such that $\mathbf{Y}_{i+1} \sim P_i(\mathbf{y})$ for $i = n,n+1\ldots,N$. As before, $n$ corresponds to the number of observed data points, whereas $N-n$ corresponds to the number of forward samples drawn from predictive resampling. The first two results follow directly from the c.i.d. property of the sequence.
\begin{theorem}\label{Th:cid_exchangeable}[\citet[Theorem~2.5]{Berti2004}]
  The sequence $\mathbf{Y}_{N+1},\mathbf{Y}_{N+2},\ldots$ is asymptotically exchangeable, that is 
 $$
 (\mathbf{Y}_{N+1},\mathbf{Y}_{N+2},\ldots) \overset{d}{\to} (\mathbf{Z}_1,\mathbf{Z}_2,\ldots) 
 $$
for $N \to \infty$, where $(\mathbf{Z}_1,\mathbf{Z}_2,\ldots)$ is exchangeable.
 \end{theorem}
 The above justifies that we may not need to average over permutations for sufficiently large $N$ when predictive resampling. 

As mentioned in Section \ref{sec:pred_coherence}, we would like $P_N(\mathbf{y}) \to P_\infty(\mathbf{y})$  at each $\mathbf{y} \in \mathbb{R}^d$, which indeed holds for predictive resampling here from the c.i.d. sequence:
\begin{theorem}\label{Th:weak} [\citet[Lemma 2.1, 2.4]{Berti2004}]
  There exists a random probability measure $P_\infty$ such that $P_N$ converges weakly to $P_\infty$ almost surely.
 \end{theorem}
 Specifically for the univariate case of the copula update above, we can strengthen this to convergence in total variation, which also implies that the martingale posterior $P_\infty$ is absolutely continuous, following from an interesting result in \cite{Berti2013}.
\begin{theorem} \label{Th:absolute_continuity}
 For $y \in \mathbb{R}$, suppose the sequence of probability measures $P_N$ has density function $p_N(y)$ and cumulative distribution function $P_N(y)$  satisfying the updates \eqref{eq:DP_copdens}. Let us assume that the initial $P_n(y)$ is continuous and its density satisfies
\begin{equation}
\int_K p^2_n(y) \, dy < \infty
\end{equation}
 for all $K$, where $K$ is a compact subset of $\mathbb{R}$ with finite Lebesgue measure. For the sequence
 $$
 \alpha_i = \left(2-\frac{1}{i}\right)\frac{1}{i+1},
 $$
  let us assume further that $\rho <1 /\sqrt{3}$. We then have
  \begin{enumerate}[\upshape (a)]
  \item $P_\infty$ is absolutely continuous with respect to the Lebesgue measure almost surely, with density $p_\infty$.
    \item $P_{N}$  converges in total variation to $P_\infty$ almost surely, that is 
    $$\lim_{N\to \infty}\int |p_N(y)  - p_\infty(y) | \, dy = 0 \quad \textnormal{a.s.}$$
  \end{enumerate}
\end{theorem}  
The assumptions hold if  $p_n(y)$ is continuous. From this, we are justified in using $p_N(y)$ as an approximate sample of the martingale posterior $p_\infty(y)$. We conjecture that the choice of $\rho <1/\sqrt{3}$  can be relaxed, and empirically it seems the case. Furthermore, this restriction on $\rho$ is not needed if $\alpha_i = (i+1)^{-1}$. Unfortunately, we have been unable to extend Theorem \ref{Th:absolute_continuity} to the multivariate copula update, as the update for $P\left(y^j \mid y^{1:j-1}\right)$ is not as easy to bound. We also conjecture that the $L_1$ convergence  holds true in the multivariate case, and again the empirical results suggest so.

We can also quantify to some degree the convergence rate to $P_\infty$ as we predictively resample. We have the following result from a variant of the Azuma-Hoeffding inequality from \cite{Mcdiarmid1998}.
\begin{proposition}  \label{prop:concentration}
For $M > N$ and any $\epsilon \geq 0$, the cumulative distribution function $P_N(\mathbf{y})$ of the density in \eqref{eq:mv_DP_copdens} satisfies 
$$\sup_{\mathbf{y}} \,\mathbb{P}\left(|P_{M}(\mathbf{y}) - P_{N}(\mathbf{y})|  \geq \epsilon\right) \leq 2\exp\left( \frac{-\epsilon^2}{\frac{2\epsilon \alpha_{N+1}}{3} + \frac{1}{2}\sum_{i=N+1}^M\alpha^2_{i}}\right).
$$
\end{proposition}    
Taking the limit (superior) as $M \to \infty$ of the above gives insight into the quality of the approximation of $P_\infty$ when we truncate the predictive resampling at $P_N$. For our choice of $\alpha_i$ from \eqref{eq:alpha_i}, we have $\sum_{i=N+1}^\infty \alpha_i^2 = \mathcal{O}(N^{-1})$, so the limiting probability of a difference greater than $\epsilon$  decreases roughly at rate $\exp(-\epsilon^2 cN)$ for some constant $c$.  Notably, this rate is independent from the dimensionality $d$, and instead depends only on the sequence $\alpha_i$. Furthermore, we have some notion of posterior contraction in Proposition \ref{prop:concentration} if  we instead consider $N$ as the number of observed data points and $M$ as the number of forward samples.

\subsection{Martingale posteriors for conditional copula regression}

For the regression case where $y \in \mathbb{R}$, $\mathbf{x}\in \mathbb{R}^d$, we analyze the update given in \eqref{eq:conditreg_conditional} and \eqref{eq:condit_alpha}. Assuming we have observed $y_{1:n},\mathbf{x}_{1:n}$, we draw the sequence $\mathbf{X}_{n+1:\infty}$ from the Bayesian bootstrap with $\mathbf{x}_{1:n}$. While this is no longer the traditional c.i.d. setup, we still have that $P_N(y \mid \mathbf{x})$ is a martingale under predictive resampling, so we have that $P_N(y \mid  \mathbf{x})$ converges pointwise for each $\mathbf{x}$ almost surely. Fortunately, \citet[Theorem 2.2]{Berti2006} assures that the martingale posterior $P_\infty(y \mid \mathbf{x})$ exists.
\begin{theorem} \label{Th:weak_reg}
  For each $\mathbf{x} \in \mathbb{R}^d$, there exists a random probability measure $P_\infty(\cdot \mid \mathbf{x})$ such that $P_N(\cdot \mid \mathbf{x})$ converges weakly to $P_\infty(\cdot \mid \mathbf{x})$ almost surely.
 \end{theorem}

We also have the  appropriate extension to Proposition \ref{prop:concentration} below.

\begin{proposition}  \label{prop:concentration_reg}
For $M > N$ and any $\epsilon \geq 0$, the cumulative distribution function $P_N({y} \mid \mathbf{x})$ of the density in \eqref{eq:conditreg_conditional} satisfies 
$$\sup_{{y}} \,\mathbb{P}\left(|P_{M}(y \mid \mathbf{x}) - P_{N}(y\mid\mathbf{x})|  \geq \epsilon\right) \leq 2\exp\left( \frac{-\epsilon^2}{\frac{4\epsilon C\alpha_{N+1} }{3}+ 2C^2\sum_{i=N+1}^M\alpha^2_{i}}\right)
$$
 for each $\mathbf{x} \in \mathbb{R}^d$, where $C$ depends only on $\rho$ and $\mathbf{x}$. 
\end{proposition}   
It can be shown that $C$ increases as $\mathbf{x}$ moves from the origin. Assuming $x_{1:n}$ is standardized, this implies that the number of forward samples needed for convergence may increase as $\mathbf{x}$ shifts away from the data. The above results can also be easily extended to the classification scenario.

{\subsection{Frequentist consistency of copula density estimation}\label{sec:consistency}
To simulate from the martingale posterior given $\mathbf{Y}_{1:n}$, we start with the density $p_n$ computed from \eqref{eq:mv_DP_copdens}, so we would like to verify that it is indeed an appropriate predictive density.
In this subsection, we thus concern ourselves with the frequentist notion of consistency, that is we look at the properties of the density estimate $p_n$ assuming $\mathbf{Y}_{1:n}$ is i.i.d. from some probability distribution with density function $f_0$ as we take $n \to \infty$. It should be noted that this is distinct from the Doob-type asymptotics of predictive resampling in the previous subsections where we take $N \to \infty$. 

The frequentist consistency of the univariate copula method was first discussed in \citet{Hahn2018} based on the `almost supermartingale' of \cite{Robbins1971}. We will now extend the result to the multivariate copula method, of which the univariate method is a special case. The full proof can be found in Appendix \ref{Appendix:consistency}. Instead of the Kullback-Leibler divergence, we work with the squared Hellinger distance between probability density functions $p_1$ and  $p_2$ on $\mathbf{y}\in\mathbb{R}^d$, defined as
$H^2(p_1,p_2) := 1- \int \sqrt{p_1(\mathbf{y}) \, p_2(\mathbf{y})} \, d\mathbf{y}.$
We then have the main result. 
\begin{theorem}\label{Th:consistency}
For $\mathbf{Y}_{1:n} \iid f_0$, suppose the sequence of densities $p_n(\mathbf{y})$ satisfies the updates in \eqref{eq:mv_DP_copdens}.
Assume that $\rho \in (0,1)$, $\alpha_i = a(i+1)^{-1}$ where  $a< {2}/{5}$, and there exists $B < \infty$ such that 
${f_0(\mathbf{y})}/{p_0(\mathbf{y})} \leq B$
for all $\mathbf{y} \in \mathbb{R}^d$. We then have that $p_n$ is Hellinger consistent at $f_0$, that is
$$
\lim_{n\to \infty}H^2(p_n,f_0)= 0 \quad \textnormal{a.s.}
$$
\end{theorem}

Intuitively, the update \eqref{eq:mv_DP_copdens} can be regarded as a stochastic gradient descent in the space of probability density functions, where $\alpha_{i+1}$ is the step-size. As is standard in stochastic optimization \citep{Kushner2003}, consistency of the copula method relies delicately on the decay of the sequence $\alpha_i$, which ensures we approach the independent copula at the correct rate. A similar condition is for example discussed in \cite{Tokdar2009} for Newton's algorithm. On the one hand, we require $\sum_{i=1}^\infty \alpha_i  = \infty$ to ensure that the initialization $p_0$ is forgotten. On the other hand, we require the sequence $\alpha_i$ to decay sufficiently quickly to $0$, that is $\sum_{i=1}^\infty \alpha_i^2 < \infty$,  for information to accumulate correctly. The requirement on $a$ also ensures the information in later terms decay properly. Notably, the condition on $a < {2}/{5}$ is different to the suggestion for predictive resampling, so a different choice of $\alpha_n$ may be more suitable when consistency is of primary interest. The second assumption is a regularity condition on the tails of the initial $p_0$ being heavier than $f_0$, which motivates a heavy-tailed initial density as also suggested by \cite{Hahn2018}. Interestingly, the bounded condition on $f_0/p_0$ is the only requirement on $f_0$ for consistency, which  follows from the nonparametric update. However, unlike the KDE there are no conditions on the bandwidth $\rho$, which likely follows from the data-dependence of the copula kernel. 

There are a number of unanswered questions when compared to the consistency of traditional Bayes. The first is whether the martingale posterior converges weakly to the Dirac measure at $F_0$, as we have only shown Hellinger consistency of the posterior mean measure of $P_\infty$. We believe this is likely to be positive, as there is a notion of posterior contraction as in Proposition \ref{prop:concentration}.  A related inquiry is the rate of convergence of $p_n$, or the martingale posterior on $p_\infty$, to the true $f_0$. The second and more ambitious question is whether the above approach provides a general method to prove consistency for other copula models. For the multivariate copula method, we only require the weak tail condition on $f_0$, but the proof relies heavily on the nonparametric nature of the update. It is still unclear what the conditions would be if the copula sequence corresponded to a parametric Bayesian model, such as the examples given in \cite{Hahn2018}. In the absence of the prior under the predictive view, a question of interest is whether an analogue to the Kullback-Leibler property of the traditional Bayesian prior (e.g. \citet[Definition 6.15]{Ghosal2017}) exists, which would highlight a predictive notion of model misspecification. }

\section{Discussion}\label{sec:discussion}
{ We see that Bayesian uncertainty at its core is concerned with the missing observations required to know any statistic of interest precisely. In the i.i.d. case, this is $Y_{n+1:\infty}$, and our task is to obtain the joint distribution $p(y_{n+1:\infty} \mid y_{1:n})$, which is simplified through the factorization into a sequence of of 1-step ahead predictive densities.  One open question is whether there are more general methods to elicit this joint beyond the likelihood-prior construction and the prequential factorization. For the more general data setting, the Bayesian would be tasked with eliciting $p(y_{\text{mis}} \mid y_{\text{obs}})$, where the  missing observations $y_{\text{mis}}$ would be specific to the setting and statistic of interest. We highlight that $y_{\text{mis}}$ must be sufficiently large to compute the statistic precisely, unlike in multiple imputation \citep{Rubin2004} where the imputed data is often finite and for computational convenience. For future work, identifying $y_{\text{mis}}$ and extending the methodology in more complex data settings such as time series or hierarchical data is of primary interest.}

In terms of practical methodology,  it is worth comparing when one would prefer to use the Bayesian bootstrap versus the copula methods. When the data is high-dimensional but a low-dimensional statistic is of interest, the copula methods may not be suitable, as computing the density on a grid or sampling the data directly is required. Fortunately, the Bayesian bootstrap shines in this setting. On the other hand, the discreteness of the Bayesian bootstrap makes it unsuitable for when smoothness is required, for example when the density is directly of interest, or in regression where we rely on smoothness with $x$. In these settings, the copula methods are highly suitable. Together, the predictive framework allows us to cover a wide variety of settings with practical advantages over the traditional Bayesian approach.

{ We believe our framework offers interesting insight into the interplay between Bayesian and frequentist approaches.  As we have seen through the lens of the Bayesian bootstrap, Bayesians and frequentists are concerned with $Y_{n+1:\infty}$ and $Y_{1:n}$ respectively.  Analysis of the frequentist asymptotic properties of martingale posteriors also offers new challenges, as we must work with the predictive distribution directly, and it is unclear if the methods used in our paper generalize to other copula models.  For generalizations of our martingale posterior framework,  imputing aspects of the population instead of the entire population directly may also help bridge the gap between Bayesian and frequentist methods. In the hierarchical example in Section \ref{sec:intro}, we can in fact treat $\theta_i$ as the mean of population $i$ from which we observe a single sample $y_i$. We would thus be imputing the means of observation populations (i.e. the random effects) instead of the entire population of observables directly. This interpretation would align well with our philosophy of only imputing what one would need to carry out the statistical task.}

\newpage
\section*{Acknowledgements}
The authors are grateful for the detailed comments of three referees and the Associate Editor on the previous version of the paper. The authors also thank Sahra Ghalebikesabi, Brieuc Lehmann, Geoff Nicholls, George Nicholson and Judith Rousseau for their helpful comments. Fong is funded by The Alan Turing Institute Doctoral Studentship, under the EPSRC grant EP/N510129/1. Holmes is supported by The Alan Turing Institute, the Health Data Research, U.K., the Li Ka Shing Foundation, the
Medical Research Council, and the U.K. Engineering and Physical
Sciences Research Council.

\bibliographystyle{apalike}

\bibliography{paper-ref}

\newpage
\begin{appendices}
\setcounter{theorem}{7} 
\setcounter{proposition}{2}
\setcounter{figure}{13}
\numberwithin{equation}{section}

\section{Notation}\label{Appendix:notation}
In this section, we summarize the notation used in the main paper. In the table below, we provide a summary of the notation introduced in Section 1 and provide some concrete examples after.
\begin{table}[!h]
\small
\begin{center}
{\def\arraystretch{1.3}\tabcolsep=10pt
\begin{tabular}{ c|c |l }
 &  \textbf{Notation} & \textbf{Definition} \\
 \hline 
Data& $Y$& One data unit as a random variable, e.g. one row in a data table\\
& $y$ & Observable as a fixed realisation of $Y$, or input into a PDF/CDF  \\
& $n$& The size of the data set, or the number of observed data units\\
&$N$ & The size of the study population (or approximating $\infty$)\\
& $ Y_{1:N}, Y_{1:\infty}$ & The conceptual complete data table for the whole study population \\
   \hline
  Distributions & $F_0$&  The true, unknown sampling distribution function where $Y_{1:N} \sim F_0$\\  
& $F_N, F_\infty$ &  The (limiting) empirical distribution function of the imputed population $Y_{1:N}$ \\
& $P_N, P_\infty$ &  The (limiting) predictive distribution function of the imputed  $Y_{1:N}$ \\
   \hline 
  Parameters& $\Theta$ & Bayesian parameter as a random variable, with distribution  $\Pi(\cdot)$ \\
  &$\bar{\theta}_N $& Posterior mean of $\Theta$ computed from $Y_{1:N}$\\
     &$\Pi(\theta \mid y_{1:n})$& The conventional Bayesian posterior distribution  \\
   & $\theta_0$ & True parameter or estimand of interest, computed from $F_0$ \\
   & $\theta_N, \theta_\infty$ & The estimate of $\theta_0$, computed from the imputed population $Y_{1:N}, Y_{1:\infty}$  \\
    &$\Pi_N(\theta_N \mid y_{1:n})$,& \multirow{2}{*}{The (finite) martingale posterior distribution} \\ 
    &$\Pi_{\infty}(\theta_\infty \mid y_{1:n})$ &\\
   \hline
\end{tabular}
}
\caption{Notation for some key values} \label{tab:notation}
\end{center}
\end{table}
\vspace{-5mm}
\subsection{Parameters}
Formally, we write $\theta(F)$ as a functional which takes as input a distribution function $F(y)$ and returns a vector in $\mathbb{R}^p$. In some cases, it can be written as
$$
\theta(F) = \argmin_\theta \int \ell(\theta,y) \, dF(y)
$$
for a loss function of interest, where for example we may have the mean functional as 
$$
\theta(F) = \int y \, dF(y).
$$
The true parameter/statistic of interest is then $\theta_0= \theta(F_0)$. The input distribution function may also be the atomic empirical distribution $F_N$, in which case $\theta_N = \theta(F_N)$. We also use the notation $\theta(Y_{1:N})$ interchangeably with $\theta(F_N)$. In the mean example then, we have
$$
\theta(Y_{1:N}) = \theta(F_N) = \frac{1}{N}\sum_{i=1}^N Y_i \cdot
$$
In the limiting case, we write $\theta_\infty = \theta(F_\infty)= \theta(Y_{1:\infty})$, where
$$
F_\infty = \lim_{N\to\infty}F_N.
$$
In some cases where $\theta(F)$ requires $F$ to be smooth, we may instead pass in a smooth predictive distribution, that is $\theta_N= \theta(P_N)$, and similarly $\theta_\infty = \theta(P_\infty)$.
  
 \newpage
 
\section{Bayesian inference as missing data}\label{Appendix:tables}
We illustrate the conceptual idea of treating Bayesian inference as a missing population imputation in the visualization below.
\begin{table}[!htb]
     \begin{minipage}{.48\textwidth}
        \caption*{\textbf{A}: Population data table \vspace{-3mm}}
        \centering
            \begin{tabular}{@{\extracolsep{1pt}}c|ccc} 
                \hline 
                Unit & A &B& C\\ 
                \hline  
                1 & 10 & 5 & 8     \\ 
                2 & 5& 21 & 13  \\ 
                3 & 13 &12 & 17 \\ 
                4 & 6 & 7 & 10 \\ 
                \vdots & \vdots & \vdots & \vdots \\ 
                $N$ & 21 & 13 & 11  \\ 
                \hline 
            \end{tabular} 
    \end{minipage}
    \begin{huge}
    $\hspace{-9.3mm}\xRightarrow{\begin{subarray}{c}\text{\small The study } \end{subarray}}\hspace{-9.2mm}$
    \end{huge}
    \begin{minipage}{.48\textwidth}
        \centering
        \caption*{\textbf{B}: Observed data \vspace{-3mm}}
            \begin{tabular}{@{\extracolsep{1pt}}c|ccc} 
                \hline 
                Unit & A &B& C\\ 
                \hline 
                1 & 10 & 5 & 8     \\ 
                2 & ?& ? & ?  \\ 
                3 & ? &? & ? \\ 
                4 & 6 & 7 & 10 \\ 
                \vdots & \vdots & \vdots & \vdots \\ 
                $N$ & ? & ? & ?  \\ 
                \hline 
            \end{tabular} 
    \end{minipage}
    
\vspace{6mm}
  \begin{minipage}{.48\textwidth}
        \caption*{\textbf{C}: Exchangeability \vspace{-3mm}}
        \centering
            \begin{tabular}{@{\extracolsep{1pt}}cccc} 
                \hline 
                Unit & A &B& C\\ 
                \hline 
                1 & 10 & 5 & 8     \\ 
                $n=2$ & 6 & 7 & 10 \\ 
                \hdashline 
                3 & ? &? & ? \\ 
                4  & ? &? & ?\\ 
                \vdots & \vdots & \vdots & \vdots \\ 
                $N$ & ? & ? & ?  \\ 
                \hline 
            \end{tabular} 
    \end{minipage}
    \begin{huge}
    $\hspace{-10mm}\xRightarrow{\begin{subarray}{c}
                \text{\small Predictively}\\
                \text{\small resample}
            \end{subarray}}\hspace{-10mm}$
    \end{huge}
    \begin{minipage}{.48\textwidth}
        \centering
        \caption*{\textbf{D}: Imputed population 1 \vspace{-3mm}}
            \begin{tabular}{@{\extracolsep{1pt}}c|ccc} 
                \hline 
                Unit & A &B& C\\ 
                \hline 
                1 & 10 & 5 & 8     \\ 
                2 & 6 & 7 & 10   \\ 
                \hdashline
                3 &{\bf{4}} & {\bf{20}} & {\bf{12}} \\ 
                4 & {\bf{12}} & {\bf{12}} &  {\bf{18}} \\ 
                \vdots & \vdots & \vdots & \vdots \\ 
                $N$ & {\bf{19}} & {\bf{15}} & {\bf{12}}  \\ 
                \hline 
            \end{tabular} 
    \end{minipage}
    
    \vspace{3mm}
    
      \begin{minipage}{.48\textwidth}
       \caption*{} \vspace{-5mm} 
        \centering
    \end{minipage}
    \begin{huge}
    $\hspace{-10mm}\xRightarrow{\begin{subarray}{c}
                \text{\small Predictively}\\
                \text{\small resample}
            \end{subarray}}\hspace{-10mm}$
    \end{huge}
    \begin{minipage}{.48\textwidth}
        \centering
        \caption*{\textbf{ } Imputed population 2 \vspace{-3mm}}
            \begin{tabular}{@{\extracolsep{1pt}}c|ccc} 
                \hline 
                Unit & A &B& C\\ 
                \hline 
                1 & 10 & 5 & 8     \\ 
                2 & 6 & 7 & 10   \\ 
                \hdashline
                3 &{\bf{6}} & {\bf{18}} & {\bf{13}} \\ 
                4 & {\bf{10}} & {\bf{9}} & {\bf{21}} \\ 
                \vdots & \vdots & \vdots & \vdots \\ 
                $N$ & {\bf{15}} & {\bf{12}} & {\bf{16}} \\ 
                \hline 
            \end{tabular} 
    \end{minipage}
    \vspace{5mm}

    \caption*{{\bf{A}}: The conceptual complete target population data table. If this table was known, then there would be no uncertainty in the statistic of interest, $\theta_0 = \theta(Y_{1:N})$, or in any resulting decision.\\  {\bf{B}}: The experiment or observational study reveals $n$ data units selected at random from the population data table. Uncertainty in $\theta(Y_{1:N})$ arises from the remaining missing data marked $(?)$. \\ {\bf{C}}: Following an assumption of exchangeability, we can relabel the observed units from 1 to $n$. \\{\bf{D}}: A predictive model allows us to impute the missing data $Y_{n+1:N} \sim p(\cdot \mid y_{1:n})$ via predictive resampling to create full synthetic data tables. The imputed synthetic data tables gives us corresponding estimates $\{\theta_N^{(1)}, \theta_N^{(2)}, \ldots\}$,  which are posterior samples that characterise the Bayesian uncertainty in $\theta_0$ arising from the missing $Y_{n+1:N}$. This notion of Bayesian inference as imputation is connected to the ideas of \citet{Rubin1974,Rubin2008}.}
\end{table}
\newpage
\section{Limiting predictive and empirical distribution}\label{Appendix:limit_pred_emp}
In this section, we summarize some asymptotic properties of conditionally identically distributed (c.i.d.) sequences and exchangeable sequences from the literature. In particular, we look at the equivalence of the limiting predictive and empirical distributions which we use in the main paper.

\subsection{Conditionally identically distributed sequences}\label{Appendix:cid_pred_emp}
We begin with the more general class of c.i.d. sequences, of which exchangeable sequences are a subset. We have the following strong law for c.i.d. sequences.
\begin{theorem}\label{Th:CID}[\citet[Theorem~2.2]{Berti2004}]
Suppose the sequence $Y_{1},Y_{2},\ldots$ is c.i.d. with respective predictive distribution functions $P_0,P_{1},\ldots$, where $P_N$ is conditional on $\mathcal{F}_N:=\sigma(Y_1,\ldots, Y_N)$. We then have
\begin{equation*}
F_\infty(y) := \lim_{N\to \infty} \frac{1}{N}\sum_{i=1}^N \mathbbm{1}(Y_i \leq y) = P_{\infty}(y) \quad \text{a.s.}
\end{equation*} 
for each $y \in \mathbb{R}$, where $P_{\infty}:=\lim_{N \to \infty} P_N$ is the limiting random predictive distribution function from the martingale posterior, obtained through predictive resampling as in Condition \ref{cond:existence}.
\end{theorem}
In summary, the limiting empirical is equivalent to the limiting predictive distribution for c.i.d. sequences (with exchangeable sequences as a special case), which justifies the interchangeability of $F_\infty$ and $P_\infty$ as discussed in Section \ref{sec:pred_coherence}.

\subsubsection{Convergence of the parameter} \label{Appendix:limit_param}
Here we consider the convergence of parameters of interest $\theta_0$ which take the form of
$$
\theta_0 = \int g(y)\, dF_0(y)
$$
for the martingale posterior with c.i.d. sequences. This form of the parameter is a special case of the more general $ \argmax_{\theta}\int \ell(\theta,y) \, dF_0(y)$, which is difficult to analyze due to the stronger convergence required of the entire function $\int \ell(\theta,y) \, dF_0(y)$. For a finite predictive sample of size $N$, we can write the parameter estimate as
 $$\theta(F_N) = \int g(y) \, dF_N(y) = \frac{1}{N}\sum_{i=1}^N g(y_i).$$
Alternatively, if we work directly with the predictive distribution function, we have
  $$\theta(P_N) = \int g(y) \, dP_N(y).$$
The more general strong law for c.i.d. sequences assures us that $\theta(Y_{1:N})$ converges with $N$ in both settings to $\theta_\infty$ from the martingale posterior almost surely. 
\begin{theorem}\label{Th:CID_param}[\citet[Lemma~2.1, Theorem~2.2]{Berti2004}]
Suppose the sequence $Y_{1},Y_{2},\ldots$ is c.i.d. with respective predictive distribution functions $P_0,P_{1},\ldots$, where $P_N$ is the predictive distribution function conditioned on $\mathcal{F}_N$. For a measurable function $g: \mathcal{Y}\to \mathbb{R}$ such that $E\left[|g(Y_1)| \right] < \infty$, we have that
\begin{equation*}
\lim_{N\to \infty} \frac{1}{N}\sum_{i=1}^N g(Y_i) = \theta_\infty
\end{equation*} 
almost surely and in $L_1$. Likewise, we have that
\begin{equation*}
\lim_{N\to \infty} \int g(y) \, dP_N(y) = \theta_\infty
\end{equation*} 
almost surely and in $L_1$. Furthermore, $\theta_\infty$ is integrable and 
$$E\left[\theta_\infty \mid \mathcal{F}_N\right] = \theta(P_N).$$
\end{theorem}
However, we prefer to define the parameter as a function of $F_\infty$ instead of a limiting parameter, that is $\theta_\infty = \theta(F_\infty)$, as $F_\infty$ always exists for c.i.d. sequences.
\subsection{Exchangeable sequences}\label{Appendix:Bayes_pred_emp}
Given the model specification of the likelihood and prior, we have random variables $(\Theta, Y_1,Y_2,\ldots)$ on some probability space which have the joint density
\begin{equation}\label{Appendix:eq:joint_param}
p(\theta,y_{1:N}) =\pi(\theta)\, \prod_{i=1}^N f_\theta(y_i)
\end{equation}
for all $N$. We write our random sequence of posterior predictive distribution functions as
\begin{equation}\label{Appendix:eq:random_predictive}
P_{N}(y) := P(Y_{N+1} \leq y \mid \mathcal{F}_{n})
\end{equation}
for all $N$. 
We then have the following equivalence result.
\begin{theorem}\label{Th:Predictive_Doob}[\citet{Doob1949,Lijoi2004}]
Suppose that $(\Theta,Y_{1:N})$ are distributed according to $P$ with density (\ref{Appendix:eq:joint_param}), then the sequence of predictive distribution functions satisfies
\begin{equation*}
P_\infty(y) := \lim_{N\to \infty} P_N(y) = F_{\Theta}(y)\quad\textnormal{a.s.} 
\end{equation*} 
for each $y\in \mathbb{R}$. Furthermore, this holds when our parameter space is the family of all densities on $\mathbb{R}$, that is $\Theta = f$ where $f$ is a random density,  and we define $F_{{f}}(y) = \int_{-\infty}^y {f}(z) \, dz$. 
\end{theorem}

From Theorems \ref{Th:CID} and \ref{Th:Predictive_Doob}, we then have a strong law result below, which justifies using the limiting empirical distribution function $F_\infty$ as in Section \ref{sec:loss_functions}.
\begin{theorem}\label{Th:Empirical_Doob}[\citet[Theorem~2.2]{Berti2004}]
Suppose that $(\Theta,Y_{1:N})$  are distributed according to $P$ with density (\ref{Appendix:eq:joint_param}), then the sequence of empirical distribution functions satisfies
\begin{equation*}
F_\infty(y) :=\lim_{N\to \infty} \frac{1}{N}\sum_{i=1}^N \mathbbm{1}(Y_i \leq y) = F_{\Theta}(y) \quad\textnormal{a.s.}
\end{equation*} 
for each $y\in \mathbb{R}$, where $\mathbbm{1}(A)$ is the indicator function for the event A.   Again, this holds when $\Theta =f$,  and $F_{{f}}(y) = \int_{-\infty}^y {f}(z) \, dz$. 
\end{theorem}

\newpage
\section{Proofs}\label{Appendix:proofs}

\subsection{Corollary \ref{corr:copula_CID}} \label{Appendix:corr_copula_CID}
From the martingale condition \eqref{eq:martingale_density}, we have that the product $p_{i+1}(y)\,p_i(y_{i+1}) = p_i(y,y_{i+1})$ is a bivariate density with marginals $p_i(y)$ and $p_i(y_{i+1})$, so from Theorem \ref{Th:Sklar} there exists a bivariate copula density $c_{i+1}$ such that $p_i(y,y_{i+1}) = c_{i+1}\{P_i(y),P_i(y_{i+1})\} \,p_i(y) \, p_i(y_{i+1})$. Dividing both sides by $p_i(y_{i+1})$ gives us the result. The reverse implication requires checking that \eqref{eq:copula_update} satisfies \eqref{eq:martingale_density}, which follows easily through a change of variables with $v = P_{i}(y_{i+1})$.
\subsection{Theorem \ref{Th:absolute_continuity}}
From \citet[Theorem~4]{Berti2013}, we require that 
\begin{equation}
\sup_N E \left[\int_K p_N^2(y) \, dy \right] < \infty
\end{equation}
for all compact $K \subset \mathbb{R}$ in order for $P_\infty$ to be absolutely continuous with respect to the the Lebesgue measure. If this holds, then from \citet[Theorem~1]{Berti2013}, we know that $P_N \to P_\infty$ in total variation and $p_N \to p_\infty$ pointwise and in $L_1$ almost surely, and $p_\infty$ is the density of $P_\infty$ with respect to the Lebesgue measure.

For notational convenience, we assume we are predictive resampling starting at $p_0(y)$. If we look at the second moment conditioned on $y_{1:n}$, we have
\begin{equation}
\begin{aligned}
E\left[p_{n+1}(y)^2 \mid y_{1:n} \right] &=p_n^2(y) \left\{(1- \alpha_{n+1})^2 + 2\alpha_{n+1}(1-\alpha_{n+1}) \,E_{v}\left[c_\rho\{P_{n}(y), v\}\right] \right. \\&+\left. \alpha_{n+1}^2 \,E_{v}\left[c^2_\rho\{P_{n}(y),v\}\right] \right\}
\end{aligned}
\end{equation}
where $v \sim \mathcal{U}[0,1]$. We have that
\begin{equation}
\begin{aligned}
\int_0^1 c_\rho(u,v) \, dv = 1
\end{aligned}
\end{equation}
and 
\begin{equation}
\begin{aligned}
q_\rho(u):= \int_0^1 c^2_\rho(u,v)  \,dv = \frac{\exp\left(\frac{ \rho^2}{1+\rho^2}\Phi^{-1}(u)^2 \right)}{\sqrt{1-\rho^4}}\cdot
\end{aligned}
\end{equation}
So we have that
\begin{equation}
\begin{aligned}
E\left[p^2_{n+1}(y) \mid y_{1:n}\right] = p_n^2(y)[1 -\alpha_{n+1}^2+ \alpha_{n+1}^2 \, q_\rho\{P_n(y)\}].
\end{aligned}
\end{equation}
From Fubini's theorem, we can write
\begin{equation}
\begin{aligned}
E\left[\int_K p^2_{n+1}(y) \, dy\mid y_{1:n}\right] &= \int_K p_n^2(y)\, [1 -\alpha_{n+1}^2+ \alpha_{n+1}^2 q_\rho\{P_n(y)\}]\, dy.\\
\end{aligned}
\end{equation}
Following \citet{Hahn2018} and \cite[Theorem 2.1]{Inglot2010}, we have
$$
\Phi^{-1}(u)^2 \leq -2\log (u \wedge \bar{u})
$$
where $\bar{u} = 1-u$. Using $P_n(y) \geq P_0(y) \prod_{i=1}^n(1-\alpha_i)$ and likewise for $\bar{P}_n$, we can upper bound 
\begin{equation}
\begin{aligned}
q_n\{P_n(y)\} &\leq \frac{1}{\sqrt{1-\rho^4}} \left\{P_n(y) \wedge \bar{P}_n(y) \right\}^{-2\gamma}\\
&\leq  \frac{1}{\sqrt{1-\rho^4}} \left\{P_0(y) \wedge \bar{P}_0(y) \right\}^{-2\gamma}\, \prod_{i=1}^n \left(1-\alpha_i \right)^{-2\gamma}
\end{aligned}
\end{equation}
where $\gamma = \rho^2 /(1+\rho^2)$. As $y\in K$ where $K$ is compact and $P_0(y)$ is continuous, we have that $\left\{P_0(y) \wedge \bar{P}_0(y) \right\}^{-2\gamma}$ is upper bounded by $A < \infty$. Plugging this in, we have
\begin{equation}
\begin{aligned}
E\left[\int_K  p^2_{n+1}(y) \, dy\mid y_{1:n}\right] &\leq  \left\{1 -\alpha_{n+1}^2 + \frac{A}{\sqrt{1-\rho^4}} \,\alpha_{n+1}^2 \prod_{i=1}^n (1-\alpha_i)^{-2\gamma}\right\}
 \int_K p_n^2(y) \, dy.
\end{aligned}
\end{equation}
We have that
\begin{equation*}
\begin{aligned}
\prod_{i=1}^n (1-\alpha_i)^{-2\gamma} &= \prod_{i=1}^n \left\{1 - \left(2-\frac{1}{i}\right)\frac{1}{i+1} \right\}^{-2\gamma}\\
&= \prod_{i=1}^n \left(\frac{i+1}{i-1+\frac{1}{i}} \right)^{2\gamma} \\
&\leq 2^{2\gamma} \prod_{i=2}^n\left(\frac{i + 1}{i-1} \right)^{2\gamma} \\
&= \left\{n(n+1) \right\}^{2\gamma}\\
&\leq (n+1)^{4\gamma}
\end{aligned}
\end{equation*}
Iterating the expectation, we then have
\begin{equation}
\begin{aligned}
 E\left[\int_K p^2_{n+1}(y) \, dy \right]
&\leq \prod_{i=1}^{n+1} \left(1 +\frac{C}{i^\kappa}\right)
\int_K p_0^2(y)\, dy
\end{aligned}
\end{equation}
where $\kappa = 2 - 4\gamma$ and $C< \infty$. We have by assumption that $\int_K p_0^2(y) \, dy$ is bounded. Finally, the product term is monotonically increasing and upper bounded by
\begin{equation}
\begin{aligned}
\prod_{i=1}^\infty \left(1 +\frac{C}{i^\kappa}\right) &= \exp\left\{\sum_{i=1}^\infty \log \left(1+ \frac{C}{i^\kappa} \right) \right\}\\
&\leq \exp \left(C \sum_{i=1}^\infty \frac{1}{i^\kappa} \right)
\end{aligned}
\end{equation}
which is bounded if $4\gamma < 1$ so  $\kappa > 1$. This implies that $\rho < 1/\sqrt{3}$ is required for boundedness.
\subsection{Proposition \ref{prop:concentration}}

Theorem 6.1 from \citet{Chung2006} states that for a martingale $X_i$ relative to $\mathcal{F}_i$, we have
$$
\text{Pr}(X_n- E[X_n] \geq \epsilon) \leq \exp\left\{\frac{-\epsilon^2}{2\left(\sum_{i=1}^n \sigma_i^2 + \frac{M\epsilon}{3}\right)} \right\}
$$
where $|X_i - X_{i-1} | \leq M$ and $\sigma_i^2 := E\left[(X_i - X_{i-1})^2 \mid \mathcal{F}_{i-1} \right] $ for $1 \leq i \leq n$. The original result is by \citet{Mcdiarmid1998}. Considering the martingale $-X_i$ and applying the union bound gives the two-sided inequality:
$$
\text{Pr}(|X_n- E[X_n]| \geq \epsilon) \leq 2\exp\left\{\frac{-\epsilon^2}{2\left(\sum_{i=1}^n \sigma_i^2 + \frac{M\epsilon}{3}\right)} \right\}.
$$
 \\

The cumulative distribution function of the multivariate copula method satisfies
\begin{equation*}
\begin{aligned}
P_{i+1}(\mathbf{y}) =P_{i}(\mathbf{y}) \left( 1- \alpha_{i+1}\right) + \alpha_{i+1} \underbrace{ \int_{-\infty}^\mathbf{y}\prod_{j=1}^d  c_{\rho} \left(u_i^{j},v^j_i\right) p_n(\mathbf{y}') \, d\mathbf{y}'}_{Q_{i+1}(\mathbf{y})}
\end{aligned}
\end{equation*} 
so we have
\begin{equation*}
\begin{aligned}
|P_{i+1}(\mathbf{y}) -P_{i}(\mathbf{y}) |  = \alpha_{i+1} \, |Q_{i+1}(\mathbf{y})  -P_{i}(\mathbf{y})| \leq \alpha_{i+1}
\end{aligned}
\end{equation*} 
for all $\mathbf{y} \in \mathbb{R}^d$ as $Q_{i+1}(\mathbf{y})$ is also a CDF and lies in the interval $(0,1)$.

When predictive resampling, we have that $v_i^ j \sim \mathcal{U}[0,1]$ independently across $j \in \{1,\ldots,d\}$, and from the property of copulas we have
$$
\int_{0}^1 c_\rho(u, v_i^j) \, d v_i^j  = 1.
$$
Defining $\mathcal{F}_i = \sigma(\mathbf{Y}_1,\ldots,\mathbf{Y}_i)$, this implies that
$$
E\left[ Q_{i+1}(\mathbf{y}) \mid \mathcal{F}_i\right] = P_i(\mathbf{y})
$$
almost surely for each $\mathbf{y} \in \mathbb{R}^d$, so $P_{i}(\mathbf{y})$ is a martingale with respect to $\mathcal{F}_i$.\\

\begin{lemma}
For the multivariate copula method, the conditional variance of the martingale satisfies
$$
 E\left[\left\{P_{N+1}({\mathbf{y}})- P_N({\mathbf{y}})\right\}^2 \mid \mathcal{F}_N\right] \leq  \frac{\alpha_{N+1}^2}{4} \quad\mbox{a.s.} 
$$
\end{lemma}
\begin{proof}
We have that
$$
\left\{P_{N+1}({\mathbf{y}})- P_N({\mathbf{y}})\right\}^2 = \alpha_{N+1}^2 \left\{Q_{N+1}(\mathbf{y}) - P_N({\mathbf{y}}) \right\}^2
$$
where 
$$
Q_{N+1}(\mathbf{y})  = \int_{-\infty}^\mathbf{y}\prod_{j=1}^d  c_{\rho} \left(u_{N}^{j},v^j_N\right) p_N(\mathbf{y}')  \, d\mathbf{y}'
$$
 and $Q_{N+1}(\mathbf{y})$ lies in $[0,1]$. Putting this together gives us
\begin{equation*}
\begin{aligned}
E\left[\left\{Q_{N+1}(\mathbf{y}) - P_N({\mathbf{y}}) \right\}^2 \mid \mathcal{F}_N \right] = \text{Var}\left[ Q_{N+1}(\mathbf{y}) \mid  \mathcal{F}_N\right] \leq \frac{1}{4}
\end{aligned}
\end{equation*}
almost surely, which follows from the maximum variance of a random variable on $[0,1]$.
\end{proof}

The martingale is also bounded in difference, as $|P_{M+1}(\mathbf{y}) - P_M(\mathbf{y})| \leq \alpha_{N+1}$ for all $M\geq N$ as $\alpha_N = \left(2-\frac{1}{N}\right)\frac{1}{N+1}$ is monotonically decreasing. McDiarmid's theorem then gives us Proposition \ref{prop:concentration}. As the bound is independent of $\mathbf{y}$, we can take the supremum of both sides.

Numerically this is tighter than Azuma's inequality due to the extra $1/4$ before $\sigma_i^2$. Assuming the sequence $\alpha_{i} \leq 2(i+1)^{-1}$, for $\epsilon = 0.05$ and $N = 5000$, we have $0.42$ for Azuma's inequality and $0.0047$ for McDiarmid's inequality. However, decreasing $\epsilon$ further makes this bound quite loose.

\subsection{Theorem \ref{Th:weak_reg}}
In the regression context, we are interested in the conditional distribution of $Y_n$ given $X_n = x$, so we cannot rely on the c.i.d. result    of \citet{Berti2004}. Fortunately, we can use Theorem 2.2 of \cite{Berti2006} to show that $P_N(\cdot \mid \mathbf{x})$ from predictive resampling converges weakly to a random probability measure almost surely for each $\mathbf{x}\in \mathbb{R}^d$.

We consider the sequence of random probability measures $\{P_N(\cdot \mid \mathbf{x}), P_{N+1}(\cdot \mid \mathbf{x}), \ldots\}$ on $S = \mathbb{R}$ defined on $(\Omega,\mathcal{A},P)$, and begin by showing that for any $f\in C_b(S)$, we have that $P_N(f \mid \mathbf{x})$ converges almost surely, where we define
\begin{equation*}
\begin{aligned}
P_N(f \mid \mathbf{x}) := \int f(y) \, p_N(y \mid \mathbf{x})\, dy.
\end{aligned}
\end{equation*}
This is indeed condition (2) of \citet{Berti2006}. We further write $Z_i = \{Y_i,\mathbf{X}_i \}$ and $\mathcal{F}_i = \sigma(Z_1,\ldots,Z_i)$. Now taking the conditional expectation, we have from Fubini's theorem 
\begin{equation*}
\begin{aligned}
E\left[P_{N+1}(f \mid \mathbf{x}) \mid \mathcal{F}_N\right]&= \int f(y) \, E\left[p_{N+1}(y \mid \mathbf{x}) \mid\mathcal{F}_N\right]\, dy \\
&= P_N(f \mid \mathbf{x})
\end{aligned}
\end{equation*}
almost surely for each $\mathbf{x} \in \mathbb{R}^d$, as $p_N(y \mid \mathbf{x})$ is a martingale with respect to $\mathcal{F}_N$ irrespective of how we draw $\mathbf{X}_{n+1:\infty}$. As $|f(y)|$ is bounded by some $B<\infty$, we also have that
\begin{equation*}
\begin{aligned}
E\left[|P_N(f \mid \mathbf{x})|\right] \leq B
\end{aligned}
\end{equation*}
for all $N$, so $P_N(f \mid \mathbf{x})$ is a martingale with respect to $\mathcal{F}_N$ and converges almost surely. As $\mathbb{R}$ is Radon, Theorem 2.2 of \citet{Berti2006} applies, so there exists a random probability measure $P_\infty$ on $S$, defined on $(\Omega,\mathcal{A},P)$ such that $P_N(\cdot \mid \mathbf{x}) \to P_\infty$ weakly almost surely.

\subsection{Proposition \ref{prop:concentration_reg}}
Following the derivation of Proposition \ref{prop:concentration}, we have that

$$
|P_{N+1}(y \mid \mathbf{x}) - P_N(y \mid \mathbf{x})| \leq \alpha_{N+1}(\mathbf{x},\mathbf{x}_{N+1})
$$
where
 \begin{equation*}
\alpha_{N+1}(\mathbf{x},\mathbf{x}_{N+1}) = \frac{\alpha_{N+1}\prod_{j=1}^d  c_{\rho_{j}} \left\{\Phi\left(x^j\right),\Phi\left(x_{N+1}^j\right)\right\}}{1- \alpha_{N+1} + \alpha_{N+1}\prod_{j=1}^d   c_{\rho_{j}} \left\{\Phi\left(x^j\right),\Phi\left(x_{N+1}^j\right)\right\}}\cdot
\end{equation*}
We have the following lemma:
\begin{lemma}
For the conditional copula method for regression, we have that
$$
\alpha_{N+1}(\mathbf{x},\mathbf{x}_{N+1})  \leq 2C \alpha_{N+1}
$$
where 
$$
C = \prod_{j=1}^d\frac{1}{\sqrt{1-\rho_j^2}}\exp\left\{\frac{x_j^2}{2} \right\}.
$$
\end{lemma}
\begin{proof}
As $c_\rho(u,v) \geq 0$, we have
$$
\alpha_{N+1}(\mathbf{x},\mathbf{x}_{N+1})  \leq \frac{\alpha_{N+1}}{1-\alpha_{N+1}}\prod_{j=1}^d c_{\rho_{j}} \left\{\Phi\left(x^j\right),\Phi\left(x_{N+1}^j\right)\right\}.
$$
We can then write
 \begin{equation*}
 \begin{aligned}
c_\rho\{\Phi(x),\Phi(x')\} &= \frac{1}{\sqrt{1-\rho^2}}\exp\left[-\frac{1}{2(1-\rho^2)}\left\{\rho^2 (x^2 + x'^2) - 2\rho x x' \right\} \right]\\
&\leq \frac{1}{\sqrt{1-\rho^2}}\exp\left\{\frac{x^2}{2} \right\}.
\end{aligned}
\end{equation*}
Finally, noting that $(1-\alpha_{N+1}) \geq 0.5$, we have the result.
\end{proof}
To get the concentration inequality, again writing $Z_i = \{Y_i,\mathbf{X}_i \}$ and $\mathcal{F}_i = \sigma(Z_1,\ldots,Z_i)$. We have that $P_N(y\mid \mathbf{x})$ is a martingale with respect to $\mathcal{F}_N$ so the following lemma and McDiarmid's theorem gives us Proposition \ref{prop:concentration_reg}, where the supremum again follows from the bound being independent of $y$.
\begin{lemma}
For the conditional regression method, the conditional variance of the martingale satisfies
$$
 E\left[\left\{P_{N+1}({y \mid \mathbf{x}})- P_N({y \mid \mathbf{x}})\right\}^2 \mid \mathcal{F}_N\right] \leq  C^2 \alpha_{N+1}^2
$$
almost surely for each $\mathbf{x} \in \mathbb{R}^d$
\end{lemma}
\begin{proof}
We have that
\begin{equation*}
\begin{aligned}
\left\{P_{N+1}({y \mid \mathbf{x}})- P_N({y \mid \mathbf{x}})\right\}^2&= \alpha_{N+1}(\mathbf{x},\mathbf{x}_{N+1})^2 \left\{Q_{N+1}(y \mid \mathbf{x}) - P_N({y \mid \mathbf{x}}) \right\}^2\\
& \leq 4C^2 \alpha_{N+1}^2\left\{Q_{N+1}(y \mid \mathbf{x}) - P_N({y \mid \mathbf{x}}) \right\}^2
\end{aligned}
\end{equation*}
from the above lemma, where 
$$
Q_{N+1}({y} \mid \mathbf{x})  = \int_{-\infty}^{y} c_{\rho}\left(q_N,r_N \right) \, p_N(y' \mid \mathbf{x}) \, dy'
$$
 and $Q_{N+1}({y \mid \mathbf{x}})$ lies in $[0,1]$. From predictive resampling, $r_N \sim \mathcal{U}[0,1]$ conditional on $\mathbf{X}_{N+1}$, so from the tower property we have that $E\left[Q_{N+1}({y} \mid \mathbf{x}) \mid \mathcal{F}_N\right] = P_N(y \mid \mathbf{x})$. Putting this together gives us
\begin{equation*}
\begin{aligned}
E\left[\left\{Q_{N+1}({y}\mid \mathbf{x}) - P_N({{y \mid \mathbf{x}}}) \right\}^2 \mid \mathcal{F}_N \right] = \text{Var}\left[ Q_{N+1}({y\mid \mathbf{x}}) \mid  \mathcal{F}_N\right] \leq \frac{1}{4}
\end{aligned}
\end{equation*}
almost surely for all $\mathbf{x} \in \mathbb{R}^d$, which follows from the maximum variance of a random variable on $[0,1]$.
\end{proof}

\subsection{Theorem \ref{Th:consistency}}\label{Appendix:consistency}
In this section, we prove frequentist consistency of the multivariate copula update \eqref{eq:mv_DP_copdens} in Section \ref{sec:multivariate_copula}. For $\mathbf{y} \in \mathbb{R}^d$, $p_n(\mathbf{y})$ is an estimate of the true density $f_0(\mathbf{y})$  from which we observe samples $\mathbf{Y}_{1:n} \iid f_0(\cdot)$, and we consider $n \to \infty$. For convenience, we restate the following assumptions:
\begin{assumption}\label{assump:a}
We have $\rho \in (0,1)$ and $\alpha_i = a(i+1)^{-1}$ where 
$$a< \frac{2}{5}\cdot$$
\end{assumption}
\begin{assumption}\label{assump:bounded}
There exists $B < \infty$ such that 
$$
\frac{f_0(\mathbf{y})}{p_0(\mathbf{y})} \leq B
$$
for all $\mathbf{y} \in \mathbb{R}^d$.
\end{assumption}
The interpretation of the assumptions are discussed in the main paper. We highlight that the assumptions are slightly different to those in \citet{Hahn2018} as the `almost supermartingale' requires different conditions in the multivariate case. We begin in the same way as \cite{Hahn2018}.  Define the KL divergence as
$$
K(f_0,p_n) = \int \log \frac{f_0(\mathbf{y})}{p_n(\mathbf{y})} \, f_0(\mathbf{y}) \, d\mathbf{y}
$$
and define
$$
T(p) := \int \int \left\{\prod_{j=1}^d c_\rho\left(u^j,v^j\right)-1\right\} f_0(\mathbf{y}) \, f_0(\mathbf{y}') \, d\mathbf{y} \, d\mathbf{y}'
$$
where again we have
$$
u^{j} = P\left(y^j \mid  y^{1:j-1}\right), \quad v^{j} = P\left(y'^j \mid y'^{1:j-1}\right).
$$

The inequality $\log(1+x) \geq x-2x^2,  x\approx 0$ can still be used as $\prod_{j=1}^d c_\rho(u^j,v^j) \geq 0$ and $\alpha_n \to 0$. We can follow through the same algebra to obtain a multivariate version of equation (17) in  \cite{Hahn2018}. Writing $\mathcal{F}_n = \sigma(Y_1,\ldots,Y_n)$, we have
\vspace{2mm}
\begin{equation*}\label{eq:sup_marg}
\begin{aligned}
&E\left\{K(f_0,p_n) \mid \mathcal{F}_{n-1} \right\} - K(f_0,p_{n-1})  \\
&\leq -\alpha_n \int \int \left\{ \prod_{j=1}^d c_\rho\left(u^{j}_{n-1},v^j_{n-1}\right)-1 \right\} f_0(\mathbf{y}) \, f_0(\mathbf{y'})\, d\mathbf{y}  \, d\mathbf{y}'  + E(R_n \mid \mathcal{A}_{n-1})
\end{aligned}
\end{equation*}
 almost surely, where $u^j_{n-1} =P_{n-1}(y^j \mid y^{1:j-1})$ and $v^j_{n-1} = P_{n-1}({y^{j}}' \mid {y^{1:j-1}}')$, and
\begin{equation}
R_n = 2\alpha_n^2 \int \left\{ \prod_{j=1}^d c_\rho\left(u^{j}_{n-1},v^j_{n-1}\right)-1 \right\}^2  f_0(\mathbf{y})\, d\mathbf{y}.
\end{equation}
This is an `almost supermartingale' in the sense of \cite{Robbins1971} if $T(p_{n})$ is positive  and $\sum_n E[R_n \mid \mathcal{F}_{n-1}] < \infty$ almost surely.

\begin{lemma}\label{lemma:T_pos}
For a density $p$ with support containing that of $f_0(\mathbf{y})$, we have that $T(p) \geq 0$ with equality if and only if $p =f_0$ Lebesgue-almost everywhere.
\begin{proof}
First we note that the copula density product can be written as
\begin{equation*}
\prod_{j=1}^d c_\rho(u^j,v^j) = \int  \prod_{j=1}^d \psi_{\theta^j}(u^j) \, \psi_{\theta^j}(v^j) \, \mathcal{N}(\theta^j \mid 0,\rho) \, d\boldsymbol{\theta}
\end{equation*}
where ${\boldsymbol\theta} = \{\theta^1,\ldots,\theta^d\}$ and
\begin{equation*}
\psi_\theta(u) =  \frac{\mathcal{N}(\Phi^{-1}(u) \mid \theta,1-\rho)}{\mathcal{N}(\Phi^{-1}(u) \mid 0,1)}\cdot
\end{equation*}
This gives us
\begin{equation*}
\begin{aligned}
T(p) &= \int \int \left\{ \prod_{j=1}^d c_\rho\left(u^{j},v^j\right)-1 \right\} f_0(\mathbf{y}) \, f_0(\mathbf{y'})\, d\mathbf{y} \, d\mathbf{y}'
\\&=   \int \left[\left\{\int \prod_{j=1}^d  \psi_{\theta^j}(u^j) f_0(\mathbf{y})\, d\mathbf{y} \right\}^2 -1 \right] \prod_{j=1}^d \mathcal{N}(\theta^j \mid 0,\rho) \, d\boldsymbol{\theta}\\
&=  \int \left\{\int \prod_{j=1}^d  \psi_{\theta^j}(u^j) f_0(\mathbf{y})\, d\mathbf{y} -1 \right\}^2 \prod_{j=1}^d \mathcal{N}(\theta^j \mid 0,\rho) \, d\boldsymbol{\theta}.
\end{aligned}
\end{equation*}
 The last line follows from $E(X^2) - E^2(X) = E\{ X-E(X)\}^2$ as 
 \begin{equation*}
 \prod_{j=1}^d \int  \psi_{\theta^j}(u^j) \, \mathcal{N}(\theta^j \mid 0,\rho) \, d\theta^j = 1
 \end{equation*}
 for all $\mathbf{u} = \{u^{1},\ldots,u^d\} \in [0,1]^d$ . The above shows that $T(p)\geq 0$. \\

For the second equality result, note that $T(p)= 0$ if and only if  $\int \prod_{j=1}^d  \psi_{\theta^j}(u^j) \, f_0(\mathbf{y}) \,  d\mathbf{y} = 1 $  for Lebesgue-almost all ${\boldsymbol\theta}$. This can be strengthened to all $\boldsymbol{\theta} \in \mathbb{R}^d$ from the continuity of $\int \prod_{j=1}^d  \psi_{\theta^j}(u^j) \, f_0(\mathbf{y}) \, d\mathbf{y}$, which follows from the continuity of $\psi_\theta(u)$, the upper bound
$$
\psi_\theta(u) \leq \frac{1}{\sqrt{1-\rho}}\exp\left(\frac{\theta^2}{2\rho} \right)
$$
and dominated convergence.\vspace{1mm}

To show that $\int \prod_{j=1}^d  \psi_{\theta^j}(u^j) \, f_0(\mathbf{y}) \,  d\mathbf{y} = 1$  for all $\boldsymbol{\theta}$ holds if and only if $p=f_0$, factorize $f_0(\mathbf{y}) = \prod_{j=1}^d {f_0}^j(y^j)$ where ${f_0}^j(y^j) = f_0(y^j \mid y^{1:j-1})$. We carry out a multivariate change of variables from $\mathbf{y}$ to $\mathbf{z} = \{z^1,\ldots,z^j\}$, where $z^j = \Phi^{-1}(u^j)$, and note the Jacobian is triangular. This gives
\begin{equation*}\label{eq:full_prod}
\begin{aligned}
\int \prod_{j=1}^d  \frac{{f_0}^j\left\{({P^j})^{-1}(\Phi(z^j)) \right\}}{p^j\left\{({P^j})^{-1}(\Phi(z^j))\right\}} \, \mathcal{N}(z^j \mid \theta^j,1-\rho) \, d\mathbf{z}
\end{aligned}
\end{equation*}
 where $({P^j})^{-1}$ is the inverse CDF for $P(y^j \mid y^{1:j-1})$, $p^j(y^j)= p(y^j \mid y^{1:j-1})$, and each ratio term with $z^j$ depends on $z^{1:j-1}$. It is clear that $T(f_0) = 0$. We now want to show that $T(p) = 0$ implies the above density ratio is 1 almost everywhere.  
 
 To do so, we point out that the multivariate normal location family $\mathcal{N}\{\mathbf{z};\boldsymbol \theta, (1-\rho)I_d\}$ is complete from \citet[Theorem 4.3.1]{Lehmann2006}, that is $E[g(\mathbf{z})] = 0$  for all $
 \boldsymbol{\theta} \in \mathbb{R}^d$ implies $g(\mathbf{z}) =0$ Lebesgue-almost everywhere. From this, $\int \prod_{j=1}^d  \psi_{\theta^j}(u^j) \, f_0(\mathbf{y}) \,  d\mathbf{y} = 1$  for all $\boldsymbol{\theta}$ implies
 \begin{equation*}
  \prod_{j=1}^d  \frac{{f_0}^j\left\{({P^j})^{-1}(\Phi(z^j)) \right\}}{p^j\left\{({P^j})^{-1}(\Phi(z^j))\right\}} = 1
 \end{equation*}
 for Lebesgue-almost all $\mathbf{z}$, so $f_0 = p$ holds Lebesgue-almost everywhere as the product of the conditionals is the joint.
\end{proof}
\end{lemma}
We highlight this above lemma to show that $T(p)$ has the makings of a probability divergence, which we will use later. We now prove the second requirement for the almost super-martingale.
\begin{lemma}\label{lemma:R_bounded}
Under the assumptions above, we have that
$$
\sum_n E\left[R_n \mid \mathcal{F}_{n-1}\right] <\infty 
$$
almost surely.
\begin{proof}
We only need to bound
$$
 \zeta_n = \int \int \prod_{j=1}^d c_\rho\left(u_{n-1}^j, v_{n-1}^j \right)^2 f_0(\mathbf{y}) \, f_0(\mathbf{y}') \, d\mathbf{y} \, d\mathbf{y}'.
$$
Following the univariate proof of \cite{Hahn2018}, we have from the mixture representation of the copula and Cauchy-Schwarz that
$$
c_\rho(u,v)^2 \lesssim \exp\left(\lambda z_u^2 \right) \, \exp\left(\lambda z_v^2 \right)
$$
where $\lesssim$ indicates the inequality up to a constant, $\lambda = \rho/(1+\rho)$ and $z_u = \Phi^{-1}(u)$. This then gives us
\begin{equation*}
\begin{aligned}
 \zeta_n &\lesssim \left[\int \prod_{j=1}^d \exp\left\{{\lambda (z^j)^2}\right\}f_0(\mathbf{y}) \,d\mathbf{y}\right]^2.
 \end{aligned}
 \end{equation*}
Again applying a change of variables from $\mathbf{y}$ to $\mathbf{z} = \{z^1,\ldots,z^d\}, z^j = \Phi^{-1}(u_{n-1}^j)$ as before,  we can write
\begin{equation*}
\begin{aligned}
\int \prod_{j=1}^d \exp\left\{{\lambda (z^j)^2}\right\}f_0(\mathbf{y}) \, d\mathbf{y} &\propto \int \exp\left\{{\left(\lambda-\frac{1}{2}\right) \sum_{j=1}^d (z^j)^2}\right\} \frac{f_0(\mathbf{y})}{p_{n-1}(\mathbf{y})} \, d \mathbf{z}\\
&= \int \exp\left\{{\left(\lambda-\frac{1}{2}\right) \sum_{j=1}^d (z^j)^2}\right\}\frac{f_0(\mathbf{y})}{p_{0}(\mathbf{y})\prod_{i=1}^{n-1}(1-\alpha_i+\alpha_i \prod_{j=1}^d c_{i,j})} \, d \mathbf{z}\\&\leq \prod_{i=1}^{n-1} (1-\alpha_i)^{-1} \int \exp\left\{{\left(\lambda-\frac{1}{2}\right) \sum_{j=1}^d (z^j)^2}\right\}\frac{f_0(\mathbf{y})}{p_{0}(\mathbf{y})} \, d\mathbf{z}
\\&\leq B\prod_{i=1}^{n-1} (1-\alpha_i)^{-1} \int \exp\left\{{\left(\lambda-\frac{1}{2}\right) \sum_{j=1}^d (z^j)^2}\right\} \, d\mathbf{z}.
\end{aligned}
\end{equation*}
 The third line follows from $c_{i,j}:= c_\rho\left(u_{i-1}^j, v_{i-1}^j \right) \geq 0$, and the last line from Assumption \ref{assump:bounded}. As $\lambda < \frac{1}{2}$, we have that $\int \exp\left\{{\left(\lambda-\frac{1}{2}\right) \sum_{j=1}^d (z^j)^2}\right\} d\mathbf{z}$ is bounded. Following through, we need to show
\begin{equation}\label{eq:consistency_sum}
\sum_{n}^\infty \alpha_n^2 \left\{ \prod_{i=1}^{n-1} (1-\alpha_i)^{-2}\right\}
\end{equation}
converges as in original proof. The product term on the right can be written as
\begin{equation*}
\begin{aligned}
\prod_{i=1}^{n-1} (1-\alpha_i)^{-2} &= \exp\left\{-2\sum_{i=1}^{n-1}\log(1-\alpha_i) \right\}\\
&\leq \exp\left\{\frac{2}{1-\alpha_1}\sum_{i=1}^{n-1}\alpha_i \right\}\\
&\leq \exp\left\{\frac{2a}{1-a/2}\sum_{i=1}^{n-1}(i+1)^{-1} \right\}\\
&\leq n^{\frac{4a}{2-a} }
\end{aligned}
\end{equation*}
where the last line follows from $\sum_{i=1}^n i^{-1} \leq \log n + 1$. Finally, for the sum in \eqref{eq:consistency_sum} to be finite, we just require
$$
\frac{4a}{2-a} \leq 1
$$
which is satisfied for $a < 2/5$ as in Assumption \ref{assump:a}.
\end{proof}
\end{lemma}

We now have the first main result, which is similar to that in \cite{Hahn2018}.
\begin{theorem}\label{Th:consistency1}
Let $p_n$ satisfy the update \eqref{eq:mv_DP_copdens} for $\mathbf{Y}_{1:n} \iid f_0(\mathbf{y})$ where $f_0$ is continuous. Under the assumptions above, we have that the following holds almost surely:
$$
K(f_0,p_n) \to K_\infty  \quad \textnormal{and} \quad \sum_{n=1}^\infty \alpha_n \,  T(p_{n-1}) < \infty,
$$
where $K_\infty$ is a random variable. 
\begin{proof}
This follows directly from \citet[Theorem 1]{Robbins1971} as we have shown the conditions required in our Lemmas \ref{lemma:T_pos}, \ref{lemma:R_bounded}.
\end{proof}
\end{theorem} 

It is not straightforward to show that $K_\infty = 0$ almost surely, even in the univariate case as claimed in \citet{Hahn2018}, as their proof by contradiction requires that
$$
K_\infty >0  \implies \liminf_n \,  T(p_{n-1}) > 0
$$ 
which is nontrivial to verify. Nonetheless, we have opted to retain the result that $K_\infty$ exists for completion, but we emphasize that this result is not necessary for the remaining proofs. We now deviate from the original proof and provide the  details missing from \cite{Hahn2018}.

We will rely on the second result of Theorem \ref{Th:consistency1}, that is the almost sure boundedness of $\sum_{n=1}^\infty \alpha_n \,  T(p_{n-1})$, to prove Hellinger consistency. These are new additional steps that depend on the divergence-like properties of $T(p)$. To proceed, we require two lemmas.
\begin{lemma}\label{lemma:liminfT}
From Theorem \ref{Th:consistency1}, we have that
$$
\liminf_n T(p_n) = 0
$$
almost surely.
\begin{proof}
For contradiction, assume that $\liminf_n \,  T(p_n) = \delta > 0$. Picking $\epsilon < \delta$, there exists $N$ such that for all $n>N$, we have
$$
T(p_n) > \delta - \epsilon.
$$ 
However, this implies that 
$$
\sum_{n=1}^\infty \alpha_n \,  T(p_{n-1}) \geq C + (\delta - \epsilon)\sum_{n=N+1}^\infty \alpha_n = \infty
$$
which is a contradiction as $\sum_{n=1}^\infty \alpha_n \, T(p_{n-1}) <\infty$.
\end{proof}
\end{lemma}
Let us now consider the density on $\mathbf{u}=\{u^1,\ldots,u^d\} \in [0,1]^d$ defined
\begin{equation}\label{eq:g_n}
g_n(\mathbf{u}) =\prod_{j=1}^d\frac{ f_0^j\{{(P_n^j)}^{-1}(u^j)\} }{p^j_n\{{(P_n^j)}^{-1}(u^j)\}}
\end{equation}
where  as a reminder
\begin{equation*}
\begin{aligned}
f_0^j(y^j) &= f_0(y^j \mid y^{1:j-1})\\
p_n^j(y^j) &= p_n(y^j \mid y^{1:j-1})\\P_n^j(y^j) &= P_n(y^j \mid y^{1:j-1})\\
 {\left(P_n^j\right)}^{-1}(u^j)&= P_n^{-1}(u^j \mid y^{1:j-1}).
\end{aligned}
\end{equation*}
We can  write $T(p_n)$ as a function of $g_n$ through a change of variables from $y^j$ to $u^j = P^j_n(y^j )$, and similarly for ${y^j}'$ to $v^j$, for $j = 1,\ldots, d$, that is
$$
T(g_n) =\int  \int \prod_{j=1}^d c_\rho(u^j,v^j)\, g_n(\mathbf{u})\, g_n(\mathbf{v})\, d\mathbf{u} \, d\mathbf{v}-1.
$$ 
Lemma \ref{lemma:liminfT} gives us the following.
\begin{lemma}\label{lemma:liminff}
The sequence of densities $g_n$ satisfies
$$
g_\infty(\mathbf{u}) := \liminf_n \, g_n(\mathbf{u}) = 1
$$
for Lebesgue-almost all $\mathbf{u}\in [0,1]^d$ almost surely.
\begin{proof}
Repeated use of Fatou's lemma and the fact that $\liminf_n x^2_n = (\liminf_n x_n)^2$ for $x_n \geq 0$ gives us
\begin{equation*}
\begin{aligned}
T(g_\infty) &= E_{\boldsymbol{\theta}}\left[ \left(\int \prod_{j=1}^d\psi_{\theta^j}(u^j) \, g_\infty(\mathbf{u})\,d\mathbf{u} \right)^2 \right] - 1\\
& \leq E_{\boldsymbol{\theta}}\left[  \left(\liminf_n \int \prod_{j=1}^d\psi_{\theta^j}(u^j)  \, g_n(\mathbf{u}) \, d\mathbf{u} \right)^2 \right] - 1\\
& = E_{\boldsymbol{\theta}}\left[ \liminf_n \left(\int \prod_{j=1}^d\psi_{\theta^j}(u^j)  \, g_n(\mathbf{u})\, d\mathbf{u} \right)^2 \right] - 1\\
& \leq \liminf_n E_{\boldsymbol{\theta}}\left[  \left(\int \prod_{j=1}^d\psi_{\theta^j}(u^j) \, g_n(\mathbf{u})\, d\mathbf{u} \right)^2 \right] - 1\\
&= \liminf_n \, T(g_n) = 0.
\end{aligned}
\end{equation*}
As $T(g_\infty)$ is non-negative, it is equal to 0. From the original proof then, $T(g_\infty) = 0$ implies $g_\infty(\mathbf{u}) = 1$ Lebesgue-almost everywhere from Lemma \ref{lemma:T_pos}.
\end{proof}
\end{lemma}

We now require the squared Hellinger distance between probability density functions $g_1$ and  $g_2$ on $\mathbf{u}\in [0,1]^d$, which is defined
$$
 H^2(g_1,g_2) := 1- \int \sqrt{g_1(\mathbf{u}) \, g_2(\mathbf{u})} \, d\mathbf{u}.
$$
A straightforward lemma follows from this.
\begin{lemma} \label{lemma:hellinger}
The density $p_n$ is Hellinger consistent at $f_0$ if and only if $g_n$ in \eqref{eq:g_n} is Hellinger consistent at the uniform density on $[0,1]^d$. 
 \begin{proof}
 Through a change of variables from $u_j$ to $y_j =  {(P_n^j)}^{-1}(u^j)$ for $j = 1,\ldots,d$, it is simple to show that
\begin{equation*}
\begin{aligned}
H^2(g_n,1) &= 1-\int \sqrt{g_n(\mathbf{u})}\,  d\mathbf{u} \\&=  1- \int \sqrt{f_0(\mathbf{y})\, p_n(\mathbf{y})} \,d\mathbf{y} \\
&= H^2(p_n,f_0).
\end{aligned}
\end{equation*}
\end{proof}
\end{lemma}
We now have the final main result.
\begin{theorem}\label{Th:consistency2}
The density $g_n$ converges in Hellinger distance to the uniform, that is
$$
\lim_n H^2(g_n,1) = 0
$$
almost surely.
\begin{proof}
The limit superior of the Hellinger distance is 
\begin{equation*}
\begin{aligned}
\limsup_n H^2(g_n,1) &= 1- \liminf_n \int \sqrt{g_n(\mathbf{u})}\, d\mathbf{u}.
\end{aligned}
\end{equation*}
From Fatou's lemma and the fact that $\liminf_n \sqrt{x_n} = \sqrt{\liminf_n x_n}$ for $x_n \geq 0$, we have
\begin{equation*}
\begin{aligned}
\liminf_n \int \sqrt{g_n(\mathbf{u})}\, d\mathbf{u} \geq \int \sqrt{  \liminf_n g_n(\mathbf{u})}  \, d\mathbf{u} = 1.
\end{aligned}
\end{equation*}
So we have
\begin{equation*}
\begin{aligned}
\limsup_n H^2(g_n,1) =0
\end{aligned}
\end{equation*}
and $0\leq\liminf_n H^2(g_n,1) \leq \limsup_n H^2(g_n,1)$, which gives us the result.
\end{proof}
\end{theorem}

From Lemma \ref{lemma:hellinger} and Theorem \ref{Th:consistency2}, we have Theorem \ref{Th:consistency} in the main paper as desired.

\section{Copula derivations}\label{Appendix:copula_deriv}

\subsection{Density estimation}

\subsubsection{The sequence of weights}\label{Appendix:alpha_deriv}

We now derive the sequence of weights $\alpha_i$ for the univariate copula update. The actual copula update for $n>1$ for the posterior DP mixture is
\begin{equation*}
p_{n+1}(y) =p_n(y) \,  \frac{\int p(y \mid G)\, p(y_{n+1} \mid G) \, d\pi(G \mid y_{1:n})}{p_n(y) \,  p_n(y_{n+1})} 
\end{equation*}
where $\pi(G \mid y_{1:n})$ is a mixture of Dirichlet processes, that is
\begin{equation*}
\begin{aligned}
[G \mid \theta_{1:n},y_{1:n}] &\sim \text{DP}\left(a + n, \frac{a G_0 + \sum_{i=1}^n \delta_{\theta_i}}{a + n}\right)\\
[\theta_{1:n} \mid y_{1:n}] &\sim \pi(\theta_{1:n} \mid y_{1:n}).
\end{aligned}
\end{equation*}
Usually, samples from the posterior over the means of the cluster assignments $\pi(\theta_{1:n} \mid y_{1:n})$ are obtained through Gibbs sampling. For tractability, we need to modify the term $\pi(\theta_{1:n} \mid y_{1:n})$.  Let us instead assume that each cluster mean is drawn independently from the prior, so we can write
$$
\pi(\theta_{1:n} \mid y_{1:n}) = \prod_{i=1}^n G_0(\theta_i).
$$
Now computing the integral term, we have
\begin{equation}\label{eq:alpha_first}
\begin{aligned}
\int p(y \mid G) \, p(y_{n+1} \mid G) \,  d\pi(G \mid y_{1:n}) = E\left[\int p(y \mid G) \,  p(y_{n+1} \mid G)\,   d\pi(G \mid \theta_{1:n}) \right]
\end{aligned}
\end{equation}
where the expectation is over $\theta_{1:n} \sim \prod_{i=1}^n G_0(\theta_i)$. We can use the stick-breaking construction of the DP for the term inside the integral to get the familiar form (see Section \ref{Appendix:multivariate}). We have the inner term 
\begin{equation}\label{eq:alpha_inner}
\begin{aligned}
\int p(y \mid G) \, p(y_{n+1} \mid G) \, d\pi(G \mid \theta_{1:n})  = &\left(1- \frac{1}{a + n+1}\right) \int  K(y \mid \theta) \,  dG_n(\theta) \int K(y_{n+1} \mid \theta') \,  dG_n(\theta') \\ +&\frac{1}{a + n+1}\int  K(y \mid \theta) K(y_{n+1} \mid \theta)\, dG_n(\theta)
\end{aligned}
\end{equation}
where we write $K(y \mid \theta) = \mathcal{N}(y \mid \theta,1)$. Here, $G_n$ is random and defined as
 \begin{equation*}
\begin{aligned}
G_n &= \frac{a G_0 + \sum_{i=1}^n \delta_{\theta_i}}{a + n} \\
\theta_{1:n} &\sim \prod_{i=1}^n G_0(\theta_i).
\end{aligned}
\end{equation*}

Taking expectation of the first term in \eqref{eq:alpha_inner}, we can write
 \begin{equation*}
\begin{aligned}
E\left[\int  K(y \mid \theta) \,dG_n(\theta) \int K(y_{n+1} \mid \theta')\, dG_n(\theta') \right] &= \frac{a^2}{(a + n)^2} \left\{\int K(y \mid \theta) \,dG_0(\theta)  \int K(y_{n+1} \mid \theta') \, dG_0(\theta')    \right\} \\ &+\frac{a}{(a+ n)^2}\left\{\int K(y \mid \theta)\, dG_0(\theta) \,\sum_{i=1}^n E\left[  K(y_{n+1} \mid \theta_i)\right] \right\} \\
&+\frac{a}{(a+ n)^2}\left\{\int K(y_{n+1} \mid \theta) \, dG_0(\theta) \,\sum_{i=1}^n E\left[  K(y \mid \theta_i)\right] \right\} \\
&+ \frac{1}{(a+n)^2} \sum_{i=1}^n \sum_{j=1}^n E\left[ K(y \mid \theta_i) \,K(y_{n+1} \mid \theta_j) \right]
\end{aligned}
\end{equation*}
We use the fact that $\theta_i \sim G_0$ and $\theta_i,\theta_j$ are independent for $i \neq j$ to simplify the above to:
 \begin{equation*}
\begin{aligned}
 \frac{a^2 + 2n a + (n^2 - n) }{(a + n)^2 } \left\{\int K(y \mid \theta)\, dG_0(\theta)  \int K(y_{n+1} \mid \theta') \,dG_0(\theta')    \right\} \\+ \frac{n}{(a+n)^2} \int  K(y \mid \theta) K(y_{n+1} \mid \theta)\, dG_0(\theta).
\end{aligned}
\end{equation*}
Taking expectation of the second term in \eqref{eq:alpha_inner} is much simpler:
\begin{equation*}
\begin{aligned}
E\left[ \int  K(y \mid \theta) \, K(y_{n+1} \mid \theta)\,dG_n(\theta) \right] = \int  K(y \mid \theta) \,K(y_{n+1} \mid \theta) \,dG_0(\theta).
\end{aligned}
\end{equation*}
Now plugging this back into \eqref{eq:alpha_first}, we have
\begin{equation*}
\begin{aligned}
&E\left[\int p(y \mid G) \,p(y_{n+1} \mid G) \,d\pi(G \mid \theta_{1:n}) \right]  \\&=  \left\{\frac{a+n }{a + n+1} \right\}\left\{ \frac{a^2 + 2n a + (n^2 - n) }{(a + n)^2 }\right\}\left\{\int K(y \mid \theta) \,dG_0(\theta)  \int K(y_{n+1} \mid \theta') \,dG_0(\theta')    \right\}\\&+ \underbrace{\left\{ \frac{1}{a + n+1} +\left\{\frac{a+n }{a + n+1} \right\} \frac{n}{(a+n )^2}\right\}}_{\alpha_{n+1}} \int  K(y \mid \theta) \, K(y_{n+1} \mid \theta)\,dG_0(\theta).
\end{aligned}
\end{equation*}
This suggests the update
\begin{equation*}
\begin{aligned}
p_{n+1}(y) = p_n(y)\left[1- \alpha_{n+1} + \alpha_{n+1} c_\rho\left\{P_n(y),P_n(y_{n+1}) \right\} \right]
\end{aligned}
\end{equation*}
where for $a = 1$, we have
$$
\alpha_n = \frac{1}{n+1} + \frac{n-1}{n(n+1)} =\left( 2 - \frac{1}{n}\right)\frac{1}{n+1}\cdot
$$
The intuitive reasoning for this discrepancy from the usually suggested $(i+1)^{-1}$ is due to the mixing over the atoms of $G_n$. Note that for $n = 1$, we still have $\alpha_1 = 0.5$, so the first copula update step still agrees with the DP mixture.

The assumption of
$
\pi(\theta_{1:n} \mid y_{1:n}) = \prod_{i=1}^n G_0(\theta_i)
$
is the only simplification required to get the copula update with the above $\alpha_i$ exactly, where $\theta_{1:n}$ are the means of the cluster allocations. For the DPMM, there are usually ties in the posterior samples of $\theta_{1:n}$, so by assuming all $\theta_i \iid G_0$, we have allocated each $y_i$ to its own cluster. In that sense, the copula update can be viewed as a mixture model where we allocate a new cluster for each data point, similar to the KDE.

\subsubsection{Multivariate copula method} \label{Appendix:multivariate}
In this section, we derive the copula update for the multivariate DPMM, focussing on just the first step. One could also follow the argument of Section \ref{Appendix:alpha_deriv} to return the same update with the specific form for $\alpha_i$. The multivariate DPMM with factorized kernel has the form
\begin{equation*}
\begin{aligned}
f_G(\mathbf{y}) = \int \prod_{j=1}^d\mathcal{N}(y^j \mid \theta^j,1) \, dG(\bm{\theta}),\quad 
 G \sim \text{DP}\left(a, G_0 \right), \quad G_0(\bm{\theta}) = \prod_{j=1}^d\mathcal{N}(\theta^j \mid 0,\tau^{-1}).
\end{aligned}
\end{equation*}
Following the example in \cite{Hahn2018} and \eqref{eq:Bayes_cop}, we want to compute the copula density for the first update step of the DPMM, that is
\begin{equation}\label{eq:mv_num}
\begin{aligned}
\frac{E\left[ f_G(\mathbf{y})\, f_G(\mathbf{y}_1) \right]}{p_0(\mathbf{y})\, p_0(\mathbf{y}_1)}\cdot
\end{aligned}
\end{equation}
From the stick-breaking representation of the DP, we can write $G$ as
$$
G = \sum_{k=1}^\infty w_k \, \delta_{\bm{\theta}^*_k}
$$
where $w_k = v_k \prod_{j< k} \{1-v_j\}$, $v_k \iid \text{Beta}(1,a)$ and $\bm{\theta}^*_k \iid G_0$.
We can then write the numerator as
\begin{equation*}
\begin{aligned}
&E\left[ \sum_{j=1}^\infty \sum_{k=1}^\infty  w_j \,w_k\, K(\mathbf{y} \mid \bm{\theta}^*_j)\, K(\mathbf{y}_1 \mid \bm{\theta}^*_k) \right] \\&=\left(1-E\left[ \sum_{k=1}^\infty w_k^2\right]\right)E\left[ K(\mathbf{y} \mid \bm{\theta}^*) \right] E\left[ K(\mathbf{y}_1 \mid \bm{\theta}^*) \right] + E\left[ \sum_{k=1}^\infty w_k^2\right] E\left[ K(\mathbf{y} \mid \bm{\theta}^*)\, K(\mathbf{y}_1 \mid \bm{\theta}^*) \right] 
\end{aligned}
\end{equation*}
where we have used the fact that $\sum_{k=1}^\infty w_k = 1$ almost surely. Here, $\theta^* \sim G_0$ and we have written 
$$
K(\mathbf{y} \mid \bm{\theta}^*)  = \prod_{j=1}^d K(y^j \mid \theta^{*j}), \quad K(y^j \mid \theta^{*j}) = \mathcal{N}(y^j \mid \theta^{*j},1).
$$
It is easy to show that
$$
\alpha_1 := E\left[ \sum_{k=1}^\infty w_k^2\right] = \frac{1}{1+a} \quad \textnormal{a.s.}
$$
As $p_0(\mathbf{y}) = E\left[ K(\mathbf{y} \mid \bm{\theta}^*) \right]$, we have that \eqref{eq:mv_num} can be written as
$$
1-\alpha_1 + \alpha_1 \frac{E\left[ K(\mathbf{y} \mid \bm{\theta}^*) \, K(\mathbf{y}_1 \mid \bm{\theta}^*) \right] }{p_0(\mathbf{y}) \, p_0(\mathbf{y}_1)} \cdot
$$
We note that the kernel $K$ factorizes with independent priors on each dimension, and $p_0(\mathbf{y}) = \prod_{j=1}^d p_0(y^j) = \prod_{j=1}^d \mathcal{N}(y^j \mid 0,1 +\tau^{-1})$, so
\begin{equation*}
\begin{aligned}
\frac{E\left[ K(\mathbf{y} \mid \bm{\theta}^*) \, K(\mathbf{y}_1 \mid \bm{\theta}^*) \right] }{p_0(\mathbf{y})\, p_0(\mathbf{y}_1)}  = \prod_{j=1}^d \frac{E\left[ K(y^j \mid \theta^{*j}) \, K({y}^j_1 \mid {\theta^{*j}}) \right] }{p_0({y}^j) \, p_0({y}^j_1)}\cdot 
\end{aligned}
\end{equation*}
Finally, we can compute each univariate term
\begin{equation*}
\begin{aligned}
\frac{E\left[ K(y \mid \theta^{*}) \,K({y}_1 \mid {\theta}^*) \right] }{p_0({y}) \,p_0({y}_1)} &= \frac{\mathcal{N}_2(\Phi^{-1}(u),\Phi^{-1}(v) \mid 0, 1,\rho)}{\mathcal{N}(\Phi^{-1}(u) \mid 0,1 )\mathcal{N}(\Phi^{-1}(v) \mid 0,1)}\\
\end{aligned}
\end{equation*}
where $\mathcal{N}_2(\cdot \mid 0,1,\rho)$  is the bivariate normal density with mean 0, variance 1 and correlation $\rho = 1/(1+\tau)$, and $u = P_0(y), v= P_0(y_1)$. This is of course exactly the Gaussian copula density $c_\rho(u,v)$.  Putting the above together gives us the copula update
\begin{equation*}
p_{1}(\mathbf{y}) = \left[1-\alpha_1 + \alpha_1 \prod_{j=1}^d c_\rho\left\{P_0(y^j),P_0(y_1^j)\right\}\right] p_0(\mathbf{y}).
\end{equation*}

\subsubsection{Categorical data}\label{Appendix:copula_discrete}
In this section, we derive a copula type update for categorical data. For $y$ on a countable space, we can derive a copula-type update with the DP prior, that is
$$
f_G(y)  = g(y), \quad G \sim \text{DP}(a,G_0)
$$
where $g$ is the probability mass function of $G$, and likewise for $g_0$ and $G_0$. The predictive probability mass function is
$$
p_n(y) = \frac{a g_0(y) + T_y^n}{a+n}\cdot
$$
where $T_y^n = \sum_{i=1}^n \mathbbm{1}(y_i = y)$. Following a similar calculation of the Dirichlet-categorical model in \cite{Hahn2018}, we have that 
\begin{equation*}
\begin{aligned}
d_\rho(y,y_1) = \frac{p_1(y)}{p_{0}(y)} &= \frac{a}{a+1}\left(1 + \frac{\mathbbm{1}(y= y_1)}{a g_0(y)} \right)\\
&= 1-\rho + \rho \frac{\mathbbm{1}(y= y_1)}{g_0(y)}
\end{aligned}
\end{equation*}
where we have used $p_0(y) = g_0(y)$ and $\rho = 1/(a+1)$. We can then compute
\begin{equation*}
\begin{aligned}
 D_\rho\{P_0(y), P_0(y_1)\}= P(Y \leq y, Y_1 \leq y_1) &= \sum_{z\leq y, z' \leq y_1} d_\rho(z,z')\,p_0(z)\,p_0(z') \\
&=(1-\rho)P_0(y)P_0(y_1) + \rho \left\{P_0(y) \wedge P_0(y_1)\right\} .
\end{aligned}
\end{equation*}
which again is the mixture of the independent and Fr\'{e}chet-Hoeffding copula.

For our updates, we can rewrite $d_\rho(y,y_1)$ as a function of $P_0,p_0$. Although $p_0(Y = k) = p_0(Y_1 = k)$ in this context, we need to keep the terms $ p_0(y), p_0(y_1)$ separate in anticipation of the multivariate case, where this equality may not hold as we will be working with conditionals $p_n(y^j \mid y^{1:j-1})$. 

 Using the above, we have
\begin{equation}\label{eq:DP_countable}
\begin{aligned}
d_\rho(y,y_1) = &\left\{[D_\rho\{P_0(y), P_0(y_1)\} - D_\rho\{P_0(y), P_0(y_1-1)\}] \right.\\
& \left. - [D_\rho\{P_0(y-1), P_0(y_1)\} - D_\rho\{P_0(y-1), P_0(y_1-1)\}]\right\}/\{p_0(y) \, p_0(y_1)\}.
\end{aligned}
\end{equation}
Here, $d_\rho(y,y_1)$ is the difference quotient of $D_\rho$, in a similar way $c_\rho(u,v)$ is the derivative of $C_\rho(u,v)$ for the continuous case. The update for the density and distribution function is then
\begin{equation*}
\begin{aligned}
p_1(y) &= d_\rho(y,y_1) \, p_0(y) \\
P_1(y ) &= [D_\rho\{P_0(y), P_0(y_1)\} -D_\rho\{P_0(y), P_0(y_1-1)\}]/p_0(y_1).
\end{aligned}
\end{equation*}
For the categorical case where $y \in \{1,\ldots,K\}$, we would have $P_0(y) = 0$ for $y <1$ and $P_0(y) = 1$ for $y \geq K$.

For mixed data where some dimensions of $\mathbf{y}$ may be discrete, the conditional factorization allows for an easy extension of the a multivariate mixed copula method. We simply substitute the bivariate Gaussian copula density $c_\rho$ with $d_\rho$ in \eqref{eq:mv_DP_copdens} for the respective discrete dimensions, where we may have different bandwidths for the discrete and continuous data. However, in the discrete case, obtaining the martingale posterior may be more computationally difficult as we do not have the property $P_i(y_{i+1}) \iid \mathcal{U}[0,1]$ as in the continuous case.

\subsection{Regression}
\subsubsection{Martingale}
We now show that predictive resampling in the regression context gives us a martingale. For the update 
$$
p_{i+1}(y \mid \mathbf{x}) = \left\{1-\alpha_{i+1}(\mathbf{x},\mathbf{x}_{i+1})+ \alpha_{i+1}(\mathbf{x},\mathbf{x}_{i+1}) \,c_{\rho_y}\left(q_i,r_i\right)\right\} p_i(y\mid \mathbf{x})
$$
where
$$
q_i = P_{i}(y\mid \mathbf{x}), \quad r_i = P_{i}(Y_{i+1} \mid \mathbf{X}_{i+1}),
$$
it is straightforward to show the martingale. Conditional on $\mathbf{X}_{i+1}  = \mathbf{x}_{i+1}$, we have that
$r_i \sim \mathcal{U}[0,1]$, so we can write
$$
E\left[ c_{\rho_y}\left(q_i,r_i\right)  \mid y_{1:i},\mathbf{x}_{1:i+1} \right]  = \int_{0}^1 c_{\rho_y}\left(q_i,r\right) \, dr = 1.
$$
So we have
$$
E\left[ p_{i+1}(y \mid \mathbf{x}) \mid y_{1:i},\mathbf{x}_{1:i+1} \right]  = p_i(y \mid \mathbf{x})
$$
and from the tower rule
$$
E\left[ p_{i+1}(y \mid \mathbf{x}) \mid y_{1:i},\mathbf{x}_{1:i} \right]  = p_i(y \mid \mathbf{x})
$$
almost surely for each $\mathbf{x} \in \mathbb{R}^d$. The martingale holds irrespective of the distribution of $\mathbf{X}_{i+1}$.
\subsubsection{Conditional regression with dependent stick-breaking}\label{Appendix:copula_regression}
We now derive the regression copula update inspired by the dependent DP. Consider the general covariate-dependent stick-breaking mixture model 
\begin{equation}\label{SM_eq:DDP_mixture_location}
\begin{aligned}
f_{G_\mathbf{x}}(\mathbf{y}) = \int \mathcal{N}(y \mid \theta,1) \, dG_{\mathbf{x}}({\theta}), \quad
G_\mathbf{x} =\sum_{k=1}^\infty w_k(\mathbf{x})\,\delta_{\theta^*_k}.
\end{aligned}
\end{equation}
For the weights, we elicit the stick-breaking prior $w_k(\mathbf{x}) = v_k(\mathbf{x}) \prod_{j< k} \{1-v_j(\mathbf{x})\}$ where $v_k(\mathbf{x})$ is a stochastic process on $\mathcal{X}$ taking values in $[0,1]$, and is independent across $k$. For the atoms, we assume they are independently drawn from a normal distribution,
$$
\theta_k^* \iid G_0,\quad G_0 = \mathcal{N}(\theta \mid 0, \tau^{-1}).
$$
Once again, we want to compute
$$
\frac{E\left[f_{G_\mathbf{x}}({y}) \, f_{G_{\mathbf{x}_1}}({y}_1) \right]}{p_0(y \mid \mathbf{x}) \, p_0(y_1 \mid \mathbf{x}_1)}\cdot
$$
Following the stick-breaking argument as in Section \ref{Appendix:multivariate}, we can write the numerator as
$$
\left\{1-\alpha_1(\mathbf{x},\mathbf{x}')\right\}E\left[ K({y} \mid \theta^*) \right] E\left[ K({y}_1 \mid \theta^*) \right] + \alpha_1(\mathbf{x},\mathbf{x}') E\left[ K({y} \mid \theta^*) \, K({y}_1 \mid \theta^*) \right] 
$$
where we write
$$K(y \mid \theta^*) = \mathcal{N}(y \mid \theta^*,1), \quad \theta^* \sim G_0,$$ 
and
$$
\alpha_1(\mathbf{x},\mathbf{x}') = \sum_{k=1}^\infty E\left[ w_k(\mathbf{x})w_k(\mathbf{x}') \right].
$$
As before, we have
$$
\frac{E\left[ K({y} \mid \theta^*) \, K({y}_1 \mid \theta^*) \right]}{p_0(y \mid \mathbf{x}) \, p_0(y_1 \mid \mathbf{x}_1)} = c_{\rho_y}\left\{ P_0(y \mid \mathbf{x}), P_0(y_1 \mid \mathbf{x}_1)\right\}
$$
where $\rho_y = 1/(1+\tau)$. We thus have the copula density as a mixture of the independent and Gaussian copula density. This then implies the first update step of the predictive takes the form
\begin{equation}
\begin{aligned}
p_{1}(y \mid \mathbf{x}) = \left[1-\alpha_1(\mathbf{x},\mathbf{x}_1)+ \alpha_1(\mathbf{x},\mathbf{x}_1) \, c_{\rho}\left\{P_0(y \mid \mathbf{x}),P_0(y_1 \mid \mathbf{x}_1)\right\}\right]\,  p_0(y\mid \mathbf{x}).
\end{aligned}
\end{equation}

\subsection{Classification}\label{Appendix:copula_classification}
\subsubsection{Beta-Bernoulli copula update}
In this section, we derive the copula update for the beta-Bernoulli model. The Bernoulli likelihood with beta prior is a special case of the DP update, where $y \in \{0,1\}$. 
We will use the update \eqref{eq:DP_countable} which simplifies drastically for the binary case.

For $y = y_1 = 0$, we have that $d_\rho(y,y_1) = D_\rho\{P_0(y), P_0(y_1)\} $, which directly gives us
$$
d_\rho(y, y_1) = 1-\rho + \rho \, \frac{p_0(y) \wedge p_0(y_1)}{p_0(y) \, p_0(y_1)} \quad \text{if } y= y_1 = 0.
$$
For $y= 1, y_1 = 0$ then, we have that any terms in \eqref{eq:DP_countable} with $y_1 -1$ are 0, giving us
\begin{equation*}
\begin{aligned}
d_\rho(y,y_1)\, p_0(y)\, p_0(y_1)&= D_\rho\{P_0(y),P_0(y_1)\} - D_\rho\{P_0(y-1),P_0(y_1)\} \\
&= p_0(y_1) - (1-\rho)\, p_0(1-y) \, p_0(y_1) - \rho \{p_0(1-y) \wedge p_0(y_1)\} \\
&= (1-\rho)\,p_0(y ) \,p_0(y_1) + \rho\left(p_0(y_1) -[\{1-p_0(1-y)\} \wedge p_0(y_1)] \right)\\
&=(1-\rho)\,p_0(y) \,p_0(y_1) + \rho\left(p_0(y) -[p_0(y)\wedge \{1-p_0(y_1)\}] \right)
\end{aligned}
\end{equation*}
where we have used
$$
p_0(y_1) -[\{1-p_0(y)\} \wedge p_0(y_1)] = p_0(y) -[p_0(y)\wedge \{1-p_0(y_1)\}].
$$
This gives us
$$
d_\rho(y, y_1) = 1-\rho + \rho \, \frac{p_0(y) -[p_0(y) \wedge \{1-p_0(y_1)\}]}{p_0(y) \, p_0(y_1)} \quad \text{if } y=1, y_1 = 0.
$$
Following the above derivations for the remaining two cases gives us
$$
d_\rho(y, y_1) = \begin{dcases}1-\rho + \rho \frac{p_0(y) \wedge p_0(y_1)}{p_0(y) \, p_0(y_1)} \quad &\text{if } y= y_1\\ 
 1-\rho + \rho \frac{p_0(y) -[p_0(y) \wedge \{1-p_0(y_1)\}]}{p_0(y)\, p_0(y_1)} \quad &\text{if } y\neq y_1\end{dcases}
$$
returning us the equation in the main paper if we plug in $q_i = p_{i}(y \mid \mathbf{x})$ and $r_i = p_{i}(y_{i+1} \mid \mathbf{x}_{i+1})$.

We now provide a quick check that the beta-Bernoulli update for classification indeed satisfies the martingale conditions. Let us write $q_i(1) = p_i(y = 1 \mid \mathbf{x})$, $q_i(0) = 1-q_i(1)$ and likewise $r_i(1) = p_i(y_{i+1} = 1 \mid \mathbf{x}_{i+1})$, $r_i(0) = 1-r_i(1)$. We first check that $q_{i+1}(1) + q_{i+1}(0) = 0$ given $q_i(1) + q_i(0) = 0$. For $y_{i+1} = 1$, we just need to check that the following term is equal to 1:
\begin{equation*}
\begin{aligned}
\frac{\{q_i(1) \wedge r_i(1)\} + q_i(0) - \{q_i(0) \wedge (1-r_i(1))\}}{r_i(1)}\cdot
\end{aligned}
\end{equation*}
If $q_i(0) < r_i(0)$, then $q_i(1) > r_i(1)$ so the numerator is $r_i(1) + q_i(0) - q_i(0) = r_i(1)$. Likewise if $q_i(0)> r_i(0)$, the numerator is $r_i(1) + q_i(0) - q_i(0) =r_i(1)$. The same applies for $y_{i+1} = 0$. Now to check the martingale condition. For $y = 1$, we want the below term to equal 1:
\begin{equation*}
\begin{aligned}
\frac{\{q_i(1) \wedge r_i(1)\} + q_i(1) - \{q_i(1) \wedge (1-r_i(0)\}}{q_i(1)} = \frac{q_i(1) +\{q_i(1) \wedge r_i(1)\}  -\{ q_i(1) \wedge r_i(1)\}  }{q_i(1)} =1.
\end{aligned}
\end{equation*}

\subsubsection{Probit copula update}
We can follow a similar derivation to the beta-Bernoulli copula update to derive an update for the conjugate probit-normal model. Here, the DPMM kernel is the normal CDF, $P(y = 1 \mid \theta) = \Phi(\theta)$, and the base measure is $G_0 = \mathcal{N}(\theta \mid 0,\tau^{-1})$. The update term instead takes the form 
\begin{equation}
\begin{aligned}
b_{\rho_y}(q_i,r_{i}) &= \begin{dcases} \frac{C_{\rho_y}(q_i,r_i)}{q_i\,  r_i} &\quad \text{if } y = y_{i+1} \\ \frac{ \{q_i - C_{\rho_y}(q_i,1-r_i)\}}{q_i \, r_i}&\quad \text{if } y  \neq y_{i+1} \end{dcases}
\end{aligned}
\end{equation}
where $q_i = p_i(y \mid \mathbf{x}), r_i = p_i(y_{i+1} \mid \mathbf{x}_{i+1})$, and $C$ is the bivariate Gaussian copula distribution function. Again, we have $\rho_y = 1/(1+\tau)$. Note the similarities to the beta-Bernoulli update, where instead of $C_{\rho_y}(u,v)$ we have $D_{\rho_y}(u,v)$.  However, this is computationally much more expensive as we need to approximate a bivariate normal integral.

\section{Practical considerations for copula methods}\label{Appendix:practicalcons}
\subsection{Implementation details}
For our copula methods, we begin by first computing $v_i^j$ for $i = \{0,\ldots,n-1\}, j = \{1,\ldots, d\}$ where $v_i^j = P_i(y_{i+1}^j \mid y_{i+1}^{1:j-1})$. In practice, this involves iterating the copula methods and computing the conditional CDF values through \eqref{eq:mv_DP_conditional} for the observed datapoints $\mathbf{y}_{1:n}$. For the just-in-time compilation, it is generally faster to work with arrays of fixed size, and so it is better to compute $P_{1:n}$ for all $\mathbf{y}_{1:n}$ and extract the needed $v_i^j$. Given $v_i^j$, we can compute $p_n$ at any test point $\mathbf{y}$, and predictively resample. 

\subsection{Computing quantiles and sampling}

For the univariate cases, we may be interested in computing quantiles for $P_n(y)$ or $P_n(y \mid \mathbf{x})$. For example, finding the median $\bar{y}$ such that $P_n(\bar{y}\mid \mathbf{x}) = 0.5$ may be of interest if we are interested in a point estimate of the regression function. Although we cannot compute $\bar{y}$ exactly, we can solve $P_n({y}\mid \mathbf{x}) =0.5$ for $y$ through efficient gradient-based optimization using automatic differentiation.  We may also be interested in sampling from $P_n(\mathbf{y})$ or $P_N(\mathbf{y})$, e.g. for computing Monte Carlo estimates of parameters. We can again utilize an optimization-based inverse transform sampling method:  we simulate $d$ independent uniform random variables $U^{1:d}$, and then find the quantiles by solving $u_n^j = U^j$ for $j = \{1,\ldots,d\}$ to obtain a random sample $\mathbf{Y}$. In practice, we can minimize the loss function,
$\mathbf{Y} = \argmin_{\mathbf{Y}} \sum_{j=1}^d \left\{U^j - P(Y^j \mid Y^{1:j-1}) \right\}^2$, which would give 0 loss at the sample.

\subsection{Optimization details}
As gradients with respect to $\rho$ are generally difficult to compute, we opt for automatic differentiation, which is available in JAX. For hyperparameter optimization, we use the SLSQP method in \texttt{scipy} \citep{Scipy2020}. However, it would be of interest to investigate the applications of stochastic gradient methods. For computing quantiles/sampling, we implement BFGS with Armijo backtracking line search directly in JAX so that the entire optimization procedure can be JIT-compiled.

\section{Experiments}\label{Appendix:experiments}
\subsection{Reproducibility of copula experiments}
Due to the non-determinism of GPU computations even with random seeds, the reported results for the copula methods in the main paper are timed on the GPU but exact numerical values and plots are computed on the CPU (with the appropriate \texttt{jaxlib} version). The results are unsurprisingly very similar on the two platforms, but the CPU allows for exact reproducibility with the provided code at {\href{https://github.com/edfong/MP}{\texttt{https://github.com/edfong/MP}}}. The baselines methods do not have this issue as they are timed and run on the CPU.

\subsection{Baseline details}\label{Appendix:baselines}
For the DPMM with MCMC examples, we use the full covariance multivariate normal kernel with conjugate priors from the \texttt{dirichletprocess} package \citep{Ross2018}, which implements Gibbs sampling from \cite{Neal2000}. The hyperprior on the concentration parameter is $a \sim \text{Gamma}(2,4)$. For the univariate case, the conjugate base measure is
$$
G_0({\mu},\sigma^2) = \mathcal{N}\left( {\mu} \mid {\mu}_0,  \frac{\sigma^2}{k_0}\right) \text{Inverse Gamma}(\sigma^2\mid \alpha,\beta).
$$
For the GMM examples, we set $\{\mu_0,k_0,\alpha,\beta\} = \{0,1,0.1,0.1\}$; we select $\alpha,\beta$ to match the smoothness of the posterior mean with the copula method. For the galaxy example (with unnormalized data), we follow the suggestions of \cite{West1991} and set $\{\mu_0,k_0,\alpha,\beta\} = \{20,0.07,2,1\}$, where again $k_0$ is set to a point estimate that matches the smoothness to the copula method.

For the multivariate DPMM, the base measure of the DP is the normal-Wishart distribution:
$$
G_0(\bm{\mu},\bm{\Lambda}) = \mathcal{N}( \bm{\mu} \mid \bm{\mu}_0, (k_0 \bm{\Lambda})^{-1})\, \text{Wi}(\bm{\Lambda} \mid \bm{\Lambda_0}, \nu)
$$
where $\bm{\Lambda}$ is the precision matrix of the kernel. For both the air quality and LIDAR example, we have $\{\bm{\mu}_0, k_0, \bm{\Lambda}_0,\nu \} = \{\mathbf{0}, d,I_d,d\}$ where $I_d$ is the identity matrix and $d = 2$m which is default in the package. For regression, we fit the DPMM to estimate the joint density $p(y,x)$ and compute the implied conditional $p(y \mid x)$. For MCMC, we sample $B = 2000$ posterior samples for the plots with a burn-in chain of length $2000$. For the timing, we have the same burn-in chain but only take $B = 1000$ to compare to the copula. Given a posterior sample of the cluster parameters $\theta_{1:n}$, we can draw an approximate posterior sample of $[G_\infty \mid \theta_{1:n}]$, which is conditionally independent from $y_{1:n}$, from the stick-breaking representation (see Key Property 5 of \citet{Ross2018}), from which we can compute a random sample of $p_\infty$ by mixing over the kernel with $G_\infty$.

The remainder methods are implemented in \texttt{sklearn}. For the DPMM with VI  (mean-field approximation), we use the diagonal covariance kernel, with default hyperparameters for the priors. For the variational approximation, we set the upper limit of clusters to $K = 30$, and initialize and optimize 100 times to avoid local minima.
For the KDE, the scalar bandwidth is set through 10-fold cross-validation with the log predictive density over a grid. For  Gaussian process regression, we use the RBF kernel with a single length scale, which is set by maximizing the marginal likelihood over 10 repeats. The same is done for Gaussian process classification with the logistic link function, which approximates the posterior with the Laplace approximation. For the linear models, we use Bayesian ridge regression and logistic regression with $L_2$ regularization with default values.

\subsection{Galaxy}
For the galaxy example, convergence to the martingale posterior is assessed in Figure \ref{fig:galaxy_conv}, where $5000$ forward samples is sufficient. 
\begin{figure}[!h]
    \centering
        \includegraphics[width=0.97\textwidth]{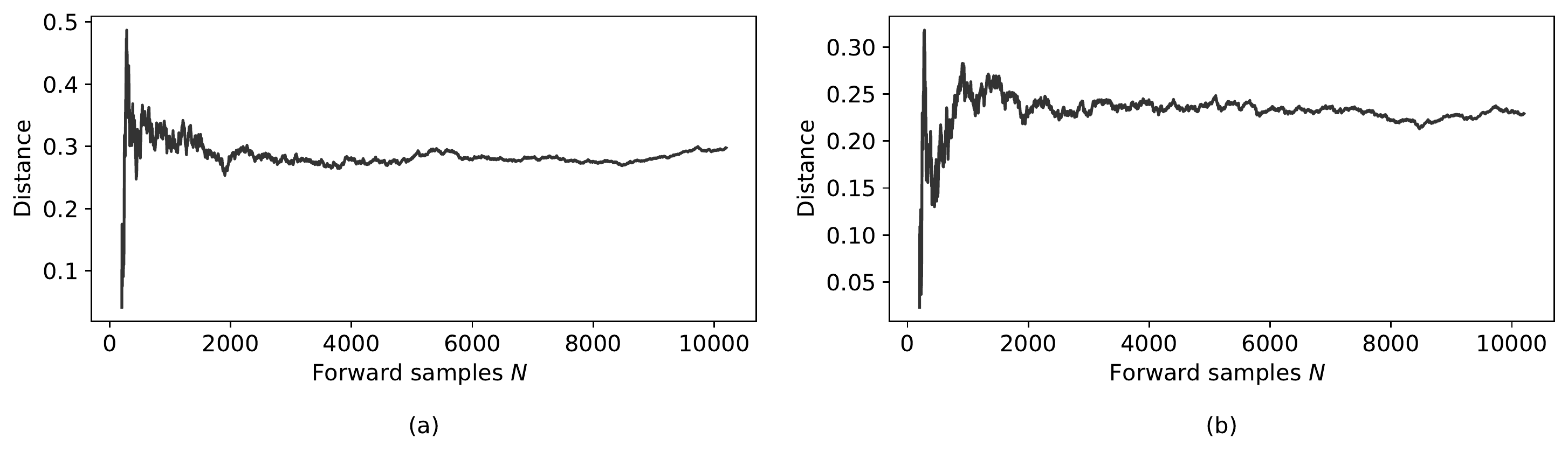}
    \caption{Estimated $L_1$ distance  (a) $\lVert p_N - p_{n} \rVert_1$  and (b) $\lVert P_N - P_n \rVert_1$ for a single forward sample} \label{fig:galaxy_conv}
\end{figure}
\subsection{Bivariate air quality}\label{Appendix:ozone}
In Figure \ref{fig:DPMM_ozone_dens} we show the posterior mean and standard deviation of the density obtained for the DPMM with MCMC, which is comparable to that of the martingale posterior from the copula method. 

\begin{figure}[!h]
    \centering
        \includegraphics[width=0.97\textwidth]{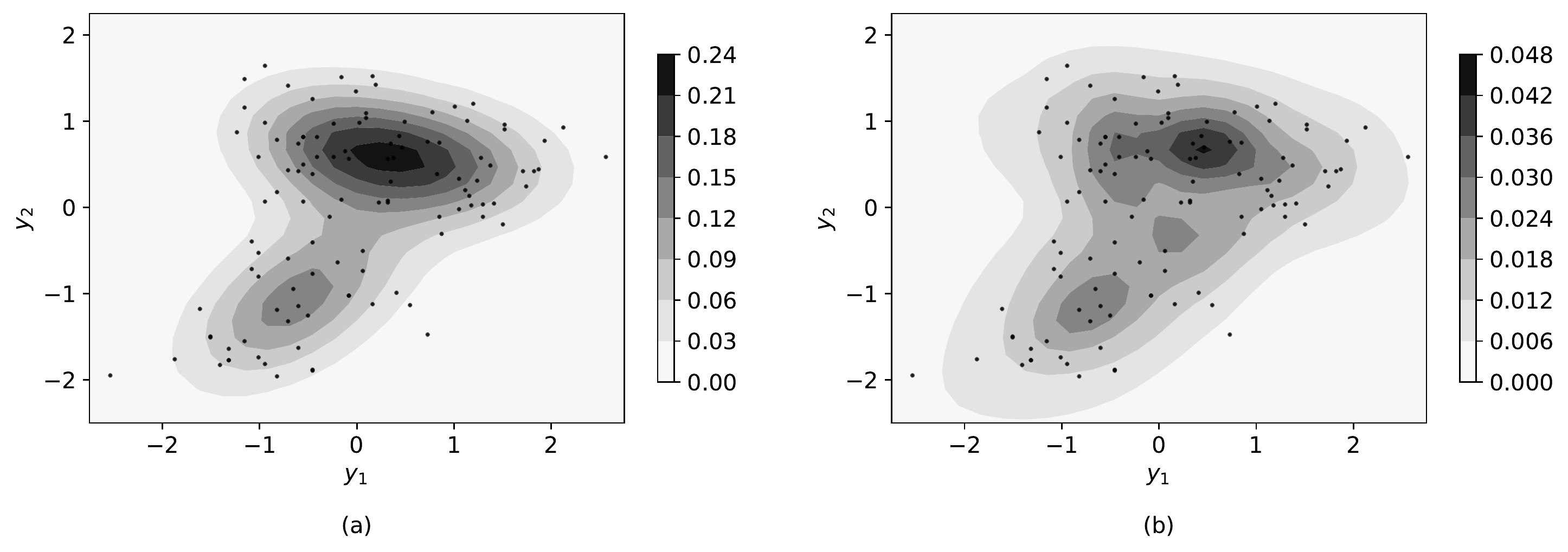}
    \caption{Posterior (a) mean and (b) standard deviation of density for DPMM, with data (\dot)} \label{fig:DPMM_ozone_dens}
\end{figure}

\subsection{Classification in simulated moon dataset}\label{Appendix:moon}
We demonstrate the conditional method of Section \ref{sec:conditreg} with the beta-Bernoulli copula update for classification on a non-linear classification example, which is particularly interesting as we can predictively resample $Y_{n+1:\infty}$ directly. We analyze the simulated moon dataset from \texttt{sklearn} \citep{Pedregosa2011} with bivariate covariates. Figure \ref{fig:moon_data} shows the decision boundary on the left, and $n=100$ simulated data points on the right with a fit GP using the logistic link and Laplace approximation. We fit the conditional copula method and get the bandwidths $\rho_y = 0.73, \rho_x = [0.92,0.74]$; optimizing, fitting and prediction on the $\mathbf{x}$-grid of size $25 \times 25$ required 2 seconds, versus around 1 second for the GP.

Predictive resampling is also possible here, where we draw $X_{n+1:N}$ with the Bayesian bootstrap and simulate $Y_{n+1:N}$ directly. Generating $B=1000$ samples of $Y_{n+1:N},X_{n+1:N}$ and computing each $p_N(y=1 \mid \mathbf{x})$ on the $25 \times 25$ grid took under $3$ seconds. The martingale posterior mean and standard deviations of $p_N(y = 1 \mid \mathbf{x})$ are shown in Figure \ref{fig:moon_fit}. The martingale posterior mean $p_n(y=1 \mid \mathbf{x})$ is similar to the GP, and we see that the uncertainty is higher in $\mathbf{x}$-regions where both classes are present. However, we must note that the uncertainty decreases to 0 as we move away from the data which is undesirable. This may be due to the Bayesian bootstrap resampling $\mathbf{x}_{n+1:\infty}$ from the observed $\mathbf{x}_{1:n}$, so $\mathbf{x}_N$ in the tails are never sampled. As a result, $\alpha(\mathbf{x},\mathbf{x}_N)$ remains small for $\mathbf{x}$ in the tails, so there is not much change in $p_N(y \mid \mathbf{x})$. This behaviour is also observed in the conditional copula regression case in Appendix \ref{Appendix:lidar}. As such, we recommend only looking at inlying $\mathbf{x}$ when using the Bayesian bootstrap for $\mathbf{x}_{n+1:\infty}$ with the conditional copula method. On the left of Figure \ref{fig:moon_conv}, we also plot the martingale posterior of $p_N(y= 1 \mid \mathbf{x} = [1,-0.8])$, which has mean $0.86$ but has relatively large uncertainty. On the right, we plot the absolute difference $|p_N(y= 1 \mid \mathbf{x}) - p_n(y= 1 \mid \mathbf{x})|$ as we carry out one forward sample, which converges relatively quickly with $N$. 

\begin{figure}[!h]
    \centering
        \includegraphics[width=0.92\textwidth]{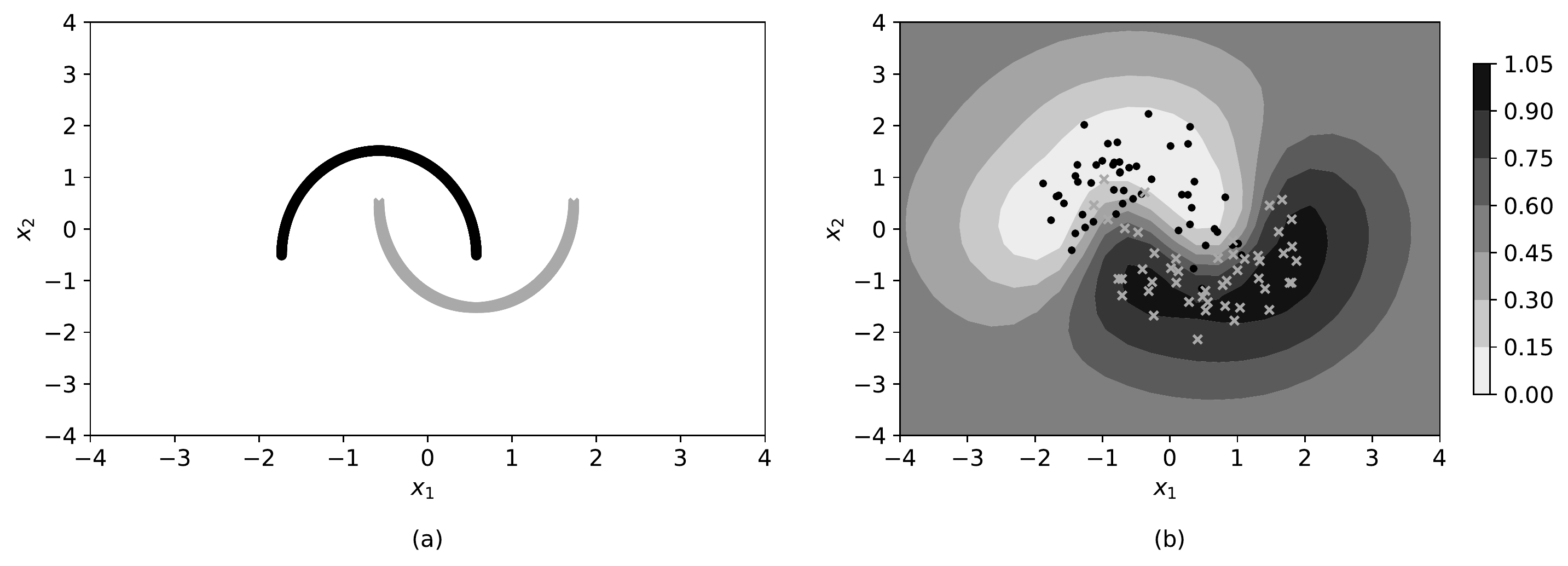}
    \caption{(a) Simulated data with no noise with $y=0$ (\dot) and $y = 1$ (\crossmid); (b) $n=100$ simulated data points with Gaussian noise (added to covariates) and $p_n(y = 1\mid \mathbf{x})$ for GP 
} \label{fig:moon_data}
\end{figure}

\begin{figure}[!h]
    \centering
        \includegraphics[width=0.92\textwidth]{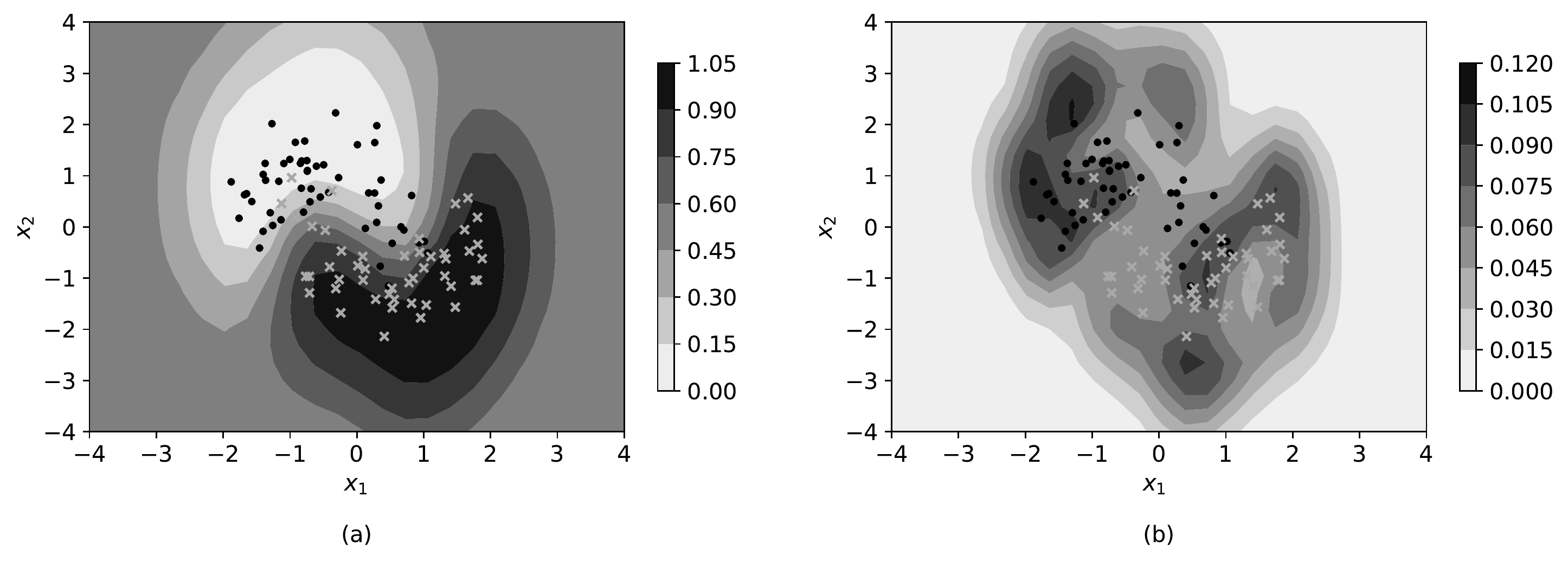}
    \caption{Posterior (a) mean  and (b) standard deviation  of $p_N(y = 1 \mid \mathbf{x})$ for conditional copula method with $y=0$ (\dot) and $y = 1$ (\crossmid)} 
    \label{fig:moon_fit} 
\end{figure}
\begin{figure}[!h]
    \centering
        \includegraphics[width=0.92\textwidth]{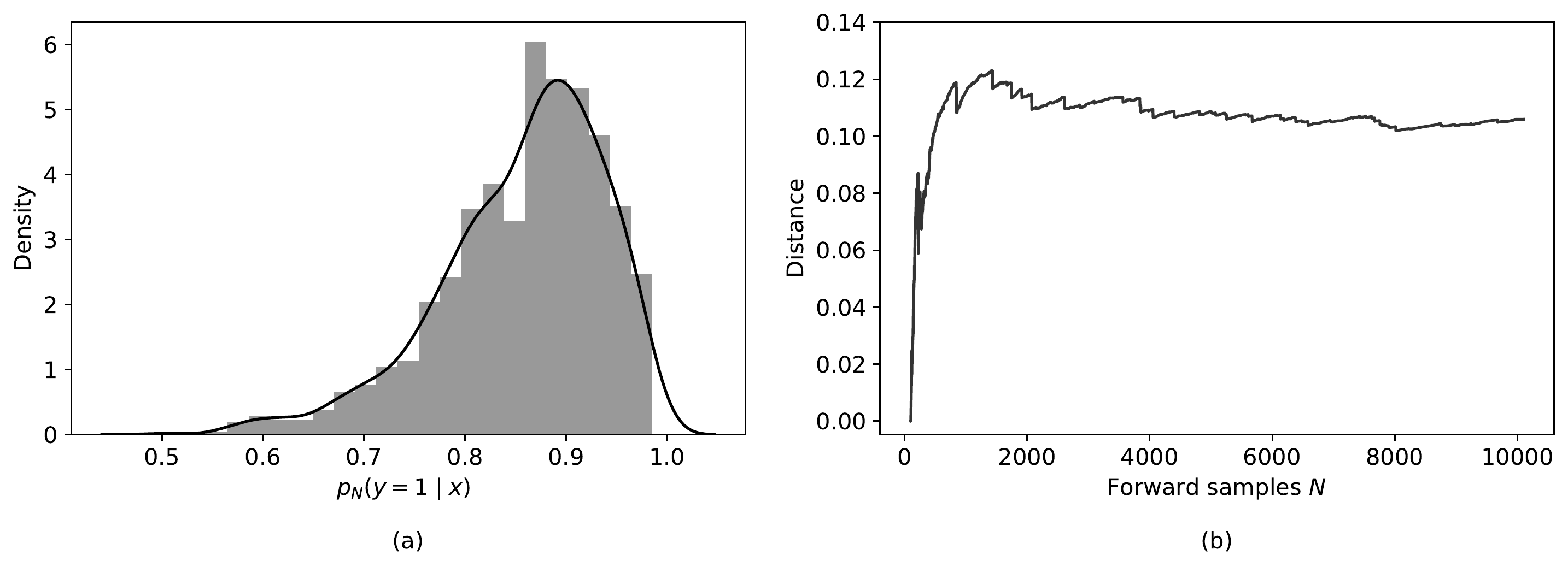}
    \caption{(a) Posterior samples of $p_{N}(y=1 \mid \mathbf{x})$; (b) Convergence plot $|p_{N}(y=1 \mid \mathbf{x}) - p_{n}(y=1 \mid \mathbf{x})|$ for $\mathbf{x} = [1,-0.8]$  } \label{fig:moon_conv}
\end{figure}

\newpage

\subsection{LIDAR}\label{Appendix:lidar}
In Figure \ref{fig:lidar_cgp}, we see the predictive means and 95\% central intervals for the conditional copula method on the left and the GP on the right. The conditional copula method still performs well, with a slight bias towards the end of range of $x$, in a similar way to the Nadaraya-Watson estimator. The GP deals well with the nonlinearity but much more poorly with the heteroscedasticity.

In Figures \ref{fig:lidar_jc_x0}, \ref{fig:lidar_jc_x-3} we see the difference in the martingale posteriors of $p_N(y \mid x= 0)$ for the joint and conditional copula method. For $x=0$ which lies within the data, the uncertainty for the joint and conditional methods are similar. However, for $x=-3$ which is far from the data, we see that the joint method has a high amount of uncertainty as expected, but the conditional method does not have any uncertainty. We believe this occurs due to the discrete nature of the Bayesian bootstrap when resampling $x_{n+1:\infty}$ for the conditional method as discussed in Section \ref{Appendix:moon}. This issue can be mitigated by resampling using a smooth density (e.g. a KDE) for $x_{n+1:\infty}$, but still occurs for suitably distant $x$. It seems that the choice of resampling distribution of $x_{n+1:\infty}$ affects the uncertainty of outlying $x$, but we have found much less sensitivity to this choice for inlying $x$. Interestingly, the joint method does not have this issue irrespective of the distance of $x$ from the dataset, so we recommend the joint method when the $x$ of interest is outlying. We leave a detailed investigation of this for future work.

We can also see the difference in $p_n(y \mid x)$ between the joint and conditional method by comparing Figures \ref{fig:lidar_cgp} and \ref{fig:lidar_jdpmm}. Although the joint method is not skewed towards the end of the range of $x$, the prequential log-likelihood is actually higher for the conditional method. The conditional method is only slightly faster in this example, but for higher dimensional covariates, the computational gain of not having to estimate $P_n(\mathbf{x})$ is much higher. For the median plot in Figure \ref{fig:lidar_conv}(a) in the main paper, we compute the $y,x$ grid of size $40 \times 40$ using the 95\% credible bands as the limits for $y$. The median is then computed numerically from the grid where $P_N(y \mid x) = 0.5$.
\begin{figure}[!h]
    \centering
        \includegraphics[width=0.92\textwidth]{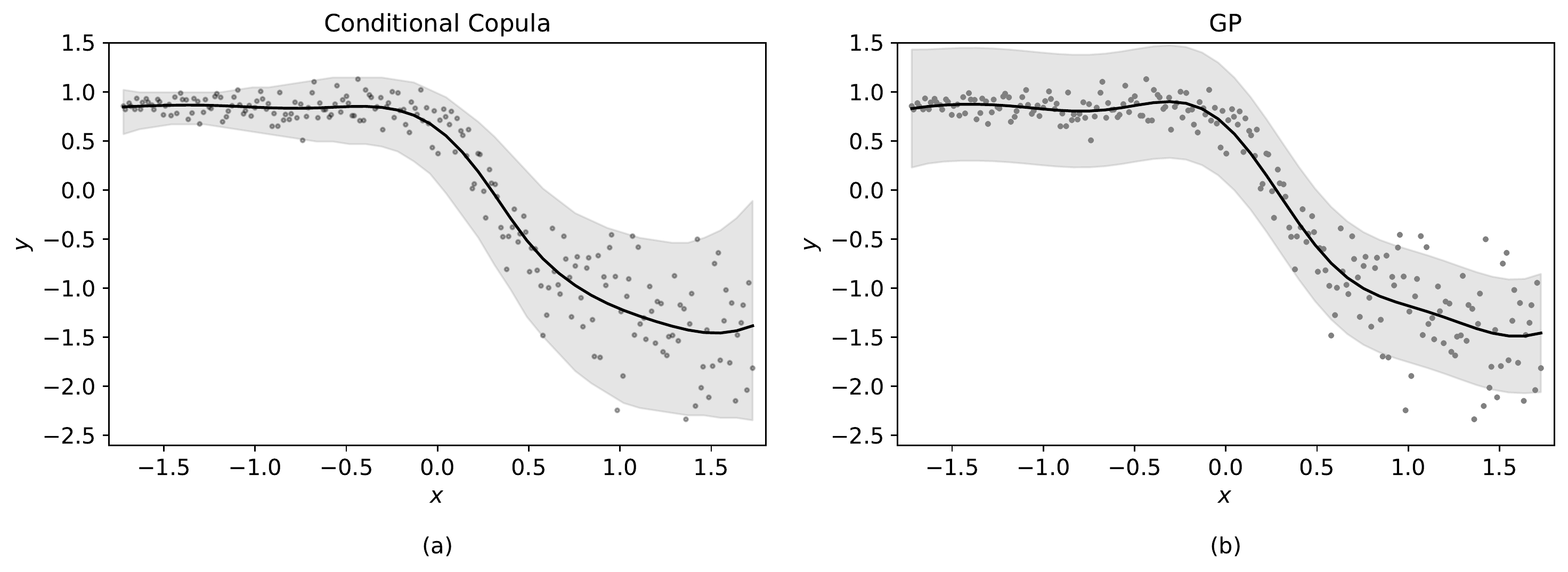}
    \caption{ $p_n(y \mid x)$ (\full) with 95\% predictive interval (\sqrlow) for the (a) conditional copula method and (b) GP, with data (\dotmid)} \label{fig:lidar_cgp}
\end{figure}
\begin{figure}[!h]
    \centering
        \includegraphics[width=0.92\textwidth]{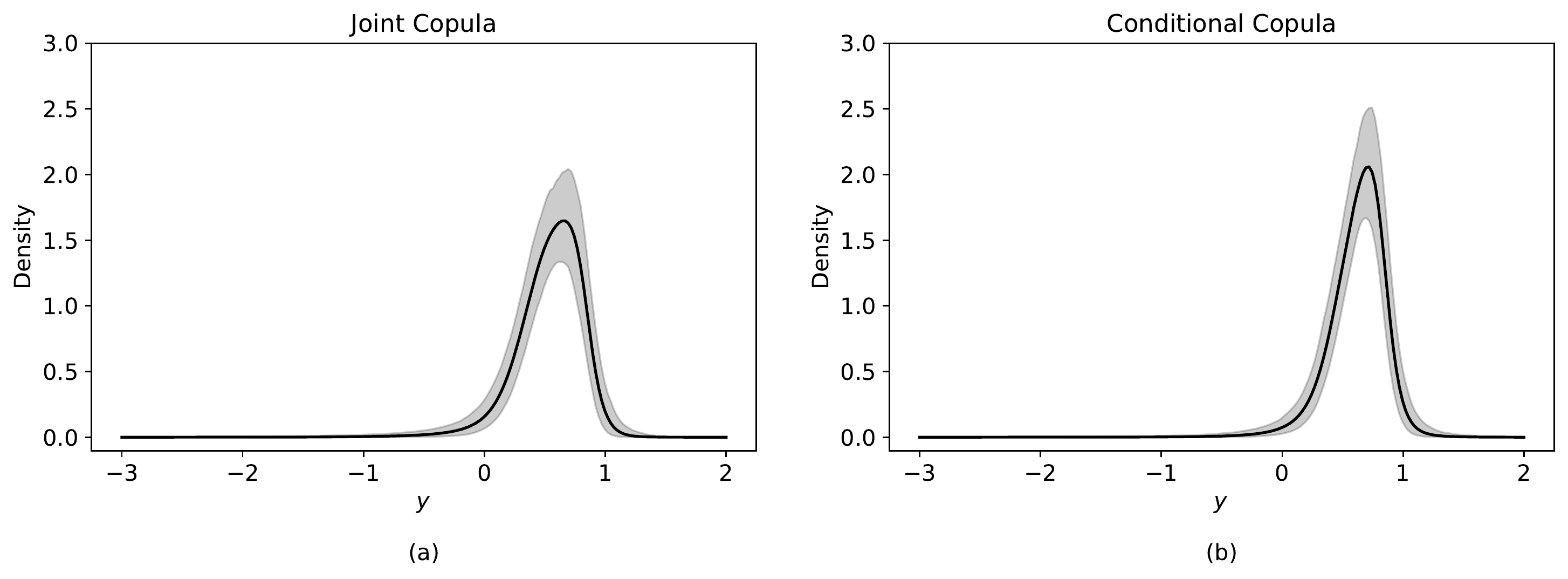}
    \caption{Posterior mean (\full) and 95\% credible interval (\sqrlow) of $p_N(y \mid x=0)$ for the (a) joint copula method  and (b) conditional copula method } \label{fig:lidar_jc_x0}
\end{figure}

\begin{figure}[!h]
    \centering
        \includegraphics[width=0.92\textwidth]{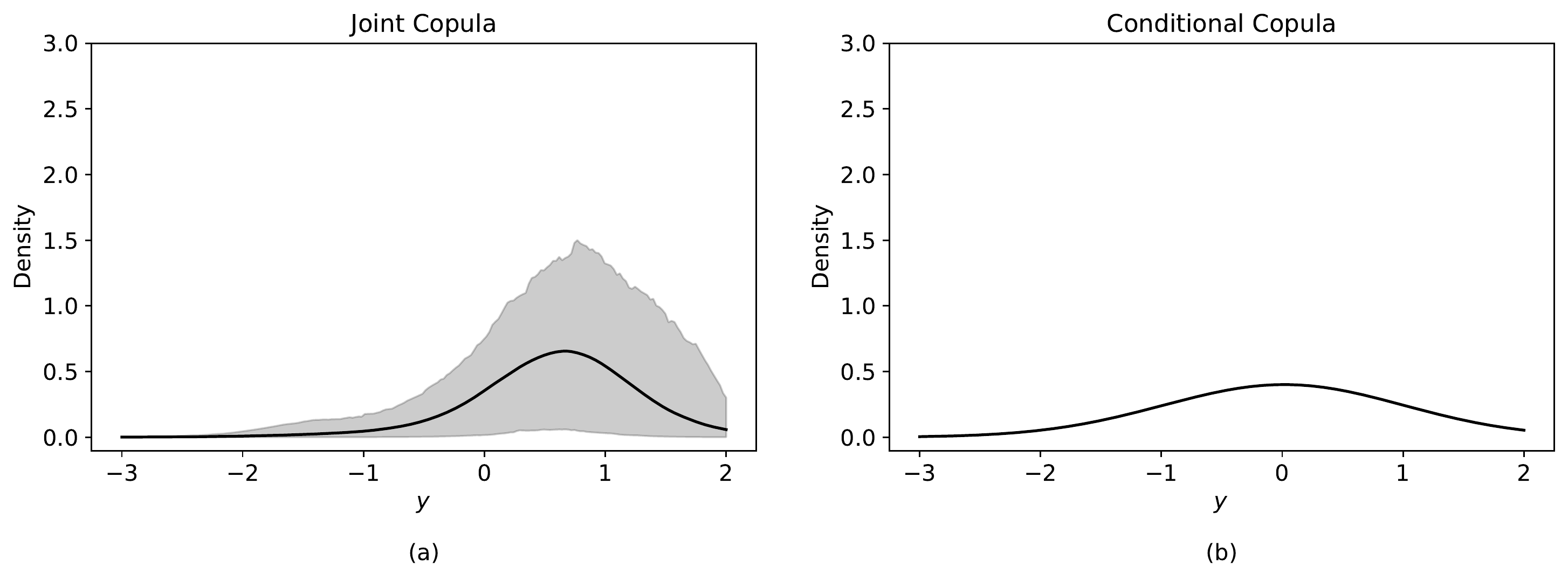}
    \caption{Posterior mean (\full) and 95\% credible interval (\sqrlow) of $p_N(y \mid x=-3)$ for the (a) joint copula method  and (b) conditional copula method } \label{fig:lidar_jc_x-3}
\end{figure}

 \newpage
\subsection{High-dimensional GMM}\label{Appendix:highd_gmm}
In this section, we consider a simulated example to demonstrate the degradation in performance with dimensionality for the DPMM with VI (with diagonal covariance matrix), which the copula method is robust to. We simulate $n = 100, n_{\mathrm{test}} = 1000$ data points from
$$
f_0(\mathbf{y}) = 0.5\, \mathcal{N}(\mathbf{y} \mid \mu^{1}_d, I_d)  + 0.5\, \mathcal{N}(\mathbf{y} \mid \mu^{2}_d, I_d) 
$$
where $\mu^{1}_d = \begin{bmatrix}-1,\ldots,-1 \end{bmatrix}^\T$ and  $\mu^{1}_d = \begin{bmatrix}2,\ldots,2 \end{bmatrix}^\T$ are both $d$-vectors, and $I_d$ is the $d\times d$ identity matrix. The DPMM is thus well-specified in this example. As before, we normalize the data before fitting the methods.

Below is a plot of the average test log-likelihoods as $d$ increases for each method. The DPMM with VI degrades steeply with dimensionality, likely due to the difficulty of the variational optimization. The KDE performs quite well in this simple example. 
\begin{figure}[!h]
    \centering
        \includegraphics[width=0.5\textwidth]{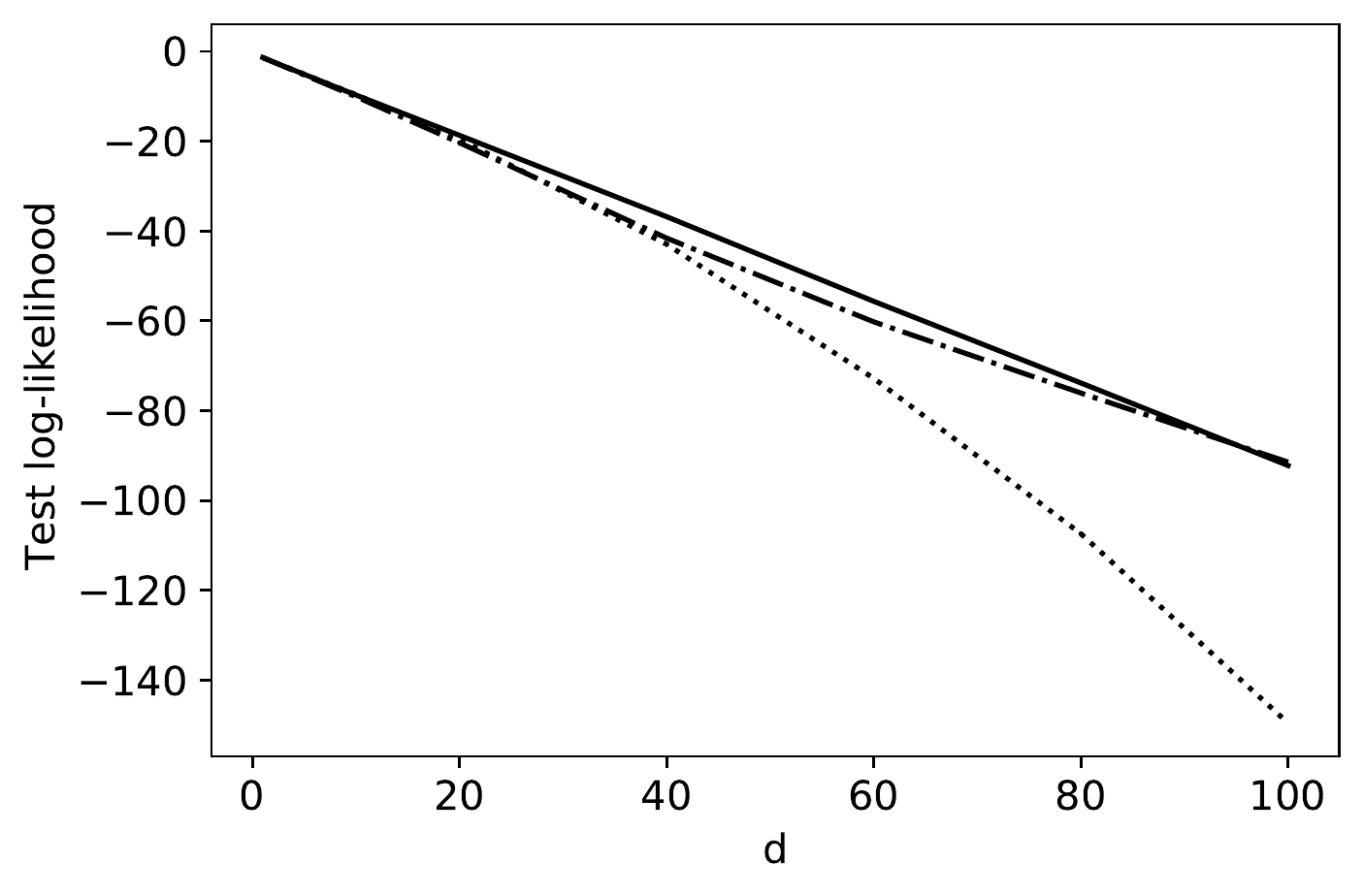}
    \caption{Average test log-likelihood with dimensionality for the copula method (\full), DPMM with VI (\dotted) and KDE(\dashdotted)} \label{fig:highd_GMM}
\end{figure}
\end{appendices}
\newpage


\end{document}